\documentclass[11pt]{article}
\usepackage{color}
\usepackage[colorlinks=true, allcolors=blue,backref=page]{hyperref}
\usepackage{amsmath, amssymb, amsthm, thmtools}
\usepackage{mathrsfs}
\usepackage{mathtools}
\usepackage[noabbrev,capitalize,nameinlink]{cleveref}
\usepackage{fullpage}
\usepackage[noadjust]{cite}
\usepackage{graphics}
\usepackage{pifont}
\usepackage{tikz}
\usetikzlibrary{arrows.meta,patterns}
\usepackage{pgfplots}
\usepgfplotslibrary{groupplots}
\pgfplotsset{compat=1.18}
\usetikzlibrary{calc}

\usepackage{bbm}
\usepackage{float} 
\usepackage[T1]{fontenc}

\usetikzlibrary{arrows.meta}

\usepackage{environ}
\usepackage{framed}
\usepackage{url}
\usepackage[linesnumbered,ruled,vlined]{algorithm2e}
\usepackage[noend]{algpseudocode}
\usepackage[labelfont=bf]{caption}
\usepackage{cite}
\usepackage{framed}
\usepackage[framemethod=tikz]{mdframed}
\usepackage{appendix}
\usepackage{graphicx}
\graphicspath{{internal/}}
\usepackage{comment}
\usepackage{booktabs}
\usepackage{tabularx}
\usepackage[textsize=tiny,textwidth=2cm]{todonotes}
\usepackage{tcolorbox}
\usepackage{enumerate}
\allowdisplaybreaks[1]
\usepackage{enumerate}
\usepackage{stmaryrd}

\usepackage{wasysym}
\usepackage[margin=1in]{geometry}
\usepackage{braket}
\usepackage{todonotes}

\usepackage[shortlabels]{enumitem}
\crefformat{enumi}{#2#1#3}
\crefrangeformat{enumi}{#3#1#4 to~#5#2#6}
\crefmultiformat{enumi}{#2#1#3}
{ and~#2#1#3}{, #2#1#3}{ and~#2#1#3}

\DeclareSymbolFont{symbolsC}{U}{pxsyc}{m}{n}
\SetSymbolFont{symbolsC}{bold}{U}{pxsyc}{bx}{n}
\DeclareFontSubstitution{U}{pxsyc}{m}{n}
\DeclareMathSymbol{\medcircle}{\mathbin}{symbolsC}{7}

\crefname{equation}{}{} %
\AtBeginEnvironment{appendices}{\crefalias{section}{appendix}} %

\usepackage[color,final]{showkeys} %

\colorlet{refkey}{orange!20}
\colorlet{labelkey}{blue!30}

\numberwithin{equation}{section}
\newtheorem{theorem}{Theorem}[section]
\newtheorem{proposition}[theorem]{Proposition}
\newtheorem{lemma}[theorem]{Lemma}

\newtheorem{corollary}[theorem]{Corollary}

\newtheorem*{question*}{Question}
\newtheorem{fact}[theorem]{Fact}

\theoremstyle{definition}
\newtheorem{definition}[theorem]{Definition}

\newtheorem{question}[theorem]{Question}
\newtheorem*{definition*}{Definition}
\newtheorem*{theorem*}{Theorem}
\newtheorem{example}[theorem]{Example}

\newtheorem{remark}[theorem]{Remark}

\newcommand{\mb}{\mathbb}
\newcommand{\mbf}{\mathbf}

\newcommand{\mc}{\mathcal}

\newcommand{\ol}{\overline}
\newcommand{\on}{\operatorname}
\newcommand{\wh}{\widehat}

\newcommand{\F}{\mathbb{F}}
\newcommand{\C}{\mathbb{C}}
\newcommand{\R}{\mathbb{R}}
\newcommand{\N}{\mathbb{N}}
\newcommand{\E}{\mathbb{E}}
\renewcommand{\P}{\mathbb{P}}
\newcommand{\FO}{\mathsf{F}_{\on{OPI}}}
\newcommand{\FL}{\mathsf{F}_{\on{LLR}}}

\newcommand{\RK}{\mathsf{R}_{\on{KK}}}
\newcommand{\SCL}{\mathsf{SCL}}

\newcommand{\TV}{d_{\mathrm{TV}}}
\newcommand{\supp}{\operatorname{supp}}

\newcommand{\poly}{\operatorname{poly}}

\newcommand{\one}{\mathbf{1}}
\newcommand{\eps}{\varepsilon}
\newcommand{\abs}[1]{\left|#1\right|}
\newcommand{\norm}[1]{\left\lVert#1\right\rVert}
\let\originalleft\left
\let\originalright\right
\renewcommand{\left}{\mathopen{}\mathclose\bgroup\originalleft}
\renewcommand{\right}{\aftergroup\egroup\originalright}

\allowdisplaybreaks

\title{Improved Local Leakage Resilience of Shamir Secret Sharing and Worst-Case Optimal Polynomial Intersection}
 \author{Yihang Sun\thanks{Stanford University.  \texttt{kimisun@stanford.edu}} \and Mary Wootters\thanks{Stanford University. \texttt{marykw@stanford.edu}}}
\date{}
\begin{document}
\maketitle
\begin{abstract}
We study two problems: Local Leakage Resilience (LLR) for Shamir secret sharing, and worst-case Optimal Polynomial Intersection (OPI).  Both problems concern polynomials $Q(X)$ of degree less than $k$, over a prime-order finite field $\mathbb{F}_p$.  In LLR for Shamir secret sharing, one asks how much one can learn about $Q(0)$ given a few bits leaked from each of $Q(\alpha_1), \ldots, Q(\alpha_n)$,  
for distinct non-zero evaluation points $\alpha_i \in \mathbb{F}_p^\times$.  In OPI, one is given (worst-case) input lists $S_1, \ldots, S_n \subset \mathbb{F}_p$, and wants to find a polynomial $Q(X)$ of degree less than $k$ so that $Q(\alpha_i) \in S_i$ for as many $i$ as possible. 
Leveraging recent connections between these two problems due to (Sun, Wootters 2026), we improve the state-of-the-art for both problems.  

For LLR, we show that there is some constant $\delta > 0$ so that, as long as $R := k/n \geq 1/2 - \delta$, Shamir secret-sharing is one-bit locally leakage resilient (meaning that one can learn only a negligible amount about $Q(0)$).  This is the first result to break the so-called ``one-half barrier'' for LLR for arbitrary evaluation points, and improves over the previous best known result, requiring $R \geq 0.668$ (Kasser, 2025).  

For OPI, we give a quantum algorithm that finds a polynomial $Q(X)$ that agrees with at least a $\mathsf{SCL}_\rho(R)-\varepsilon$ fraction of the lists, for every fixed $\varepsilon>0$, where $\mathsf{SCL}_\rho$ is the \emph{semicircle law} of (Jordan et al., 2025).  This improves previous algorithmic (and existential) results of (Jo, 2026) and (Horinaga, Yamakawa, 2026). 
We also give further improved existential results for all OPI instances; and when the evaluation points $\alpha_i$ are chosen at random, for every rate $R > 0$ we show that, with high probability, for worst-case $S_i$ there exists \emph{perfect} solution.
Finally, we adapt the hardness result of (Yamakawa, Zhandry, 2024) to apply to OPI (rather than a folded version); over large fields, this gives an unconditional separation between the quantum and classical hardness of OPI relative to a membership oracle.

Our techniques develop a common Fourier-analytic framework that reduces both OPI and LLR to \emph{list-decoding} and \emph{list-recovery} of the relevant dual code. This allows us to leverage recent breakthroughs in list-decoding and list-recovery of Reed--Solomon codes (Brakensiek et al., 2026; Jeronimo, 2026; Dao, Kominers, Thaler, 2026) to achieve our improved combinatorial and quantum-algorithmic results.
\end{abstract}
\thispagestyle{empty}

\newpage
\tableofcontents
\thispagestyle{empty}

\newpage
\clearpage
\pagenumbering{arabic} 
\pagestyle{plain}     

\section{Introduction}
\label{sec:intro}

In this paper, we consider two problems: \emph{Local Leakage Resilience} (LLR) of Shamir secret sharing; and \emph{Optimal Polynomial Intersection} (OPI).  While seemingly unrelated at first, these two problems were connected by \cite{us}.  Leveraging this connection, we obtain improvements on both problems.  Our main results, described in more detail below, include improved rate thresholds for LLR for Shamir, past the so-called ``rate-$1/2$ barrier''; and improved results (both existential and algorithmic) for OPI. %

\subsection{Local Leakage Resilience}
Local Leakage Resilience is a property of a secret sharing scheme.  We focus on \emph{Shamir secret sharing.}  In Shamir secret sharing, a secret $s \in \F_p$ is shared among $n$ parties as follows: We choose a uniformly random polynomial $P$ of degree at most $k$ so that $P(0) = s$.  Then we share $P(\alpha_i)$ with party $i$, where $\alpha_1, \ldots, \alpha_n$ are non-zero, distinct evaluation points.  This scheme has the property that any $k+1$ of the parties together can recover the secret $s$, while any $k$ of them together learn nothing about it. We parameterize a sharing polynomial as $P(X)=s+XQ(X)$ with $\deg Q<k$. 

Informally, we say that a secret sharing scheme is ($b$-bit) \emph{Local Leakage Resilient} (LLR) if, when each of the parties may leak up to $b$ bits of their shares, the secret is still secure.  Formally, we have the following definition.

\begin{definition}[Local Leakage Resilience for Shamir's Secret-Sharing Scheme]
\label{def:llr}
Fix nonzero evaluation points $\alpha_1, \ldots, \alpha_n \in \F_p^{\times}$ and let $C \subset \F_p^n$ be the linear code\footnote{A \emph{linear code} is just an $\F_p$-linear subspace of $\F_p^n$.}
\begin{equation}
 \label{eq:secret-zero-code}
 C:= \{(\alpha_i Q(\alpha_i))_{i=1}^n \,:\, \deg(Q) < k \}.
\end{equation}
 For $v \in \F_p^n$, let $L_v$ denote the distribution of the \emph{transcript} $(\ell_1(c_1), \ldots, \ell_n(c_n))$, when $c$ is chosen uniformly at random from the coset $v + C$.
 We say that (the secret-sharing scheme defined by) $C$ is \emph{$b$-bit Local Leakage Resilient (LLR)} if, for any \emph{leakage maps} $\ell_i:\F_p \to \{0,1\}^b$, and for any pair of secrets $s,s' \in \F_p$, 
 \[ \TV(L_{s\one}, L_{s'\one}) \leq e^{-\Omega(n)}.\]
\end{definition}
While LLR makes sense for any secret-sharing scheme, we specialize to Shamir's scheme.  LLR for Shamir's scheme has received special attention lately, first because Shamir's scheme is a fundamental secret sharing scheme, and second because it is known~\cite{gw,tyb18} that over \emph{extension fields}, Shamir's scheme is \emph{not} locally leakage resilient.  Surprisingly, \cite{bdir_conference} (who also first explicitly introduced the notion of LLR) showed that over \emph{prime} fields, in fact Shamir's scheme is $1$-bit LLR provided that $R := k/n$ satisfies $R \geq 0.92$.  This state of affairs leads to the following question:
\begin{question}\label{q:llr}
What is the smallest rate $R^*$ so that, for all $R > R^*$, Shamir's scheme over a prime field $\F_p$ with $k/n = R$ is one-bit LLR?
\end{question}
 After a long line of work~\cite{bdir, mpsw21,mnpw22,kk23,n24,k24}, the best answer to \cref{q:llr} prior to this paper was due to \cite{k24}, and states that one can take $R^* \leq 0.668$.  Moreover, the work \cite{n24} identified a natural barrier at $R^* = 1/2$ (see \cref{sec:related_LLR}), and asked if it was possible to obtain $R^* < 1/2$. 

\paragraph{Our results for local leakage resilience.} We improve the state-of-the art on \cref{q:llr} for arbitrary evaluation points, and in particular we break the ``1/2-barrier'' for one-bit LLR.  More precisely, we prove the following theorem.

\begin{restatable}[Local Leakage Resilience of Shamir Schemes]{theorem}{mainllr}
\label{thm:main-llr}
Every Shamir secret sharing scheme with rate $R$ over a prime field is
\begin{enumerate}
    \item $b$-bit locally leakage resilient for every fixed $b$ and every fixed $R>1/2$;
    \item one-bit locally leakage resilient whenever
    $R>1/2-\delta$, for some absolute constant $\delta >0$.
\end{enumerate}
\end{restatable}
This improves on the best previous result of \cite{k24} that required $R > 0.668$, and is the first result to break the $1/2$ barrier, for $1$-bit worst-case leakage functions.  It is also the first result to even approach $R = 1/2$ for any $b \geq 1$ for $b$-bit worst-case leakage functions.

\begin{remark}
    We emphasize that \cref{thm:main-llr} applies for \emph{arbitrary} leakage functions (in particular, worst-case and not necessarily ``balanced,'' meaning that $\ell_i^{-1}(1)$ may have any size between $0$ and $p$).  However, it still improves on the best known result in the balanced case (when $|\ell_i^{-1}(1)| \sim p/2$ for all $i$).  This result was due to \cite{kk23}, and established $R^* \leq 0.58$ in the balanced case.
    In fact, as part of our proof, we show the balanced case is in fact the \emph{hardest}, in the sense that the sufficient rate $R^*$ threshold in the balanced case hold for arbitrary biases.
\end{remark}

\begin{remark}[Independent and concurrent work on LLR]\label{rem:concurrent}
In independent and concurrent work, Nguyen~\cite{n26} has obtained similar results to \cref{thm:main-llr}.  
In another independent work, \cite{HMSY26} have shown that the answer to \cref{q:llr} is $R^* = 0$ for \emph{most} choices evaluation points $\alpha_1, \ldots, \alpha_n$.   That is, when the $\alpha_i$ are chosen randomly, with high probability, \cite{HMSY26} shows that Shamir's scheme with any rate constant rate $R>0$ is one-bit LLR.  In this work, we focus on \emph{arbitrary} evaluation points for LLR, although as we will see later, in this version of the paper (which has been updated to reflect the works \cite{n26,HMSY26}), we are able to adapt the techniques of \cite{HMSY26} to obtain a similar result for Optimal Polynomial Intersection with random evaluation points, discussed below.  We discuss both \cite{n26,HMSY26} more in \cref{sec:related_LLR}.
\end{remark}

\subsection{Optimal Polynomial Intersection}

The second problem we study is \emph{Optimal Polynomial Intersection} (OPI).  OPI is a discrete optimization problem about low-degree polynomials over finite fields.  Given \emph{input lists} $S_1, \ldots, S_n \subset \F_p$ and distinct evaluation points $\alpha_1, \ldots, \alpha_n \in \F_p$, the goal is to find a polynomial $Q(X) \in \F_p[X]$ of degree less than $k$ so that $Q(\alpha_i) \in S_i$ for as many $i$ as possible.  Formally, we have the following definition.

\begin{definition}[Optimal Polynomial Intersection]
\label{def:opi}
Fix distinct evaluation points $\alpha_1,\ldots,\alpha_n\in\F_p$ and nonempty input lists $S_1,\ldots,S_n\subset\F_p$. The OPI problem asks for a polynomial $Q \in \F_p[X]$ with $\deg(Q) < k$ that maximizes the \emph{satisfaction ratio}
\[
s(Q):=\frac1n\abs{\left\{i:Q(\alpha_i)\in S_i\right\}}.
\]
We will often refer to $R := k/n$ as the \emph{rate} of the OPI instance.\footnote{We use this terminology as $R$ is the rate of the corresponding Reed--Solomon code.}
\end{definition}
OPI can be viewed as the problem of finding a vector in or near the set $C \cap \mathbf{S}$, where $\mathbf{S} = S_1 \times \cdots \times S_n$\footnote{Sometimes we abuse notation and refer by $\mbf{S}$ to both $(S_1, \dots, S_n)$ as well as $S_1\times\dots\times S_n$.}
and $C$ is the \emph{Reed--Solomon code}
\[
C := \{ (Q(\alpha_1), \ldots, Q(\alpha_n)) : Q \in \F_p[X], \deg(Q) < k \}.
\]
If $\rho_i = |S_i|/p$ is the density of $S_i$, and if the sets $S_i$ were random, we would expect $|C \cap \mathbf{S}| \approx |C| \prod_{i=1}^n \rho_i$.  The OPI problem above asks us to find an element in $C \cap \mathbf{S}$ if there is one; but at an existential level we may also wonder such an element exists for worst-case $S_i$, and if so how many.  In particular, is the number of such elements close to what we would expect for random lists?  This leads to the following definition, inspired by \cite{dlmnrr}, which studied similar discrepancy notions for random linear codes.
\begin{definition}[Near-optimal OPI Discrepancy]\label{def:disc}
Let $C\subset\F_p^n$, and fix nonempty $S_1, \ldots, S_n \subset \F_p$. Let $\rho_i:=|S_i|/p$ and let $\mbf{S}:=S_1\times\cdots\times S_n$. We say that $C$ has \emph{near-optimal OPI discrepancy} for input lists $S_1,\ldots,S_n$ if
\begin{equation}
\label{eq:def-perfect-opi}
\abs{\frac{|C\cap \mbf{S}|}{|C|\prod_{i=1}^n\rho_i}-1}
\le e^{-\Omega(n)}.
\end{equation}
\end{definition}

OPI is related to many problems that have been studied in different areas over the years.  For example, when the lists $S_i$ are relatively small, the question is related to \emph{list-recovery} of Reed-Solomon codes~\cite{gs,GR06,LP20,lms25,LS25,bcdz,bcpzz,gj26}; and in other parameter regimes, similar problems have been proposed as cryptographic hardness assumptions~\cite{np99,bn00,NP06}.  Most recently, OPI has gained prominence as a candidate for quantum supremacy.  The work \cite{dqi} gave a quantum algorithm, \emph{Decoded Quantum Interferometry} (DQI), which can find solutions to OPI with satisfaction ratio at least $\SCL_\rho(R/2)-o(1)$, where $R = k/n$ and where $\SCL$ is the \emph{semicircle law}
\begin{equation}
\label{eq:semicircle-law}
\SCL_\rho(\alpha):=
\begin{cases}
\left(\sqrt{\rho(1-\alpha)}+\sqrt{(1-\rho)\alpha}\right)^2 &\text{if }0\le \alpha\le1-\rho\\
1 &\text{if } \alpha\ge 1-\rho
\end{cases}
\end{equation}
In contrast, the best known classical algorithm can only find solutions with satisfaction ratio about $\rho + R(1 - \rho)$, which is smaller.  This leads to the following question.
\begin{question}\label{q:OPI}
For a fixed rate $R = k/n$ and fixed $\rho \in (0,1)$, what is the largest  $s^* \in (0,1]$ so that every OPI instance of rate $R$ with $|S_i|\ge \rho p$ for all $i$ admits a solution $Q \in \F_p[X]$ with satisfaction ratio $s(Q) \geq s^*$?  What is the largest fraction $s^*_{\mathrm{alg}}  \in (0,1]$ so that a solution $Q \in \F_p[X]$ with $s(Q) \geq s^*_{\mathrm{alg}}$ can be found efficiently with a quantum algorithm?
\end{question}
This question (both the existential and algorithmic versions) was first addressed by \cite{dqi}, who as mentioned above showed that $s^*, s^*_{\mathrm{alg}} \geq \SCL_\rho(R/2)$.  
Since then, several works~\cite{chailloux1,chailloux2,us,hy,sjo} have improved on this result, both existential and algorithmic. These improvements are depicted in \cref{fig:opi-picture}, and discussed more in \cref{sec:related_OPI}.  In particular, the best previous algorithmic results were limited at $s^*_{\mathrm{alg}} \leq \SCL_\rho(\alpha_J(R))$, where $\alpha_J(R) = 1 - \sqrt{1 - R}$ is the \emph{Johnson radius}; and for most of the range $R \in (0,1)$, this was not for worst-case input lists.\footnote{The works \cite{chailloux1,chailloux2} assume that the input lists $S_i$ are translated by i.i.d. random additive shifts $v_i \in \F_p$.}  In the balanced case, perfect solutions could be found efficiently only for $R \geq 0.75$.
 As we will see next, our results yield an algorithmic improvement all the way to $s^*_{\mathrm{alg}} = \SCL_\rho(R)$ (which we call the \emph{double semicircle law}), including efficient computation of a perfect solution for $R > 1/2$ in the balanced case.

\paragraph{Our results for OPI.}

Our first result for OPI improves the state-of-the-art for \cref{q:OPI}, giving an algorithmic result that outperforms even the best-known existential results.

\begin{restatable}[Quantum Algorithms for OPI]
{theorem}{mainalgo}
\label{thm:main-algo}
For any $R, \rho, \eps\in (0, 1)$, 
there exists a quantum algorithm with $\poly(n, \log p)$ running time that, given any OPI instance over $\F_p$ with rate $R$ and coherent membership oracle access to input lists $S_i$ where $|S_i|\ge \rho p$ for every $i$, has the following guarantees with high probability:
\begin{enumerate}
    \item Outputs a solution polynomial $Q$ with satisfaction ratio $s(Q)\ge \SCL_\rho(R)-\varepsilon$, 
    \item If $R+\rho>1$, outputs a solution polynomial $Q$ with $s(Q)=1$, i.e. $Q$ satisfies every constraint.
\end{enumerate}
\end{restatable}

\begin{remark}[Access model for the input lists]\label{rem:access}
    \cref{thm:main-algo} guarantees a quantum algorithm with membership oracle access to the lists $S_i$.  If the size of the field $\F_p$ is $p = \poly(n)$, then no membership oracles are needed, as the algorithm can simply take the lists $S_i$ as input and be its own oracle.  However, our results also work when $p$ is much larger than $\poly(n)$, in which case coherent membership oracle access is needed for an efficient algorithm. The precise computation model is similar to that of \cite{yz,hy} and is given in \cref{def:input-model}.
\end{remark}
We note that Item 2 of \cref{thm:main-algo} improves Item 1 in the case that $R + \rho > 1$: instead of a near-perfect solution (satisfying $1-\eps$ fraction of constraints for any $\eps>0$), we get a perfect solution that satisfies \emph{every} constraint. The result of \cref{thm:main-algo} is plotted in \cref{fig:opi-picture} for $\rho =1/2$ and $\rho = 1/4$.

\begin{figure}
    \centering
\input{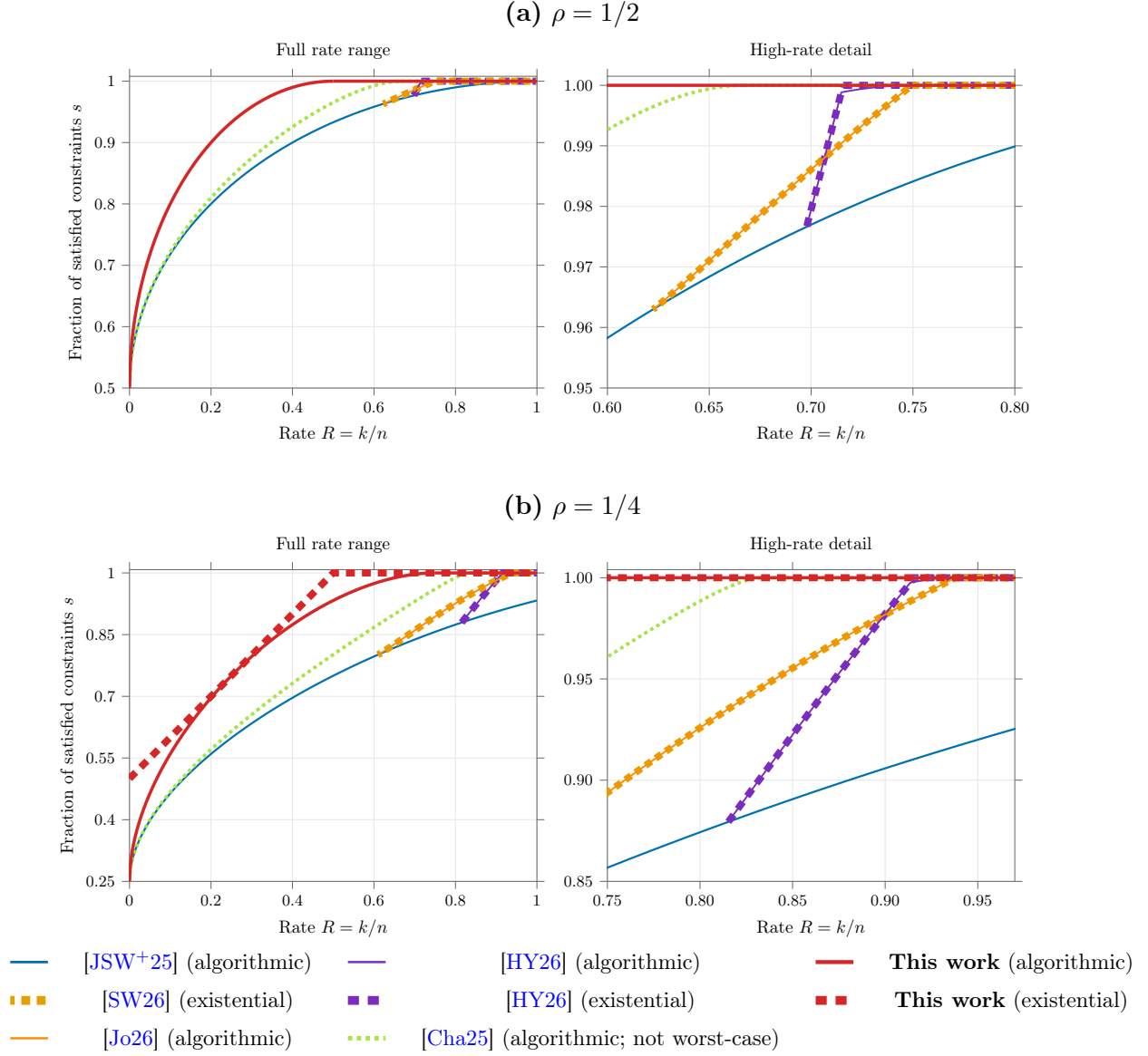}
\caption{Results for worst-case OPI (with arbitrary evaluation points $\alpha_i$) with (a) $\rho=1/2$ (balanced input lists) and (b) $\rho = 1/4$.  Existential results are shown in dashed lines, and (quantum) algorithmic results in solid lines.  We also include the result of \cite{chailloux2}, which is algorithmic but not worst-case (rather, it is average-case over random shifts in the sets $S_i$).  Our existential improvements (\cref{thm:main-existential}) does not improve over our algorithmic results for $\rho = 1/2$, and so does not appear in panel (a). As written, the algorithmic results of \cite{sjo} do not apply for $\rho = 1/4$, but a simple extension (which we also use in this work) applies, so we include \cite{sjo} in panel (b) as well as (a).  We note that \cref{thm:random-opi} implies that one can improve the red dashed line to the constant function ``$1$'' with high probability when the evaluation points are chosen at random, but the sets $S_i$ are still worst-case.}\label{fig:opi-picture}
\end{figure}

The algorithmic results in \cref{thm:main-algo} follow from improved existential results.  In more detail, the work \cite{sjo} showed how to make the existential results of \cite{us} algorithmic; we improve the existence results of \cite{us}, and use the techniques \cite{sjo} to arrive at \cref{thm:main-algo} above.  These improved existential results follow from improved results about optimal OPI discrepancy, which we state as our next main result below.

\begin{restatable}[Near-optimal OPI Discrepancy for Reed--Solomon Codes]{theorem}{mainopi}
\label{thm:main-opi}
Fix any $\rho \in (0, 1)$. Every Reed--Solomon code of rate $R$ over a prime field has near-optimal OPI discrepancy for every instance with $|S_i|\ge \rho p$ for every $i$, provided that at least one of the following holds:
\begin{enumerate}
    \item $R>1/2$;
    \item  $R>1/2-\delta$ and $\rho\ge1/2-\delta$, for some absolute constant $\delta >0$;
    \item $R+\rho>1$.
\end{enumerate}
\end{restatable}
In particular, \cref{thm:main-opi} implies that, whenever $1,2$ or $3$ hold, there exists a perfect solution to OPI (provided that $R, \rho > 0$ are constants and that $n$ is sufficiently large).
This means that there are parameter regimes where \cref{thm:main-opi} guarantees the existence of a perfect solution, but where \cref{thm:main-algo} is only guaranteed to find one with partial satisfaction (see \cref{fig:phase}). 

Given the above, it is natural to ask if one can get better existential guarantees than the ones implied by the algorithmic double semicircle law of \cref{thm:main-algo}, even in parameter regimes when \cref{thm:main-opi} does not apply.  Indeed we can, and this is our next result.

\begin{restatable}[Existential improvement beyond the double semicircle law]
{theorem}{mainexistential}
\label{thm:main-existential}
Fix $\rho\in(0,1/2)$ and $\varepsilon>0$. For sufficiently large $n$, every Reed--Solomon code
$C\subset\F_p^n$ of fixed rate $R\in(\rho, 1/2]$ over a prime
field with $p\ge n$ admits, for every instance with $|S_i|\ge\rho p$
for all $i$, a solution $Q\in\F_p[X]$ with $ \deg Q < k$ with satisfaction ratio $s(Q)\ge V_\rho(R)-\varepsilon$, where
\begin{equation}
    V_\rho(R):=1-(1-2\rho)(1-2R).
\end{equation}
\end{restatable}

We remark that this bound improves over the double semicircle law guarantee throughout the region $R\in(\rho, 1/2]$, as
\[
    V_\rho(R)-\SCL_\rho(R)
    =
    \left(\sqrt{\rho(1-R)}
          -\sqrt{(1-\rho)R}\right)^2>0.
\]
The existential improvement of \cref{thm:main-existential} is not visible in \cref{fig:opi-picture}(a) because $V_\rho(R) - \SCL_\rho(R)\to 0$ as $\rho \to 1/2$; the improvement can be seen for example when $\rho=1/4$, shown in \cref{fig:opi-picture}(b).

\Cref{thm:main-opi,thm:main-existential} establish improved existential results for \emph{all} OPI instances.  We are able to obtain even more improved existential results when the evaluation points $\alpha_i$ are random (but the input lists $S_i$ are still worst-case, chosen after the evaluation points are fixed).  Formally, we have the following theorem.
\begin{restatable}[Near-optimal OPI discrepancy for random Reed-Solomon codes]
{theorem}{randomopi}
\label{thm:random-opi}
Fix $R,\rho\in(0,1)$. For
sufficiently large $n$ and prime $p$ with $\log p= n^{\Omega(1)}$, with probability at least $1-e^{-\Omega(n)}$ over a uniformly
random tuple of distinct evaluation points $\alpha_1, \ldots, \alpha_n \in \mathbb{F}_p$, the Reed-Solomon code $C$ over $\F_p$ with rate $R$ and evaluation points $\alpha$ has near-optimal OPI discrepancy for every instance with $|S_i|\ge \rho p$ for every $i$.
\end{restatable}
In particular, \cref{thm:random-opi} implies that with high probability over the choice of the evaluation points $\alpha_i$, for all choices of input lists $S_1, \ldots, S_n \subseteq \mathbb{F}_p$ of density at least $\rho$, 
there are roughly $p^k\prod_i\rho_i$ perfect
solutions to the corresponding OPI instance.  We emphasize that this is still a result about worst-case OPI, in the sense that the input lists $S_i$ may be adversarially chosen depending on the evaluation points $\alpha_i$.

The proof of \Cref{thm:random-opi} builds on the techniques in a very recent result of Hwang, Maji, Shkredov and Ye~\cite{HMSY26}, who, as noted in \cref{rem:concurrent}, resolve \cref{q:llr} about LLR for \emph{random} evaluation points.

\vspace{.5cm}
We include a phase diagram in \cref{fig:phase} (over the phase space $(R, \rho)$) comparing \cref{thm:main-algo}, \cref{thm:main-opi}, \cref{thm:main-existential} and \cref{thm:random-opi} over the whole parameter space. 

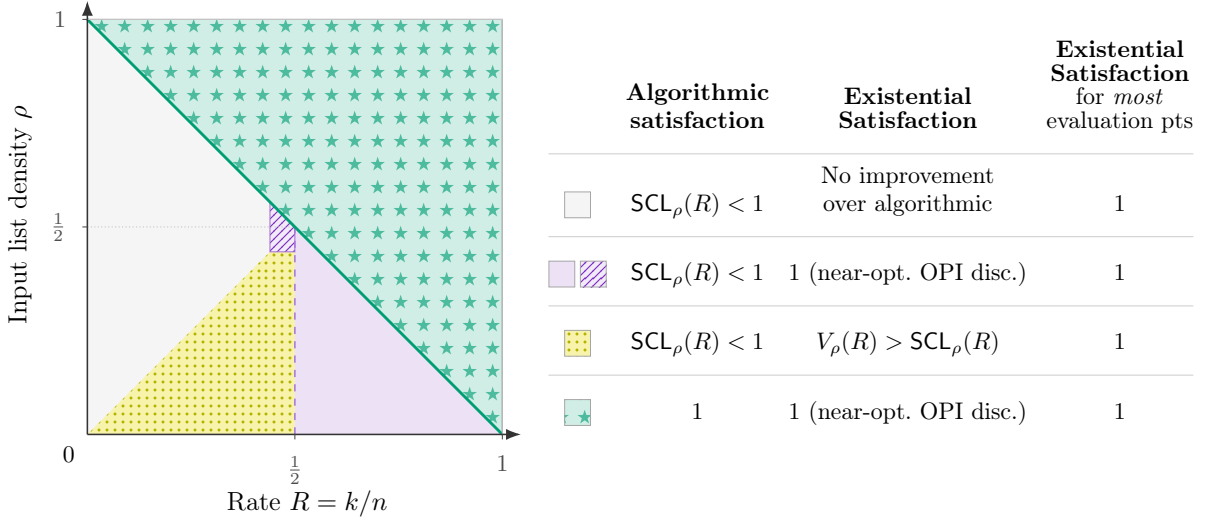
\begin{figure}[H]
    \centering
    \scalebox{0.9}{

\providecolor{OPIGreen}{HTML}{009E73}
\providecolor{OPIPurple}{HTML}{7B2CBF}
\providecolor{OPIYellow}{HTML}{F0E442}

\begin{tikzpicture}[
    x=6.1cm,
    y=6.1cm,
    >=Latex,
    every node/.style={font=\small},
    legend text/.style={anchor=west,align=left,text width=5.05cm,font=\small},
    axis label/.style={font=\normalsize},
]

\path[fill=black!4] (0,0) rectangle (1,1);

\path[fill=OPIYellow!45]
    (0,0) -- (0.5,0.5) -- (0.5,0) -- cycle;
\path[pattern=dots,pattern color=black!30!yellow]
    (0,0) -- (0.5,0.5) -- (0.5,0) -- cycle;
\path[color=black!30!yellow](0,0)--(0.5,0.5);

\path[fill=OPIGreen!18] (0,1) -- (1,1) -- (1,0) -- cycle;
\filldraw[pattern=fivepointed stars, pattern color=OPIGreen!70](0,1) -- (1,1) -- (1,0) -- cycle;

\path[fill=OPIPurple!14] (0.5,0) -- (1,0) -- (0.5,0.5) -- cycle;

\path[fill=OPIPurple!10]
    (0.44,0.44) -- (0.50,0.44) -- (0.50,0.50) -- (0.44,0.56) -- cycle;
\path[
    pattern=north east lines,
    pattern color=OPIPurple,
    draw=OPIPurple!70,
    line width=.35pt
]
    (0.44,0.44) -- (0.50,0.44) -- (0.50,0.50) -- (0.44,0.56) -- cycle;

\draw[densely dotted,black!24,line width=.45pt] (0,0.5) -- (0.5,0.5);
\draw[OPIPurple!65,dashed,line width=.65pt] (0.5,0) -- (0.5,0.5);
\draw[OPIGreen,line width=1.2pt] (0,1) -- (1,0);
\draw[black!22,line width=.45pt] (0,1) -- (1,1) -- (1,0);

\draw[black!80,line width=.65pt,-{Latex[length=2.2mm]}] (0,0) -- (1.045,0);
\draw[black!80,line width=.65pt,-{Latex[length=2.2mm]}] (0,0) -- (0,1.045);

\draw[black!70,line width=.45pt] (0.5,0) -- +(0,-0.012)
    node[below=3pt] {$\frac12$};
\draw[black!70,line width=.45pt] (1,0) -- +(0,-0.012)
    node[below=3pt] {$1$};
\draw[black!70,line width=.45pt] (0,0.5) -- +(-0.012,0)
    node[left=3pt] {$\frac12$};
\draw[black!70,line width=.45pt] (0,1) -- +(-0.012,0)
    node[left=3pt] {$1$};
\node[below left=2pt] at (0,0) {$0$};

\node[axis label,anchor=north] at (0.53,-0.112)
    {Rate $R=k/n$};
\node[axis label,rotate=90,anchor=south] at (-0.112,0.53)
    {Input list density $\rho$};


\def\OPILegendSwatch#1#2{%
    \tikz[x=1cm,y=1cm,baseline=-0.5ex]{%
        \path[fill=#1] (0,0) rectangle (0.25,0.25);
        \path[#2] (0,0) rectangle (0.25,0.25);
        \draw[black!35,line width=.35pt]
            (0,0) rectangle (0.25,0.25);
    }%
}
\def\OPILegendSwatch#1#2{%
    \tikz[x=1cm,y=1cm,baseline=0.10cm]{%
        \path[fill=#1] (0,0) rectangle (0.38,0.38);
        \path[#2] (0,0) rectangle (0.38,0.38);
        \draw[black!35,line width=.35pt]
            (0,0) rectangle (0.38,0.38);
    }%
}
\node[anchor=north west,inner sep=0pt,font=\small]
    at (1.11,0.95) {%
    \renewcommand{\arraystretch}{1.55}%
    \setlength{\tabcolsep}{4pt}%
    \def\OPILegendRule{%
        \noalign{\vskip 5pt}%
        \noalign{%
            {\color{black!20}\hrule height .35pt}%
        }%
        \noalign{\vskip 5pt}%
    }%
    \begin{tabular}{@{}c c c c@{}}
        & \shortstack{\textbf{Algorithmic}\\\textbf{satisfaction}}
        & \shortstack{\textbf{Existential}\\\textbf{Satisfaction}} 
        & 
\shortstack{\textbf{Existential}\\\textbf{Satisfaction} \\ for \emph{most} \\evaluation pts}
        \\
        \OPILegendRule

        \OPILegendSwatch{black!4}{}
        & $\SCL_\rho(R) < 1$
        & \shortstack{No improvement \\ over algorithmic}  &$1$\\
        \OPILegendRule

        \OPILegendSwatch{OPIPurple!14}{}%
        \hspace{2pt}%
        \OPILegendSwatch{OPIPurple!10}
            {pattern=north east lines,pattern color=OPIPurple}
        & $\SCL_\rho(R) < 1$
        & $1$ (near-opt.\ OPI disc.) &$1$ \\
        \OPILegendRule

        \OPILegendSwatch{OPIYellow!45}
            {pattern=dots,pattern color=black!30!yellow}
        & $\SCL_\rho(R) < 1$
        & $V_\rho(R) > \SCL_\rho(R)$  &$1$\\
        \OPILegendRule

        \OPILegendSwatch{OPIGreen!18}
            {pattern=fivepointed stars,pattern color=OPIGreen!70}
        & $1$
        & $1$ (near-opt.\ OPI disc.) &$1$
    \end{tabular}%
};

\end{tikzpicture}}
    \caption{Comparing results of
    Theorems \ref{thm:main-algo}, \ref{thm:main-opi}, \ref{thm:main-existential}, and \ref{thm:random-opi} for worst-case OPI.  In the gray, purple, and yellow regions, \cref{thm:main-algo}(1) guarantees an algorithmic solution with satisfaction ratio at least
    $\on{SCL}_{\rho}(R)-\varepsilon < 1$.
    In the purple regions, perfect solutions exist (and follow by \cref{thm:main-opi}(1) for solid purple or \cref{thm:main-opi}(2) for hatched purple; the $\delta > 0$ that controls the hatched purple region is not to scale).
    In the dotted yellow region, there exists a solution with satisfaction $V_\rho(R) > \SCL_\rho(R)$ by \cref{thm:main-existential}.
     In the starred green region, \cref{thm:main-opi}(3)
    gives near-optimal discrepancy and hence a perfect solution, while 
    \cref{thm:main-algo}(2) gives a quantum algorithm to find this solution.  Finally, in the entire region, \cref{thm:random-opi} implies that for \emph{most} evaluation points $\alpha_1, \ldots, \alpha_n$, for \emph{all} input lists $S_1, \ldots, S_n$ of density $\rho$, there exists a perfect solution to the corresponding OPI problem.}
    \label{fig:phase}
\end{figure}

\vspace{.5cm}

Finally, we turn our attention to the {classical} hardness of OPI.  It is known from~\cite{yz,folded-opi} that a similar problem to OPI---essentially the same problem, but defined for \emph{folded}\footnote{A folded RS code is a code $C \subset (\F_p^m)^N$, where the $i$'th symbol is $\vec{Q}(\alpha_i) := (Q(\alpha_i), Q(\gamma \alpha_i), \ldots, Q(\gamma^{m-1}\alpha_i)) \in \F_p^m$, where $\gamma \in \F_p^\times$ is a primitive element and where $\alpha_1, \ldots, \alpha_N \in \F_p^\times$ are distinct evaluation points so that $\gamma^a \alpha_i \neq \gamma^b \alpha_j$ for any $a,b \in \{0,\ldots, m-1\}$ and any $i \neq j \in [N]$.  Thus, the ``YZ problem'' studied by Yamakawa and Zhandry in \cite{yz} is, given oracle access to lists $S_1, \ldots, S_N \subset \F_p^m$, to find a polynomial $Q$ of degree less than $k$ so that $\vec{Q}(\alpha_i) \in S_i$ for all $i$.  The work \cite{folded-opi} proves a similar hardness result when a quantum algorithm may not find perfect solutions.  That is, they show that there is a quantum algorithm that can solve the folded OPI problem with some satisfaction ratio; but that any classical algorithm making only a polynomial number of queries to the membership oracles cannot exceed a smaller satisfaction ratio.} Reed--Solomon codes---is quantumly easy but classically hard (unconditionally) over very large folded alphabets, relative to a membership oracle for the lists $S_i$.\footnote{We note that a membership oracle is necessary, as in this parameter regime, the input lists $S_i$ are superpolynomial in $n$, and thus cannot be read as input efficiently.  A similar model is used in \cite{yz,hy,folded-opi}.}  Our next result adapts the hardness result of \cite{yz} to apply to OPI directly, rather than the folded version. 

\begin{restatable}[Quantum Advantage for OPI, Relative to a Membership Oracle---Informal, see \cref{thm:balanced-opi-advantage}]
{theorem}{mainadvantage}
\label{thm:main-opi-advantage}
There exists a random ensemble of balanced OPI instances of length $n$ with rate $1/\mathrm{polylog}(n)$ over prime fields $\F_p$ with quasipolynomial $p = \exp({\mathrm{polylog}(n)})$ 
such that the following hold, given access to a membership oracle for the input lists $S_i$:
\begin{itemize}
\item There is a polynomial-time quantum algorithm that finds a perfect solution with all but negligible probability.
\item Every randomized classical
algorithm making polynomially many queries to the membership oracle has negligible probability of outputting a solution with satisfaction ratio $\frac12+\Omega\left(\frac{1}{\log\log n}\right)$.
\end{itemize}
\end{restatable}
Here, the quantum algorithm has coherent access to the membership oracle, and the classical algorithm is allowed unrestricted computation
between queries. The probabilities are over the sampled instance from the ensemble and the algorithm, and the failure probabilities here are $e^{-\Omega((\log n)^2)}$. By tuning the parameters, one can obtain $e^{-\Omega(n^c)}$ probabilities for a constant $c>0$ at the cost of lower rate and classical satisfaction ratio control. The full theorem is stated in \cref{thm:balanced-opi-advantage} and this alternative parametrization is discussed in \cref{rem:alternative-param}.

\subsection{The relationship between LLR and OPI}

The relationship between LLR and OPI was first observed in \cite{us}.  At a high level, the connection is that similar \emph{Fourier proxies}---that is, quantities that bound the Fourier transform of the discrepancy $\left| \frac{|C \cap \mathbf{S}|}{|C|\prod_i \rho_i } - 1 \right|$ (for OPI) or the distance $\TV(L_{s\one},L_{s'\one})$ between two transcript distributions (for LLR)---can be used to solve both problems.  The work \cite{us} used this connection to import ideas from the LLR literature to make progress on OPI; but was not able to make progress on LLR itself.  In contrast, we make progress on both OPI and LLR, as outlined above.  In addition to this progress, we expand on the relationship between OPI and LLR.  

\paragraph{Our results on the relationship between LLR and OPI.}

We formalize the relationship observed by \cite{us}.  As in \cref{def:llr}, let $C \subset \F_p^n$ be the linear code $C = \{ (\alpha_i Q(\alpha_i))_{i=1}^n \,:\, \deg(Q) < k\}$, so that the LLR question concerns transcripts $(\ell_1(c_1), \ldots, \ell_n(c_n))$ when $c$ is chosen uniformly at random from a coset $s\one + C$.  It turns out that near-optimal OPI discrepancy is ``harder'' than LLR, in the sense that establishing near-optimal OPI discrepancy for the cosets $s\one + C$ will establish LLR for $C$.  We have the following theorem.

\begin{theorem}[Near-optimal OPI discrepancy is at least as hard as LLR---Informal, see \cref{thm:opi-harder}]
\label{thm:informal-opi-harder}
Let $\ell_1, \ldots, \ell_n:\F_p \to \{0,1\}$ be one-bit leakage functions.  For $z \in \{0,1\}^n$, define $S_i(z_i) = \ell_i^{-1}(z_i)$.  Let 
$C = \{ \alpha_i \cdot Q(\alpha_i) \,:\, \deg(Q) < k \}$ be the generalized Reed--Solomon code in \cref{eq:secret-zero-code}.  Suppose that, for all $s \in \F_p$, and for all $z \in \{0,1\}^n$,  $s\one + C$ has near-optimal OPI discrepancy for input lists $S_1(z_1), \ldots, S_n(z_n)$.
Then $C$ is 1-bit local leakage resilient against these leakages.
\end{theorem}
One might wonder if near-optimal OPI discrepancy is ``strictly'' harder than LLR; that is, does a converse to \cref{thm:informal-opi-harder} hold?  In fact, such a converse does not hold, and can fail to hold by up to an exponential factor, even for
fixed-rate GRS codes; see \cref{ex:opi-llr-strict}.

In the more formal version of the theorem (\cref{thm:opi-harder}), we show that more generally, if $v+C$ has near optimal OPI discrepancy for \emph{every} $v \in \F_p^n$, then $\TV(L_v, L_{v'})$ is small for \emph{every} $v,v' \in \F_p^n$.  Thus, LLR can be viewed as a restriction of near-optimal OPI discrepancy to just the cosets of $C$ given by constant shifts $s\one$.  In particular, any future improvements to near-optimal OPI discrepancy will imply improvements in LLR for Shamir sharing.\footnote{One subtlety is that in LLR, we typically assume that the $\rho_i$ can be arbitrary; while for OPI we assume that $\rho_i \geq \rho$ for some $\rho$, and the results are in terms of that $\rho$.  Thus, improvements of this sort for OPI will only lead to improvements of the same sort for LLR.  However, it is the case that, for example, improvements to the balanced case of OPI will imply improvements to the balanced case of LLR.}

\subsection{Technical Overview}
\label{subsec:overview}

In this section, we outline our approach.  At a high level, the outline is as follows:
\begin{itemize}
    \item \textbf{Step 1: Define Fourier Proxies.}  As with most previous work on both LLR and OPI, we bound the quantity of interest---the discrepancy for OPI, and the TV distance between two transcript distributions for LLR---by taking a Fourier transform.  We then bound the Fourier transform by an easier-to-think-about expression that we call a \emph{Fourier proxy.}\footnote{The term ``Fourier proxy'' is borrowed from the literature on LLR.}  The observation of \cite{us} was that similar Fourier proxies can be used for both LLR and OPI.  For OPI, we use the same proxy that was used in \cite{us}.  For LLR, we use a similar Fourier proxy as in \cite{kk23}, except that we re-center and re-normalize to deal with arbitrary biases~$\rho_i$.
    After we define the Fourier proxies that we will use, we also prove (a generalization of) \cref{thm:informal-opi-harder}, which informally states that OPI is ``harder'' than LLR.
    \item \textbf{Step 2: Above rate 1/2, both LLR and OPI reduce to dual list-recovery; proving \cref{thm:main-llr}(1) and \cref{thm:main-opi}(1).} In coding theory, a code $C \subset \F_p^n$ is \emph{list-recoverable} if, given any input lists $S_1, \ldots, S_n \subset \F_p$ of size at most $\ell$, there are not too many elements $c \in C$ so that $c_i \in S_i$ for all but a $\delta$-fraction of $i$'s (see \cref{def:list-recovery}).  
    For LLR, the work \cite{mpsw21} implicitly showed that $b$-bit LLR for a code $C$ follows from list-recovery of $C^\perp$; combined with recent breakthroughs on the list-recovery of Reed--Solomon codes~\cite{bcpzz,gj26}, this immediately proves \cref{thm:main-llr}(1), the statement for LLR with rate approaching $1/2$.  Translating this result to OPI to prove \cref{thm:main-opi}(1) requires a more involved argument, discussed more below.
    \item \textbf{Step 3: Crossing the rate $1/2$ barrier; proving \cref{thm:main-llr}(2) and \cref{thm:main-opi}(2).}
    Our techniques from Step 2 cannot work below rate $1/2$, due to a result of~\cite{n24} about the limitations of the phase-forgotten Fourier proxy that we use.  To get past rate $1/2$---that is, to prove the statements \cref{thm:main-llr}(2) and \cref{thm:main-opi}(2)---we analyze the Fourier proxies (without forgetting the phases) and use Gowers norms to take advantage of phase cancellations.
    \item \textbf{Step 4: Non-perfect OPI solutions; proving \cref{thm:main-algo}(1) and \cref{thm:main-existential}.}  Step 3---in particular, \cref{thm:main-opi}(2)---establishes the existence of a \emph{perfect} solution for OPI  at rate slightly less than $1/2$.  However, \cref{thm:main-algo,thm:main-existential} give partial-satisfaction guarantees at lower rates. To show the existence of such solutions, we use the framework of \cite{us}, who propose a quantum state $|\varphi_\ell\rangle$ that, when measured, returns a good OPI solution with high probability.  We improve the analysis of this state, leveraging our earlier connection to list-recovery, as well as the results of \cite{bcpzz,gj26} on list-recovery of Reed--Solomon codes to capacity.  To make this result algorithmic and prove \cref{thm:main-algo}(1), we use a result of Jo in \cite{sjo}, which shows how to efficiently prepare the state $\ket{\varphi_\ell}$ provided that the dual code can be list-decoded efficiently up to a certain radius.  Recent algorithmic list-decoding results of \cite{bcpzz,dkt}, along with the existential results above, allow us to apply the algorithm of \cite{sjo} to prove \cref{thm:main-algo}(1). For \cref{thm:main-existential}, we leverage our framework to analyze a conditioned DQI distribution; our sampler does not efficiently prepare this conditioned distribution, but existentially it improves the double semicircle law to $V_\rho(R)$.
    \item \textbf{Step 5: OPI when $\rho + R  > 1$; proving \cref{thm:main-opi}(3) and \cref{thm:main-algo}(2).}  In this step we study the case when $R + \rho > 1$.  In this case, \cref{thm:main-algo}(1) already shows that an \emph{near}-perfect solution exists---that is, one satisfying a $1 - \eps$ fraction of the constraints, for any $\eps > 0$---but we need to show that a \emph{exactly} perfect solution exists (to prove \cref{thm:main-opi}(3)) and show how to find it efficiently (to prove \cref{thm:main-algo}(2)).  For this, we use a truncation argument, described more below, to be able to  apply our list-recovery-based framework.  Once we have the existential result, we can again use the ideas of \cite{sjo} to prove the algorithmic result \cref{thm:main-algo}(2).

\item \textbf{Step 6: OPI at every positive rate for random evaluations; proving \cref{thm:random-opi}.}
To prove near-optimal OPI discrepancy for a random Reed-Solomon code, we consider the high probability good event that \cite{HMSY26} used to imply one-bit LLR. We show that, in fact, the good event implies near-optimal OPI discrepancy using our Fourier proxies.
    
    \item \textbf{Step 7: Classical hardness of OPI; proving \cref{thm:main-opi-advantage}.}  This step is technically separate from the rest of the paper, and can be read as a stand-alone section.  
    We give two scalar encodings of the
Yamakawa--Zhandry source problem~\cite{yz} to transform the folded Reed--Solomon hard instance it constructs into an (unfolded) Reed--Solomon instance. The first permits
arbitrary target rate but has sparse input lists. The second uses
balanced lists and a consistency argument, at the cost of a
vanishing target rate.
\end{itemize}

Before we expand on the technical details for the rest of this overview, we make a quick remark about the organization of the paper. 
\begin{remark}[Organization of the paper]\label{rem:org}
    We organize the technical overview into these seven steps, as these are the seven conceptual steps needed to prove our results.  However, we organize the paper itself slightly differently.  While Steps 2 and 3 are conceptually similar for LLR and OPI, the details differ, so we instead organize the paper to separate LLR and OPI, as follows. After some preliminaries in \cref{sec:prelim}, we complete Step 1 in \cref{sec:proxies}.  Then we complete Steps 2 and 3 for LLR in \cref{sec:llr}.  In \cref{sec:opi}, we complete Steps 2 and 3 again for OPI, as well as the non-algorithmic part of Step 5, and Step 6.  We translate our combinatorial results into algorithmic results (Step 4 and the algorithmic part of Step 5) in \cref{sec:algo}.  Finally, we handle Step 7 in \cref{sec:advantage}.
\end{remark}

Now, we expand on the technical details of each of the seven steps.  We begin by setting up some notation.  We use the same notation for both LLR and for OPI, by using $S_i$ for both the input lists in OPI and for the sets $S_i = \ell_i^{-1}(1)$ in LLR. 
In both cases, we have sets $S_i \subset \F_p$, and we define
\begin{equation}
\label{eq:def-gi}
g_i(x):=\frac{\one_{S_i}(x)-\rho_i}{\sqrt{\rho_i(1-\rho_i)}}.
\end{equation}

Throughout this discussion (and often in the paper going forward), we drop the polynomial notation $Q \in \F_p[X]$, referring instead to vectors $x \in C$, where $C$ is a Reed--Solomon code (possibly with no constant term, as in \cref{def:llr}).  Thus, for $x \in C$, $x_i$ represents the evaluation $Q(\alpha_i)$ of some appropriate polynomial $Q$.

\subsubsection*{Step 1: Define Fourier Proxies}
As mentioned above, we will bound the quantities we care about by \emph{Fourier proxies}, which we define here.  For OPI, fix $\mathbf{S} = S_1 \times \cdots \times S_n$, and let $C$ be a Reed--Solomon code.  Then the quantity that we care about (as in \cref{def:disc}) is
\begin{equation}
\label{eq:FO}
\FO(0)
:=\left|\frac{|C\cap \mbf{S}|}{|C|\prod_i\rho_i}-1\right|
=\left|\sum_{0\ne y\in C^\perp}
\prod_{i\in\supp(y)}\sqrt{\frac{1-\rho_i}{\rho_i}}\,\wh g_i(y_i)\right|,
\end{equation}
where the second equality follows from a Fourier transform over $\F_p$ (see \cref{subsec:fourier} for the Fourier conventions we use in this paper and \cref{lem:translated-opi-proxy} for a proof of \cref{eq:FO}).
Applying the triangle inequality to forget the phases gives the positive relaxation
\begin{equation}
\label{eq:FO+}
\FO^+
:=\sum_{0\ne y\in C^\perp}
\prod_{i\in\supp(y)}\sqrt{\frac{1-\rho_i}{\rho_i}}\,|\wh g_i(y_i)|
\ge\FO(0).
\end{equation}
Thus, to prove near-optimal OPI discrepancy for $C$ and $\mathbf{S}$, it suffices to show that either $\FO(0)$ or $\FO^+$ are exponentially small.  The Fourier proxy $\FO^+$ was used in \cite{us} to analyze OPI, and can be viewed a centered form of the Fourier proxy used in \cite{bdir}.

We contrast this with the following Fourier proxy for LLR. Recall that $L_{v}$ is the distribution on $\{0,1\}^n$ of the leakage transcript applied to uniform random vector in $v+C$. For $v=s\one$, this is the distribution of the leakage transcript  of the Shamir scheme, conditioned on secret $s$. 

For LLR, the quantity that we care about is $\TV(L_v, L_{v'})$.  We will bound this by the proxy
\begin{equation}
\label{eq:FL}
\FL(v,v')
:=
\frac14\sum_{\emptyset\ne A\subset[n]}
\left|
\sum_{\substack{y\in C^\perp\\\supp(y)=A}}
\left(e_p(-\langle v,y\rangle)-e_p(-\langle v',y\rangle)\right)
\prod_{i\in A}\wh g_i(y_i)
\right|^2.
\end{equation}

In \cref{lem:restated-llr-proxy}, we prove that
\begin{equation}\label{eq:llr-proxy}
\TV(L_v,L_{v'})^2\le\FL(v,v')\le\FL^+,
\end{equation} where $\FL^+$ is the positive relaxation 
\begin{equation}
\label{eq:FL+}
\FL^+
:=\sum_{\emptyset\ne A\subset[n]}
\left(\sum_{\substack{y\in C^\perp\\\on{supp} (y)=A}}\prod_{i\in A}|\wh g_i(y_i)|\right)^2.
\end{equation}
$\FL^+$ is similar to the proxy used by \cite{kk23}.  That argument obtained a rate threshold of $0.579$ for \emph{balanced} one-bit leakages.  For arbitrary biases, they obtain a threshold of $0.688$ via a reduction to the balanced case.  Our version can be viewed as a centered case for arbitrary biases, which lets us analyze the general case directly, which we do in \cref{subsec:llr-proxy}.  In fact, our results directly imply that the $0.579$ result of \cite{kk23} applies in the non-centered case too (see \cref{rem:bias-aware-kk}).

After we set up our Fourier proxies, we prove \cref{thm:informal-opi-harder}, which informally says that OPI is ``harder'' than LLR.  In addition to the statement in \cref{thm:informal-opi-harder} above, we also prove a statement at the level of Fourier proxies,  in the sense that the LLR proxies are bounded above by the OPI proxies.  
 
In more detail, write
$\mbf{S}_z:=\ell_1^{-1}(z_1)\times\cdots\times\ell_n^{-1}(z_n)$
and $\mbf{S}:=\mbf{S}_{\one}$. Then we show that
\[
\sqrt{\FL(s\one,s'\one;C,\mbf{S})}
\le
\max_{\substack{u\in\{s\one,s'\one\}\\z\in\{0,1\}^n}}
\FO(u;C,\mbf{S}_z),
\qquad
\sqrt{\FL^+(C;\mbf{S})}
\le
\max_{z\in\{0,1\}^n}\FO^+(C;\mbf{S}_z).
\]
The full statement---which includes the statement of \cref{thm:informal-opi-harder}, a generalization of it to arbitrary shifts, and the statement about about proxies---is given as \cref{thm:opi-harder}.

\subsubsection*{Step 2: Above rate 1/2, LLR and OPI reduce to dual list-recovery.}

In this step, we prove \cref{thm:main-llr}(1) and \cref{thm:main-opi}(1), which concern LLR and OPI discrepancy for any fixed rate $R > 1/2$.

We begin with LLR.  Our starting point is \cite{mpsw21}, which reduces $b$-bit LLR to conditions on the dual code.  They show that the \emph{Massey scheme} based on a \emph{random linear code}\footnote{This secret sharing scheme can be viewed as follows.  Choose a uniformly random matrix $G \in \F_p^{n \times k}$, and fix it; this defines the scheme.  Now, for a secret $s \in \F_p$, choose a random vector $r \in \F_p^k$ so taht $r_0 = s$.  Then the $n$ shares of $s$ are given by the $n$ elements of $G \cdot r$.} satisfies these conditions with high probability, provided the rate of the scheme is $R > 1/2$.  However, \cite{mpsw21} was not able to extend these results to Shamir sharing.

We observe in \cref{subsec:llr-above} that the condition identified in \cite{mpsw21} is implied by list-recovery of $C^\perp$.  For Shamir's scheme, $C^\perp$ is a \emph{generalized Reed--Solomon} (GRS) code.
Recent breakthroughs \cite{bcpzz,gj26} imply list-decoding and list-recovery up to capacity of GRS codes, which immediately proves \cref{thm:main-llr}(1). 

Thus, \cref{thm:main-llr}(1) follows essentially by assembling prior work as black boxes.  However, in order to apply these ideas to our other results, we need to un-black-box the role of list-recovery.  Our key observation is that the important consequence of list-recovery, which applies not just to LLR for $R > 1/2$ but also to our other results, is the following bound on weighted sums over the dual code.

\begin{restatable}[Weighted Fourier Thresholding]{lemma}{wftl}
\label{lem:wftl}
Fix constants $\gamma \in (0, 1)$ and $\Gamma>0$.  Let $p \geq n$ and let $C \subset \F_p^n$ be a generalized Reed--Solomon code of dimension $k$. Suppose that the weights $w_i:\F_p^\times \to \R_{\ge 0}$ for every $i\in [n]$ satisfies that
\[ \Vert w_i\Vert_\infty\le \gamma+o(1)\quad\text{and}\quad\Vert w_i\Vert_2^2=\sum_{\xi\ne 0}w_i(\xi)^2\le \Gamma+o(1)\] with both errors uniformly in $i$. Then, for every $t\in [n]$ with $2k-t=\Omega(n)$,
\begin{equation}
\label{eq:wftl}
\sum_{0\ne y\in C^\perp:|y|\le t} \prod_{i:y_i\ne 0} w_i(y_i)\le e^{o(n)} \gamma ^k=e^{-\Omega(n)}.
\end{equation}
\end{restatable}

We prove \cref{lem:wftl} in \cref{subsec:wftl}.
To see what this lemma has to do with LLR and OPI, observe that with $w_i =| \widehat{g_i}|$, the expression in \cref{lem:wftl} is similar to the expressions in the Fourier proxies $\FO^+$ in \cref{eq:FO+} and $\FL^+$ in \cref{eq:FL+}.  So the idea is to use it to bound those quantities, establishing LLR and near-optimal OPI discrepancy.  As a proof of concept, we show how to use \cref{lem:wftl} to prove a more quantitative version of \cref{thm:main-llr}(1) than follows as a black box from \cite{mpsw21,bcpzz,gj26}---see \cref{cor:one-bit-above}.  Next, we show how to use \cref{lem:wftl} to prove \cref{thm:main-opi}(1).  Going forward, we will also use \cref{lem:wftl} in Steps 3-5.

The proof idea for \cref{lem:wftl} is the following.
Fix a threshold $\theta > 0$.  Then, define $H_i=\{\xi\in\F_p^\times:w_i(\xi)\ge \theta\}$. List-recovery of $C^\perp$ (which for a GRS code follows from \cite{bcpzz,gj26}) implies that there are not too many codewords $y \in C^\perp$ with $y_i \in H_i$ for many $i$.  As there are not too many such $y$, they cannot contribute too much to the left-hand side of \cref{eq:wftl}.  On the other hand, the remaining codewords $y$ also have small contribution to the left-hand side of \cref{eq:wftl}, since for many coordinates $i$, we have $y_i \not\in H_i$, which means that $w_i(y_i) \leq \theta$.  Thus, the left-hand side of \cref{eq:wftl} cannot be too large.

Finally, we note an apparent barrier for this technique at $R=1/2$: for \cref{lem:wftl} to control every ``shell''  of $C^\perp$ (that is, the set $\{y \in C^\perp \,: \, |y| = t\}$ for each $t$), it needs to hold for weight $t=n$, which requires  $R>1/2$. This is further discussed in \cref{rem:wftl-param}.

\subsubsection*{Step 3: Crossing the rate-1/2 barrier.}

In this step, we prove \cref{thm:main-llr}(2) and \cref{thm:main-opi}(2), which concern LLR and OPI discrepancy for rates slightly smaller than $1/2$.

To understand our approach here, we explain a second apparent barrier at rate 1/2, which was identified by~\cite{n24} in the context of LLR.  That work constructs an example of leakage functions (or equivalently sets $S_i$) for which the phase-forgotten proxy $\FO^+$ is large for $R < 1/2$, but for which the value $\FO(0)$ that we actually care about is small.\footnote{We note that $\FO^+$ is our proxy for OPI, but \cite{n24} studies it in the context of LLR.  This is because this proxy is similar to one used by previous works in LLR~\cite{bdir,mpsw21,mnpw22}.  The leakage functions in \cite{n24} are based on quadratic residues.}
Thus a below-rate-1/2 result (for OPI at least) must retain the phases and take advantage of cancellation in the original expression $\FO(0)$ (and likely also in $\FL(v,v')$).  We do so with higher order Gowers norms---$U^3$ in particular---introduced in \cref{subsec:fourier}.

For the balanced case where $\rho\sim 1/2$, fix a sufficiently small $\delta>0$ and consider
$R\in[1/2-\delta,1/2]$. We split the terms in the sums over $y \in C^\perp$ in $\FO(0)$\footnote{Technically, we apply this to $\FO(v)$, which replaces $C$ with $C + v$ in $\FO(0)$.} and $\FL(v,v')$ by their weight, depending on whether $|y|$ is larger or smaller than 
\begin{equation}
\label{eq:puncture-size}
m:=\lfloor(2-\delta)k\rfloor
\end{equation}
For $y \in C^\perp$ of weight $|y| \leq m$,  we may apply \cref{lem:wftl} directly to $C^\perp$ (as $
    m\le 2k-\delta k$) to see their contributions to the relevant expressions are exponentially small. 
    
    On the other hand, suppose that $|y| > m$.  We focus on the LLR case for simplicity. Consider a support set $A$ for $y$ with $|A| > m$. Then the expression
    \begin{equation}\label{eq:wtb} \sum_{\substack{y \in C^\perp \\ \supp(y) = A}} e_p(-\langle v,y \rangle) \prod_{i \in A} \widehat{g_i}(y_j) \end{equation}
    that appears in $\FL(v,v')$ should be small; intuitively, this is because, if $A$ is large, then we should have lots of cancellations.  To formalize this, we use a version of the \emph{generalized von Neumann inequality} (\cref{lem:block-u3}), which bounds \cref{eq:wtb} in terms of the $U_3$ norms of the functions $g_i$.  Then we bound these norms $\|g_i\|_{U^3}$ in \cref{lem:u3-gap}, and this shows that \cref{eq:wtb} is small.  Plugging this into $\FL(v,v')$, we show that this too is small, as desired for our LLR result.  A similar argument works for $\FO$ in the case of OPI.

The approach above proves \cref{thm:main-llr}(2) and \cref{thm:main-opi}(2) in the balanced cases. We complete the proof to include arbitrary bias LLR and OPI with arbitrary density $\rho \geq 1/2 - \delta$ in \cref{sec:app-llr,sec:app-opi}.
The OPI extension uses subsampling from \cref{subsec:subsampling} The arbitrary-bias LLR extension instead uses the centered $L^1$ expansion and the Brascamp-Leib inequality to exploit strongly biased coordinates.

\subsubsection*{Step 4: Partial OPI satisfaction}
So far our results on OPI have concerned near-optimal OPI discrepancy, which can guarantee when a perfect solution exists, but doesn't say anything directly about imperfect solutions (that is, with satisfaction ratio less than $1$).
In this step, we prove \cref{thm:main-algo}(1) and \cref{thm:main-existential}.  The first says that  there is an efficient quantum algorithm that finds solution to OPI with satisfaction ratio about $\SCL_\rho(R)$ (the ``double semicircle law'').  The second says that there \emph{exist} solutions even beyond the double semicircle law when $\rho<R\le 1/2$. Both of these improved guarantees are less than $1$ for $R < 1- \rho$.

The idea is to consider the \emph{DQI state} $\ket{\varphi_\ell}$ proposed in \cite{dqi} to target a satisfaction ratio of $\SCL_\rho(\ell/n)$, where $\ell$ is a parameter we choose so that $\ell/n\le 1-\rho$. To define this state $\ket{\varphi_\ell}$, first define the \emph{$r$-wise discrepancy}\footnote{In contrast to \cite{us}, we define discrepancy in the ambient space of the code, $\F_p^n$, rather than the message space $\F_p^k$, to allow non-codewords.} of $x\in\F_p^n$ be
\begin{equation}
    q_r(x):= \sum_{A\in \binom{[n]}{r}} \prod_{i\in A}g_i(x_i) 
    \end{equation}
    where as before,
    $g_i(a) := \frac{\one _{S_i}(a)-\rho_i} {\sqrt{\rho_i(1-\rho_i)}}$.
    Fix some width $\sigma$ with $1\ll \sigma \ll \ell$, the DQI state is 
\begin{equation} 
\label{eq:dqi-u}
\ket{\varphi_\ell}\propto\sum_{x\in C}\sum_{r=0}^\ell u_rq_r(x) \ket{x}\quad\text{where}\quad u_r:=\binom{n}{r}^{-1/2}\one \{\ell - \sigma\le r\le \ell\}.
\end{equation}
    The algorithm is then to (1) prepare the state $\ket{\varphi_\ell}$, and (2) measure it to obtain codeword $x \in C$.  
With this approach, the analysis boils down to finding the largest $\ell$ such that:
\begin{enumerate}
    \item We can efficiently prepare the state $\ket{\varphi_\ell}$, and
    \item Measuring $\ket{\varphi_\ell}$ yields, on expectation, a good solution to the OPI instance.
\end{enumerate}

In \cite{dqi}, it was shown that for $\ell< \min\{k/2, (1-\rho)n\}$, Item (1) can be achieved by unique decoding of $C^\perp$, and that Item (2) holds using the fact that, with this choice of $\ell$, $2\ell +1\le k$ is below the distance of $C^\perp$ (which is $k+1$).  This choice of $\ell$ results in a quantum algorithm achieving a satisfaction ratio of at least $\SCL_\rho(\ell/n) = \SCL_\rho(R/2)$, which is the semicircle law.

In \cite{us}, it was shown that taking larger $\ell=k/2+\delta n$ is permissible for Item (2), provided one can bound an appropriate error term contributed from dual codewords with Hamming weight $t\in [k+1, \min\{n,2\ell+1\}]$. 
However, \cite{us} was not able to establish Item (1), and so their result was existential rather than algorithmic.  Follow-up work \cite{sjo} showed how to achieve Item (1) and efficiently prepare $\ket{\varphi_\ell}$, whenever $C^\perp$ can be list-decoded up to Hamming radius $\ell$. Using the Guruswami-Sudan algorithm, $\ket{\varphi_\ell}$ can be prepared up to any fixed slack below the \emph{Johnson bound} of $C^\perp$, which is $\ell = (1 - \sqrt{1-R})n$.
Thus, the work of \cite{sjo} established a quantum algorithm that recovers a state with expected satisfaction ratio $\SCL_\rho(\alpha)$ for any $\alpha < \min\{1-\sqrt{1-R}, R/2+\delta\}$ where the two cutoffs correspond to permissible $\ell$ for items (1) and (2) in \cite{sjo} and \cite{us}, respectively.\footnote{We note that the work \cite{hy} studied the same two items (1) and (2) for a related product state, but not for $\ket{\varphi_\ell}$ mentioned above, obtaining similar improvements (see \cref{sec:related_OPI}).  We follow the approach of \cite{us,sjo} in this work.}

To obtain our results in this step, we combine our list-recovery machinery developed above (in particular \cref{lem:wftl}) with the techniques from \cite{us,sjo}, and with the recent breakthroughs in algorithmic list-decoding and list-recovery of Reed--Solomon codes~\cite{bcpzz,gj26,dkt}.  For Item (1), preparing the state, we can use the result of \cite{sjo} as a black box, using the fact that \cite{bcpzz,dkt} implies that we may actually list-decode $C^\perp$ to radius approximately $R$ rather than $1 - \sqrt{1-R}$.  For Item (2), we use \cref{lem:wftl} (which itself uses the list-recovery result of \cite{gj26}) to control the error contributed by dual codewords with Hamming weight larger than $k$. Together, these ideas establish \cref{thm:main-algo}(1).

To prove \cref{thm:main-existential} and show the existence of solutions beyond the double semicircle law when $R \in (\rho, 1/2]$, we use an approach conceptually similar to Prange's algorithm.  In more detail, we choose a set $T \subset [n]$ of $t$ coordinates, and consider puncturing the relevant Reed--Solomon code to the coordinates indexed by $T$.  
We show that, whenever $t + 2\ell \leq 2k - \Omega(n)$, then conditioned on $Q(X)$ satisfying every remaining constraint indexed by $T$, its satisfaction ratio among the constraints in $[n] \setminus T$ is still given by the double semicircle law, with argument $\ell/(n-t)$.  

Together, the final satisfaction ratio is at least 
\[
\frac{t}{n}+\frac{n-t}{n} \on{SCL}_\rho\left(\frac{\ell}{n-t}\right).
\]
Then, optimizing $\ell, t$ as functions of $R$ gives $V_\rho(R)$, as claimed in \cref{thm:main-existential}. 

We emphasize this improvement is non-algorithmic, as discussed in \cref{rem:exist-no-algo}.

\subsubsection*{Step 5: OPI when $\rho+R>1$}
In this step, we prove \cref{thm:main-opi}(3) and \cref{thm:main-algo}(2), which both deal with the case when $\rho + R > 1$.  In this case, we establish near-optimal OPI discrepancy (which implies the \emph{existence} of a perfect solution), and show how to find that solution efficiently.

We begin with the existential result, \cref{thm:main-opi}(3).
As discussed in Step 3, the difficulty of OPI in the regime where $R<1/2$ is that \cref{lem:wftl} cannot apply to control Fourier errors from all shells of $C^\perp$ up to $t=n$.  However, when $R < 1/2$ but $\rho >1-R$, we can take advantage of the fact that the instance should intuitively be ``easy'' in that the $S_i$'s are very dense. Fix $\ell\sim \alpha n$ for $\alpha \in (1-\rho, R)$. For near-optimal OPI discrepancy, it suffices to show that for $\mathbf{S}:=S_1\times\dots \times S_n$
\begin{equation}
\label{eq:disc-check}
\mb{E}_{x\in C}[\rho^{-n}\one _{\mathbf{S}}(x)]=\mb{E}_{x\in C}[H(x)^2] = 1\pm e^{-\Omega(n)}\quad\text{where}\quad H(x):=\sum_{j=0}^n \rho^{(n-j)/2}(1-\rho)^{j/2}q_j(x).
\end{equation}
Indeed, the definition of near-optimal OPI discrepancy implies that we would like to show that $\E_{x \in C}[\rho^{-n} \mathbf{1}_{\mathbf{S}}(x)] = 1 \pm e^{-\Omega(n)}$; the expression in terms of $H$ follows by arithmetic manipulation (see \cref{eq:def-H}).
Inspired by \cite{hy}, we truncate the summands of $H(x)$ to those indexed by $j\le \ell$ to obtain $H_\ell$. 
We show that $H_\ell$ is a good proxy for $H$ on $\mbf{S}$: in particular, $H_\ell=(1-e^{-\Omega(n)})H$ pointwise on $\mbf{S}$, while its squared mass outside $\mbf{S}$ on $C$ is exponentially small (see  \cref{lem:key-truncation}).  Now, if we can show that
\begin{equation}
\label{eq:disc-proxy-check}
\mb{E}_{x\in C}[H_\ell^2]=1\pm e^{-\Omega(n)},
\end{equation}
this will imply \cref{eq:disc-check} and therefore near-optimal OPI discrepancy. 
To establish \cref{eq:disc-proxy-check}, we can again use our list-recovery-like lemma \cref{lem:wftl}. The key is that the truncation of $H_\ell$ means that we won't have to handle ``shells'' all the way up to $t = n$, which was the issue identified earlier; instead, it turns out that we only need to deal with shells $t \leq \min\{n, 2\ell + 1\} \leq 2k -\Omega(n)$ (see \cref{lem:key-truncation}).  This establishes \cref{eq:disc-proxy-check} and hence \cref{thm:main-opi}(3).

For a quantum algorithm to sample a perfect solution in this regime, we lower bound the probability that any particular perfect solution $\ket{x}$ is sampled when we measure $\ket{\varphi_\ell}$ for $\ell \sim (1-\rho)n$ by a direct computation.  Then we observe that these perfect states are abundant by the preceding paragraph.  Now, repeating the same algorithm as in Step 4 polynomially many times until we find a perfect solution will do the job, proving \cref{thm:main-algo}(2).

\subsubsection*{Step 6: OPI at every positive rate for random evaluations}

In this step, we prove \cref{thm:random-opi}, which establishes
near-optimal OPI discrepancy for a Reed-Solomon code with random evaluation points with high probability. 

Our proof builds on the very-recent result of \cite{HMSY26} for one-bit LLR.  The work \cite{HMSY26} was independent and concurrent from the first version of our work, which did not include \cref{thm:random-opi}; in this version, we have adapted the techniques of \cite{HMSY26} to prove  \cref{thm:random-opi}.

Let $C \subset \F_p^n$ be a Reed--Solomon code with random evaluation points $\alpha_1, \ldots, \alpha_n$, and let $k=Rn$ be the dimension of $C$.
\cite{HMSY26} approximates each local indicator as
$\one_{S_i}=G_i+H_i$, where  
$G_i$ is appropriately structured and $H_i$ has
small Gowers $U^s$-norm for some large constant $s$. Upon expanding $\prod_i\one_{S_i}(x_i)$ in the LLR proxy, consider a set $J\subset [n]$ indexing products of $H_i$. If $|J|$ is large, the small Gowers $U^s$-norm gives exponential decay. To handle the terms with small $|J|$, 
\cite{HMSY26} uses the structure of $G_i$ for $i\not\in J$ to construct a good event $\mathsf{Good}$ (over the randomness of the choice of evaluation points) that implies that the LLR proxy corresponding to those small $J$ are small, and they show that $\mathsf{Good}$ occurs with high probability for a random Reed--Solomon code. 

We show that on the event $\mathsf{Good}$, in fact the OPI proxy is also small, thereby proving \cref{thm:random-opi}. 
We can use much of \cite{HMSY26} as a black box, and our work is to extend their techniques to work for OPI instead of LLR.
This is done in \cref{subsec:rrs} and relies heavily on our Fourier analytic framework and proxies. The comparison is further discussed in \cref{rem:like-bdir}.
\\ 

\subsubsection*{Step 7: Quantum Advantage of OPI}
In this step, we prove \cref{thm:main-opi-advantage}, about the classical hardness of OPI relative to a membership oracle.  We begin with the breakthrough result of
Yamakawa and Zhandry~\cite{yz}.  They  construct a high-rate folded Reed--Solomon code
$C_0\subset\Sigma^m$, where $\Sigma = \F_q^s$ for a prime power $q$ and some integer $s$, together with independent random oracles
$O_j:\Sigma\to\{0,1\}$, which determine the input lists $S_j$.  They show that there is a quantum algorithm that finds a perfect solution
$z\in C_0$ satisfying $O_j(z_j)=1$ for every $j$, whereas no
classical algorithm making polynomially many
queries can find a codeword with, say, $99\%$ satisfaction. See \cref{subsec:yz-recap} for more details.

We cast this construction as an OPI instance.  We are motivated by the fact that the recent list-decoding algorithms of \cite{bcpzz, dkt} can handle exponentially large fields (while still running in time polynomial in $n$, the length of the relevant Reed--Solomon code).  This means that a single scalar in a very large field can encode all $m$ coordinates of the Yamakawa-Zhandry instance. More specifically, we construct two types of input lists: $S_n$ gives an oracle that is equivalent to all $(O_1, \dots, O_m)$ simultaneously, so we can simulate the source oracles in our reductions.  Then we duplicate for $n-1$ times another oracle that encodes ``perfectly solving the Yamakawa-Zhandry instance''. This gives us polynomial time reductions so that quantumly solving the folded instance perfectly translates to solving the new unfolded instance perfectly; while classically it is hard to solve the folded instance and therefore hard to satisfy any of the $n-1$ constraints in the unfolded instance. This construction is spelled out in \cref{subsec:sparse-ensemble}.

This construction treats evaluations completely independently, as all the encoding happens at the scalar level. It therefore permits arbitrary target RS rate, but the key duplicated input list is
necessarily sparse: a random source message passes all $m$ independent
source constraints with probability $2^{-m}$.  In
\cref{subsec:balanced-ensemble}, we remove this sparsity by
distributing it into many small, exactly balanced parity tests.  The new
difficulty is ``consistency'', which relies on constraining the dimension of $C$ to require any polynomial $Q$ to present a consistent solution to each parity check. This gives the balanced input list but vanishing rate construction in \cref{subsec:balanced-ensemble} that proves \cref{thm:main-opi-advantage}.
\\

\subsection{Related Work}
\label{subsec:literature}
Here, we further survey related work.  We begin with related work on local leakage resilience for Shamir's scheme.

\subsubsection{Local Leakage Resilience for Shamir's Scheme}\label{sec:related_LLR}
Local leakage resilience was introduced in \cite{bdir}, who were inspired by results in the regenerating codes literature~\cite{gw,tyb18} which implied that Shamir's scheme was \emph{not} LLR over extension fields.  The work \cite{bdir_conference} showed that Shamir's scheme (with arbitrary evaluation points) was $1$-bit LLR provided that $R = k/n$ satisfied $R > 0.907$.  Later, this was improved to $0.8675,0.85,0.78,0.69,0.668$ by \cite{mpsw21,bdir,mnpw22,kk23,k24}, respectively.  In the case of \emph{balanced} leakage functions, \cite{kk23} showed that $R > 0.58$ sufficed.  However, all of these works faced a barrier at $R = 1/2$.  This barrier was intrinsic in the technique used by \cite{bdir_conference,bdir,mpsw21,mnpw22}; in particular, the Fourier proxy that they bounded was not in fact small for $R < 1/2$ for certain one-bit leakage functions, even though the term that this proxy bounds is indeed small~\cite{mpsw21,n24}.
Other work \cite{kk23} that used a different proxy did not face such a clear barrier at $R = 1/2$, but their techniques and proofs do seem to require $R > 1/2$.  Inspired by this, \cite{n24} asked if it was possible to beat the so-called ``$1/2$ barrier.'' That work showed that for certain leakage functions (in particular, \emph{quadratic residue} leakage), one could get LLR with arbitrarily small constant rate $R$.  However, the question of whether one could break or even approach the $R = 1/2$ barrier for \emph{worst-case} leakage functions remained open until our work.

We note that the work \cite{MNPSWYY22} did prove a result approaching the $1/2$ barrier, but this applied only to bounded classes of leakage functions and to Shamir schemes with {random} evaluation points.\footnote{In more detail, the result shows that, for any family of leakage functions $\mathcal{L}$ that is sufficiently small, then a Shamir scheme of rate $1/2 + \delta$ with random evaluation points, over a sufficiently large field, will be LLR to any leakage functions in $\mathcal{L}$ with high probability.}  Other works considered restricted classes of leakage functions, for example physical bits (rather than leaked bits that require non-trivial computation of the stored physical bits)~\cite{physical_leakage}.  In this work, we focus on the question of \emph{worst-case} families of leakage functions.

Other works have managed to break the $1/2$ barrier (to obtain results for any positive rate) for secret sharing schemes other than Shamir, specifically for Massey secret sharing schemes corresponding to \emph{random linear codes}\footnote{That is, we draw a uniformly random matrix $G \in \F_p^{n \times k}$.  To share a secret $s$, let $x = (s, r_1, \ldots, r_{k-1})$, where $r_i$ are chosen at random from $\F_p$, and then use the $n$ values in $Gx$ as the $n$ shares.} (Shamir's scheme is the Massey scheme corresponding to Reed-Solomon codes).  The work \cite{mpsw21} established that rates \emph{approaching} $1/2$ would work for these secret sharing schemes, with high probability.  The work \cite{dlmnrr} established it for any positive rate, for fields of size $|\F| = O(n^{1/5})$; since MDS codes do not exist in this parameter regime, the resulting secret schemes are \emph{ramp} secret sharing schemes, with a gap between the number of parties required to recover the secret and the number which can be guaranteed to learn nothing about it.  To the best of our knowledge, there are no threshold (as opposed to ramp) secret sharing schemes known to be one-bit LLR below rate $1/2$ before this work  We note that \cite{dlmnrr} also studies a notion of discrepancy for random linear codes, which inspired the near-optimal OPI discrepancy that we study in this work.

\paragraph{Independent and Concurrent work.}  In independent work, \cite{n26} also obtained a result analogous to \cref{thm:main-llr}. In particular, they also prove a version of \cref{thm:main-llr}(2), showing that there is a constant $\delta > 0$ so that Shamir's scheme with rate $1/2 -\delta$ is one-bit LLR.  Their version of \cref{thm:main-llr}(1) is more quantitative, and shows that Shamir's scheme with rate $1/2 + \delta$ is $b$-bit leakage-resilient, for $b = O_\delta(\log n)$.  The techniques are similar to those in \cref{sec:llr}, combining list-recovery of Reed--Solomon codes, a generalized von Neumann inequality, and Gowers norms. Instead of thresholding the Fourier weights in \cref{lem:wftl}, the Efron-Stein inequality is used.

Concurrent to the first arXiv version of this work, which did not include \cref{thm:random-opi}, \cite{HMSY26} showed that \emph{random} Reed-Solomon codes are one-bit LLR for every positive rate.  For the current version of this work, we have adapted their techniques to obtain our result for near-optimal OPI discrepancy of \emph{random} Reed-Solomon codes at every positive rate, \cref{thm:random-opi}.

\subsubsection{Optimal Polynomial Intersection}\label{sec:related_OPI}
As mentioned earlier, questions related to OPI have been studied in other parameter regimes, most notably in the coding theory community as \emph{list-recovery} of Reed--Solomon codes~\cite{gs,GR06,LP20,lms25,LS25,bcdz,bcpzz,gj26}; similar problems have also been proposed (and refuted) as cryptographic hardness assumptions~\cite{np99,bn00,NP06}.  It is also related to the problem studied in the breakthrough work of Yamakawa and Zhandry~\cite{yz}, who show that a similar problem for \emph{folded} Reed--Solomon codes over exponentially large fields is classically hard but quantumly easy (relative to a membership oracle for the input lists).

In our main parameter regime (with polynomial field sizes), OPI was first studied by \cite{dqi}, who gave a quantum algorithm that finds a solution with satisfaction ratio up to $\SCL_\rho(R/2)$. It has since been the target of related quantum algorithms, including~\cite{chailloux1,chailloux2,khattar2025verifiable,rosmanis2026nearly,sjo,hy}, which are inspired by \emph{Regev's reduction}~\cite{regevReduction}.

The works \cite{chailloux1,chailloux2} gave algorithms to find solutions in a not-quite-worst-case model up to $\SCL_\rho(\alpha_J(R))$, where $\alpha_J(R) = 1 - \sqrt{1 -R}$ is the Johnson radius.  More precisely, in this model the input list begin as worst-case $\widetilde{S}_i \subset \F_p$, except that they are then translated by a random field element $v_i \in \F_p$ to obtain $S_i = \widetilde{S}_i + v_i$.  This result guarantees an algorithmic solution (with high probability over the input lists) provided that $R > 0.75$.

The work \cite{us} improved the existential bound in the worst-case setting, guaranteeing the existence of solutions with satisfaction ratio $\SCL_\rho(R/2 + \delta)$ for some $\delta > 0$; in particular, this implied the existence of an asymptotically perfect solution whenever $R \geq 0.749$, slightly smaller than the $0.75$ guaranteed by \cite{chailloux2} (although not algorithmic).  We note that \cite{us} also observed the key connection between OPI and LLR that we exploit, namely that similar Fourier proxies can be used to get a handle on both problems (although their improvements were not enough to improve on the state-of-the-art for LLR, only for OPI).  We make their observations more precise in \cref{thm:informal-opi-harder}.  

The work \cite{sjo} showed how to make the results in \cite{us} algorithmic, provided that $R/2 + \delta < \alpha_J(R)$.  In particular, this showed how to efficiently find a solution with $s(Q) \geq 1-o(1)$ at rates $R\ge 3/4$ for balanced worst-case input lists, matching the quantitative result of \cite{chailloux2}.  We use the results of \cite{sjo} to turn our improved existential results into improved algorithmic results.

Finally, the work \cite{hy} (concurrent with \cite{sjo}) recovered similar algorithmic bounds as \cite{sjo} via a different method, and improved the existential results.  Quantitatively, \cite{hy} guarantees the existence of a perfect solution for $R \geq 0.715$.

In very recent work, \cite{folded-opi} gave an algorithm for a folded variant of OPI that also approached the same ``double semicircle law.'' However, their paper is for a folded version of OPI, while we study the original ``scalar'' OPI.

On the hardness side, the works \cite{yz,folded-opi} show that a \emph{folded} version of OPI over very large fields that can be solved by a quantum algorithm with oracle access to the input lists is hard for classical algorithms, in the sense that any classical algorithm requires super-polynomially many queries to the membership oracles.  Our separation result \cref{thm:main-opi-advantage} relies on the result of \cite{yz}, and transforms it into a scalar OPI instance.  In other parameter regimes, OPI has been proposed as a cryptography hardness assumption~\cite{np99}, but was broken by \cite{bn00} using classical lattice attacks; however, these lattice attacks do not seem to work in the OPI parameter regime (see \cite{dqi}).  OPI can be seen as a special case of \emph{max-LINSAT}, where the goal is to satisfy as many linear constraints as possible.  It is known that general max-LINSAT problems are NP-hard even to approximate even beyond the straightforward ``random assignment'' threshold~\cite{KSE26}.

\subsection{Open Problems}
We conclude with a few open problems.
\begin{itemize}
    \item What is the ``correct'' rate threshold for LLR for Shamir schemes?  That is, what is the smallest $R^*$ so that any Shamir scheme with rate at least $R^*$ is one-bit LLR?  A conjecture of \cite{bdir} posits that $R^*$ can be an arbitrarily small constant; we have shown that $R^*$ can be less than $1/2$.  The very recent result of \cite{HMSY26} resolved this conjecture affirmatively for \emph{almost} every tuple of evaluation points, over sufficiently large fields.  Does a similar result hold for \emph{all} evaluation points?
    \item What is the ``correct'' satisfaction ratio for an OPI instance of rate $R$?  That is, given $R \in (0,1)$ and $\rho \in (0,1)$, what is the largest $s^*$ so that, for any OPI instance of rate $R$ and input lists with density $\rho$, there exists a solution with satisfaction ratio at least $s^*$?  In some parameter regimes, we have shown that $s^* = 1$, but in the other parameter regimes, we have only given bounds on $s^*$.  In particular, we have shown that $s^* \geq V_\rho(R)$, but our results \cref{thm:main-opi}(2) suggest that
    one can do better, at least by some arbitrarily small $\delta > 0$.   We have shown in \cref{thm:random-opi} that the answer is $s^* = 1$ for all $\rho, R \in (0,1]$, for \emph{most} choices of evaluation points; does the same hold for \emph{all} choices of evaluation points?
    \item What is the largest fraction $s^*_{\mathrm{alg}}\le s^*$ such that a solution to OPI with satisfaction at least $s^*_{\mathrm{alg}}$ can be found efficiently by a quantum algorithm? In particular, there are parameter regimes where we have shown the existence of solutions with satisfaction ratio beyond what is currently attainable quantum algorithmically; can this gap be closed?
    \item What is the classical hardness of OPI?  We show a separation between classical and quantum relative to membership oracles when the field is very large.  Can one show a similar separation (perhaps under some standard cryptographic assumptions)  for polynomial field sizes?
\end{itemize}

\subsection{Statement on AI Usage}
Prior to the recent breakthroughs on list decoding and list recovery of RS code \cite{bcpzz,gj26}, the connections between OPI, LLR, and dual list recovery in \cref{sec:llr} were made by the authors, and the Steps 1 to 4 in \cref{subsec:overview} were executed for random Reed-Solomon codes. Subsequently, ChatGPT 5.6-Sol was heavily consulted for various techniques to improve the numerical thresholds for RS codes, e.g. using Brascamp-Lieb inequalities and (recent improvements to) the linear programming bound of McEliece, Rodemich, Rumsey, and Welch \cite{mrrw}.

Following the breakthroughs on list decoding and list recovery of RS codes, these numerical techniques were redundant, as is the distinction between random RS codes and arbitrary RS codes. The authors then used ChatGPT 6-Astra to understand existing literature and it produced the key truncation idea in Step 5 of \cref{subsec:overview}, as well as the extension to the arbitrary bias case of \cref{thm:main-llr}(2) in \cref{sec:app-llr}. Similarly in Step 6, ChatGPT 6-Astra is  used to understand the higher order Fourier analysis in \cite{HMSY26}. It isolated the black box \cref{fact:hmsy-model} from \cite{HMSY26}, which we used to prove \cref{thm:random-opi}, building on our framework in the earlier steps.

Inspired by the new list decoding algorithms that handle exponentially large fields, the authors attempted to apply the techniques of \cite{yz} to show relative separation of OPI, and constructed the sparse ensemble in Step 7 of \cref{subsec:overview}. Prompted with the idea of  a balanced instance via distributing the oracle checks, ChatGPT 6-Astra obtained a version of the balanced ensemble and \cref{thm:main-opi-advantage}
upon which the authors further simplified the construction and proof.

Outside of the aforementioned contributions, AI was used only for copyediting and proof checking, as well as for generating the \texttt{tikz} code that produces \cref{fig:opi-picture,fig:phase,fig:R-KK-benchmark}. All proofs appearing in the manuscript were written by the authors. All responsibility for errors lies with the authors.

\subsection*{Acknowledgements} 
YS thanks Rohan Goyal for helpful discussions. YS is funded by the NSF Graduate Research Fellowship and the Stanford Graduate Fellowship. MW is partially funded by NSF grants CCF-2231157 and CNS-2321489, and is grateful for funding from Google to support this work.

\section{Preliminaries}
\label{sec:prelim}

\subsection{Notation and Conventions}
\label{subsec:notation}
We write $[n]:=\{1,\ldots,n\}$, with $[0]=\emptyset$, and
$\binom{S}{k}$ for the family of $k$-element subsets of $S$.  
We write $\langle u,v\rangle$ for the standard dot product
and $\E_{x\in X}$ for expectation over a uniformly random $x\in X$.

We use the standard asymptotic notation $f\lesssim g$ for $f=O(g)$, $f\ll g$ for $f=o(g)$, $f\gtrsim g$ for $g\lesssim f$, $f\asymp g$ for $f=\Theta(g)$, and $f\sim g$ when $f/g\to 1$. Subscripts indicate the
parameters on which the implicit constants may depend.

All unqualified logarithms are natural.  For $x\in[0,1]$, define the natural-base and base-two binary entropy functions by continuity at the
endpoints:
\[
        H(x):=-x\log x-(1-x)\log(1-x),
        \qquad
        H_2(x):=\frac{H(x)}{\log2}
        =-x\log_2x-(1-x)\log_2(1-x).
\]
Both functions are left undefined outside $[0,1]$.  Whenever an optimization
contains $H$ or $H_2$, its domain is implicitly restricted so that every
entropy argument lies in $[0,1]$.
We use $\sup\emptyset=-\infty$ and $\inf\emptyset=+\infty$, and use the
continuous perspective convention $aH(b/a)=0$ when $a=b=0$.

\subsection{Coding Theory}
\label{subsec:codes}

An \emph{$[n,k]$ linear code} is a $k$-dimensional subspace of $\F_p^n$; its \emph{rate} is $k/n$, and its \emph{dual} is $C^\perp:=\{y:\langle x,y\rangle=0\text{ for every }x\in C\}$. We write $|y|:=|\supp(y)|$ for Hamming weight. A code is \emph{Maximum Distance Separable} (MDS) if its distance is $n-k+1$, or equivalently if projecting to any $k$ coordinates forms a bijection $C\to\F_p^k$. 

We will be interested in Reed--Solomon (RS) codes as well as generalized Reed--Solomon (GRS) codes. A \emph{generalized Reed--Solomon} code is the set
\[
C=\{(a_iQ(\alpha_i))_{i=1}^n: Q \in \F_p[X], \deg Q<k\},
\]
where the $\alpha_i$ are distinct and every $a_i\ne0$. 
When the $a_i$ are all equal to $1$, we say that $C$ is a RS code with evaluation points $\alpha_1, \ldots, \alpha_n$.

RS codes and generalized RS codes are MDS, and the dual of a GRS code is an $[n,n-k]$ GRS code of distance $k+1$. For $I\subset[n]$, the \emph{puncturing} $C$ to $I$ is the set $C|_I = \{ c|_I \,:\, c \in C\}$, where $c|_I$ is the restriction of $c$ to the coordinates in $I$.  The dual of $C|_I$ consists of the restrictions to $I$ of the words of $C^\perp$ that vanish outside $I$. Thus $(C|_I)^\perp$ is the shortening of $C^\perp$ to $I$. 
Indeed, a word on $I$ is orthogonal to every punctured codeword exactly when its extension by zero is orthogonal to $C$. When $|I|\ge k$ and $C$ is MDS, $C|_I$ has dimension $k$, and the shortened dual has dimension $|I|-k$.

We now introduce the recent breakthrough results on list decodability and recoverability of Reed-Solomon codes at capacity \cite{bcpzz,gj26,dkt}. We need algorithmic list decoding, but combinatorial list recovery properties only.

\begin{definition}[List decoding and recovery]
\label{def:list-recovery}
A code $C\subset\F_p^n$ is \emph{$(\delta,\ell,L)$-list recoverable} if, for
any coordinate lists $V_1,\ldots,V_n\subset\F_p$ with
$|V_i|\le\ell$, at most $L$ words $y\in C$ satisfy $y_i\in V_i$ on at
least $(1-\delta)n$ coordinates.
\emph{List decoding} is the singleton-list case: a code is $(\delta,L)$-list
decodable if every Hamming ball of radius $\delta n$ contains at most
$L$ codewords.
\end{definition}

\begin{theorem}[Reed-Solomon list decoding {\cite[Theorem~9.16]{dkt}}]
\label{thm:rs-ld}
For $\eps\in(0,1)$ and all sufficiently large $n$ and prime $p\ge n$, every generalized Reed-Solomon code $C\subset\F_p^n$ with dimension $1\le k\le(1-\eps)n$ admits a deterministic algorithm that, given any $z\in\F_p^n$, in $n^{O_\eps(1)}\widetilde O(\log p)$ running time (uniformly in $k$ and the evaluation places) outputs every $c\in C$ such that
\[
|c-z|\le\left(1-\frac{k}{n}-\eps\right)n.
\]
\end{theorem}

\begin{theorem}[Reed-Solomon list recovery {\cite[Theorem~1.1]{gj26}}]
\label{thm:rs-lr}
For $\eps\in(0,1)$, $\ell\in\N$, and all sufficiently large $n$ and prime $p\ge n$, every generalized Reed-Solomon code $C\subset\F_p^n$ with dimension $1\le k\le(1-\eps)n$ is
\[
\left(1-\frac{k}{n}-\eps,\ell,n^{O_{\eps,\ell}(1)}\right)
\text{-list recoverable}.
\]
The exponent is uniform in $p,k$ and the evaluation
points.
\end{theorem}

\subsection{Fourier Analytic Tools}
\label{subsec:fourier}

We adopt the following conventions on Fourier analysis over $\F_p^n$. Let $ e_p(t):=\exp(2\pi i t/p)$. Then, for functions $f:\F_p^n\to\C$, 
\[
        \wh{f}(y):=\E_{x\in\F_p^n}
        \left[f(x)e_p(\langle x,y\rangle)\right],\qquad 
        f(x)=\sum_{y\in\F_p^n}\wh{f}(y)e_p(-\langle x,y\rangle).
\]

\begin{lemma}[Arc Bound; see \cite{us}]
\label{lem:indicator-L-infty-bound}
Let $S\subset\F_p$ have size $|S|=\rho p$.  Then, for every
$\xi\ne0$, as $p\to\infty$ uniformly in $\rho$, we have that
\[
\left|\widehat{\one_S}(\xi)\right|
\le
\frac{\sin(\pi\rho)}{p\sin(\pi/p)} = \frac{\sin(\pi\rho)}{\pi}+O(p^{-2}).
\]
\end{lemma}

\begin{corollary}[$L^\infty$-Fourier bound]
\label{cor:L-infty-bound}
For any $\rho\in(0,1)$,
\begin{equation}
\label{eq:L-infty-bound}
\left|\widehat g(\xi)\right|
\le
\gamma(\rho)+O(p^{-2})
\end{equation}
for every $\xi\ne 0$, where
\begin{equation}
\label{eq:def-gamma}
g(x):=\frac{\one_S(x)-\rho}{\sqrt{\rho(1-\rho)}}
\quad\text{and}\quad
\gamma(\rho):=
\frac{\sin(\pi\rho)}{\pi\sqrt{\rho(1-\rho)}}.
\end{equation}
Moreover, setting $\gamma(0)=\gamma(1)=0$ by continuity,
$\gamma:[0, 1]\to [0, \infty)$ is increasing on $(0,1/2]$, decreasing on
$[1/2,1)$, and $\gamma(1/2)={2}/{\pi}$.
\end{corollary}

\begin{proof}
For $\xi\ne 0$, if $h(x) = \rho$ is a constant function then $\hat{h}(\xi) = 0$. In general,  \cref{lem:indicator-L-infty-bound} gives
\[
\abs{\wh{g}(\xi)}\le \frac{\sin(\pi\rho)}{p\sin (\pi/p)\sqrt{\rho(1-\rho)}}=\frac{\gamma (\rho)\cdot \pi}{p\sin(\pi/p)}=\gamma(\rho)+O(p^{-2}),
\]
by linearity of the Fourier transform.
Now, monotonicity follows by differentiating.
\end{proof}

For later use in \cref{subsec:rrs}, we record an extension of \cref{lem:indicator-L-infty-bound}
to finite-dimensional vector space $V$ over $\F_p$. Similarly to \cite[Section~9.1]{HMSY26}, we define dual space $V^\star = \mathrm{Hom}_{\F_p}(V, \F_p)$ to be the space of $\F_p$-linear maps from $V$ to $\F_p$. We also use the convention that, for a function $f$ defined on $V$ and for $\xi \in V^\star$, \[ \widehat f(\xi)=\E_{x\in V}f(x)e_p(\xi(x)).\]

\begin{corollary}[Fractional arc bound]\label{cor:soft-arc}
Let $V$ be a finite-dimensional vector space over $\F_p$ and let $f:V\to[0,1]$
have mean $\theta$. Then, for every $\xi \in V^\star \setminus \{0\}$, as $p \to \infty$ uniformly in $V$ and $f$, we have that
\[
|\widehat f(\xi)|\le\frac{\sin(\pi\theta)}{p\sin(\pi/p)}
 =\frac{\sin(\pi\theta)}{\pi}+O(p^{-2}).
\]
Moreover, $\sum_{\xi\in V^\star}|\widehat f(\xi)|^2\le\theta$.
\end{corollary}
\begin{proof}
Fix $\xi\ne0$ and set $h(a)=\E[f(x)\mid\xi(x)=a]$.
Since $\xi$ is surjective, $h:\F_p\to[0,1]$ has mean $\theta$
and $\widehat h(1)=\widehat f(\xi)$. Write
$h(a)=\int_0^1\one_{S_t}(a)\,dt$, where
$S_t=\{a:h(a)\ge t\}$ and $\theta_t=|S_t|/p$.
By \cref{lem:indicator-L-infty-bound}, the triangle inequality,
and concavity of $\sin(\pi\theta)$,
\[
 |\widehat f(\xi)|
 \le\int_0^1|\widehat{\one_{S_t}}(1)|\,dt
 \le\frac{\int_0^1\sin(\pi\theta_t)\,dt}{p\sin(\pi/p)}
 \le\frac{\sin(\pi\theta)}{p\sin(\pi/p)},
\]
using $\int_0^1\theta_t\,dt=\theta$.
Finally, Parseval's identity and the fact that $f^2\le f$ give
$\sum_{\xi\in V^\star}|\widehat f(\xi)|^2=\E f^2\le \mb{E}f=\theta$.
\end{proof}

We introduce some higher order Gowers norms used for additive combinatorics over the finite field model. See \cite{gowers-wolf} for a more complete treatment.

\begin{definition}[Gowers norms over $\F_p$]
\label{def:gowers-norms}
For an integer $s\ge1$ and a function $F:\F_p\to\C$, define
\[
  \|F\|_{U^s}^{2^s}
  :=
  \E_{x,h_1,\ldots,h_s}
  \prod_{\omega\in\{0,1\}^s}
  \mathcal C^{|\omega|}F(x+\langle \omega, h\rangle),
\]
where $\mathcal C$ denotes complex conjugation and all variables are uniform
on $\F_p$.  In particular,
\begin{equation}\label{eq:fourierU}
        \|F\|_{U^1}=|\E F|=|\wh{F}(0)|,
        \qquad
        \|F\|_{U^2}^4
        =\sum_{\xi\in\F_p}|\wh F(\xi)|^4.
\end{equation}
\end{definition}
For $\Delta_hF(x):=F(x)\overline{F(x+h)}$, the Gowers norms satisfy
\begin{equation}\label{eq:gowers_recursion}
        \|F\|_{U^{s+1}}^{2^{s+1}}
        =\E_h\|\Delta_hF\|_{U^s}^{2^s}.
\end{equation}
We use this identity only for $s\in \{1, 2\}$.  For the applications below, we need a uniform bound on $U^3$-norms strictly
below one whenever the minority probability $\tau$ is bounded away from $0$.

\begin{lemma}[$U^3$-gap]
\label{lem:u3-gap}
For every fixed $\tau\in(0,1/2)$ there is $q<1$ such that the
following holds for all sufficiently large primes $p$. If
$f:\F_p\to\R$ satisfies that $|f(x)|-1=o(1)$ uniformly over $x\in\F_p$, and that $ |\E_{x\in\F_p} f(x)|\le 1-2\tau+o(1)$, then
$\|f\|_{U^3}\le q$.
\end{lemma}

\begin{proof}
We first claim that it suffices to consider $\pm1$-valued $f$ with no
$o(1)$ errors. Let $\widetilde f:\F_p\to\{-1,1\}$ be a nearest-sign
function, so $\eps
    =\max_{x\in\F_p}\bigl||f(x)|-1\bigr|=o(1)$.
Then, $|\E\widetilde f|  \le |\E f|+\eps \le 1-\tau$
for all sufficiently large $p$.  Applying the $\pm 1$-valued statement
with parameter $\tau/2$ gives
$\|\widetilde f\|_{U^3}\le q_0$ for some fixed $q_0<1$.  Hence,any fixed $q\in(q_0,1)$ works for all sufficiently large $p$ as
\[\|f\|_{U^3}
    \le \|\widetilde f\|_{U^3}
       +\|f-\widetilde f\|_{U^3}
    \le q_0+\eps.\] 

It therefore remains to prove the exact $\pm1$-valued case, i.e. $f:\F_p\to\{-1,1\}$ such that
$|\E f|\le1-2\tau$. Let $\mu:=\E f\in [2\tau-1, 1-2\tau]$.
Fix any $c\in(2/\pi,1)$. By
\cref{lem:indicator-L-infty-bound}, for all sufficiently large $p$, $\left|\wh f(\xi)\right|\le c$
whenever $\xi\ne0$.
Parseval and the Fourier formula in \cref{eq:fourierU} for the $U^2$ norm give
\begin{equation}
\label{eq:u2-gap-bd}  
\norm{f}_{U^2}^4 =\sum_{\xi\in\F_p}\left|\wh f(\xi)\right|^4 \le|\mu|^4+c^2\left(1-|\mu|^2\right).
\end{equation}

Since $f$ is $\pm1$-valued, let $f_h(x)=\Delta_h f(x)=f(x+h)f(x)$ for $h\in\F_p$ and let $m_h :=\E f_h$. Then, 
$f_h$ is again $\pm1$-valued, so $|m_h|^4\le |m_h|^2$ and the same argument as \cref{eq:u2-gap-bd} gives
\[
\|f_h\|_{U^2}^4
\le |m_h|^4+c^2\left(1-|m_h|^2\right)
\le c^2+\left(1-c^2\right)|m_h|^2.
\]
Moreover,
we have by \cref{eq:fourierU,eq:gowers_recursion} that
\[
     \|f\|_{U^2}^4= \E_h \|\Delta_h f\|_{U^1}^2
=\E_h |\E \Delta_h f|^2 = \E_h |m_h|^2,
\]
so $\E_h |m_h|^2 = \|f\|_{U^2}^4$.
Together, we have
\[
\|f\|_{U^3}^8 = \E_h \|f_h\|_{U^2}^4
\leq c^2 + (1 - c^2)\E_h|m_h|^2
= c^2 + (1 - c^2)\|f\|_{U^2}^4
\le c^2+\left(1-c^2\right)\left[|\mu|^4+c^2\left(1-|\mu|^2\right)\right],
\]
where the first equality follows from \cref{eq:gowers_recursion}.
Since $|\mu|\le 1-2\tau$, $|\mu|^4\le |\mu|^2\le (1-2\tau)^2$, so 
\[ \|f\|_{U^3}^8 \le c^2+(1-c^2)[c^2+(1-c^2)(1-2\tau)^2]=1-(1-c^2)^2[1-(1-2\tau)^2] <1\]
Therefore, choosing $q:=(1-(1-c^2)^2[1-(1-2\tau)^2])^{1/8}<1$ suffices.
\end{proof}

Next, we need an MDS specialization of the generalized von Neumann inequality. The standard repeated Cauchy--Schwarz argument is used in \cite{gowers-wolf,manners} to bound true complexity by Cauchy--Schwarz complexity. To prepare, we first show the following elementary Cauchy--Schwarz step.
\begin{lemma}
\label{lem:elem-cs}
For any $G, P:\F_p^k\to\C$,
suppose
$G$ is $K$-invariant for some subspace $K\le\F_p^k$, i.e. $G(x+h)=G(x)$ for every $x\in\F_p^k$ and $h\in K$, 
and $G$ is $1$-bounded, i.e. $\|G\|_\infty\le1$.  Then
\begin{equation}
    \left|\E_{x\in\F_p^k}G(x)P(x)\right|^2
   \le
    \E_{\substack{x\in\F_p^k\\h\in K}}
    P(x)\ol{P(x+h)}=: \E_{\substack{x\in\F_p^k\\h\in K}} \Delta_h P(x).
\label{eq:block-fiber-cs}
\end{equation}
\end{lemma}
\begin{proof}
For any $x\in\F_p^k$ and $h\in K$, $\E_x G(x)P(x+h)=\E_y G(y)P(y)$ for $y=x+h$ by $K$-invariance. Therefore, $\E_xG(x)P(x)=\E_x G(x)\E_{h\in K}P(x+h)$, so by Cauchy-Schwarz,
\[
 \left|\E_{x\in\F_p^k}G(x)P(x)\right|^2 \le \E_x |G(x)|^2\cdot \E_x \left|\E_hP(x+h)\right|^2\le \mb{E}_{x}\E_{h, h'\in K}P(x+h)\ol{P(x+h')}
\]
by $1$-boundedness of $G$. Now, setting $y=x+h\in\F_p^k$ and $w=h'-h\in K$ rewrites the right hand side as $\E_y P(y)\ol{P(y+w)}$, which is the desired bound up to renaming the variables. The final equality follows the definition of $\Delta_h$.
\end{proof}

We now state the MDS generalized von Neumann inequality and give a standard proof specialized to the case of MDS codes. We use the $K$-invariance and $1$-bounded notions in \cref{lem:elem-cs}.

\begin{restatable}[MDS generalized von Neumann inequality]
{lemma}{blockuthree}\label{lem:block-u3}
Let $1\le k\le n$, and let
$F_1,\ldots,F_n:\F_p^k\to\F_p$ be linear maps such that $(F_i(x))_{i\in S}:\F_p^k\to\F_p^S$
is an isomorphism for every
$S\in\binom{[n]}{k}$. Let $J\subset[n]$ satisfy
\begin{equation}\label{eq:block-u3-dimension}
    |J|\le\frac{3k-n}{2}.
\end{equation}
Then, for arbitrary functions $L_1,\ldots,L_n:\F_p\to\C$,
\begin{equation}\label{eq:block-u3}
    \abs{\E_{x\in\F_p^k}\prod_{i=1}^n L_i(F_i(x))}
    \le
    \left(\prod_{j\in J}\norm{L_j}_{U^3}\right)
    \left(\prod_{i\in[n]\setminus J}\norm{L_i}_\infty\right).
\end{equation}
The same conclusion holds if each $F_i(x)$ is replaced by
$F_i(x)+c_i$ for an arbitrary $c_i\in\F_p$.
\end{restatable}

\begin{proof}
First, it suffices to show the case where $F_i$ is linear without the affine shift by $c_i$ and that $L_i$ is $1$-bounded  for every $i\in[n]\setminus J$, i.e. $\norm{L_i}_\infty\le1$.
Indeed, if $\norm{L_i}_\infty=0$ for some $i\notin J$, then the
left-hand side of \cref{eq:block-u3} is zero. Otherwise, replace each
$L_i$ with $L_i/\norm{L_i}_\infty$ for $i\notin J$ and restore the
factor $\prod_{i\notin J}\norm{L_i}_\infty$ at the end to prove the general $L_i$ case. Now, suppose that $F_i(x)$ is replaced by $F_i(x)+c_i$. Apply the
linear case to $\widetilde L_i(t):=L_i(t+c_i)$. Observe that
translations preserve both relevant norms $ \norm{\widetilde L_i}_{U^3}
    =
    \norm{L_i}_{U^3}$ and $\norm{\widetilde L_i}_\infty
    =
    \norm{L_i}_\infty$,
which proves the translated statement.

Let $d:=|J|$. Since $n-d\le3(k-d)$, partition
$[n]\setminus J=I_1\sqcup I_2\sqcup I_3$ with $|I_a|\le k-d$.
Define
\[
    F_J(x):=(F_j(x))_{j\in J},\qquad
    H(z):=\prod_{j\in J}L_j(z_j),\qquad
    G_a(x):=\prod_{i\in I_a}L_i(F_i(x)),
\]
so $G_a$ is $1$-bounded since each $L_i$ for $i\not\in J$ is. The quantity whose norm we need to control is
\[
    \Lambda
    :=
    \E_{x\in\F_p^k}\prod_{i=1}^nL_i(F_i(x))
    =
    \E_{x\in\F_p^k}H(F_J(x))G_1(x)G_2(x)G_3(x).
\]

Fix any $a\in[3]$. Observe that $G_a$ is
$K_a$-invariant where $K_a:=\bigcap_{i\in I_a}\ker F_i$. We show that $F_J|_{K_a}:K_a\to\F_p^J$ is surjective. Indeed, coordinates of the map $(F_i)_{i\in J\cup I_a}:\F_p^k\to \F_p^{J}\times F_p^{I_a}$ are independent by the MDS assumption, so the above map is surjective. In particular, for every $z\in \F_p^J$, there exists $y\in \F_p^k$ such that $(F_i)_{i\in J\cup I_a}(y)=(z, 0)$. Hence, $F_J(y)=z$ and $y\in K_a$, as claimed.

Now, we apply \cref{lem:elem-cs} three times.
We will use the following simple observations: first, $\Delta_h(PQ)=\Delta_hP\Delta_hQ$; second, if $G$ is
$K$-invariant and $1$-bounded, then $\Delta_hG$ is also $K$-invariant
and $1$-bounded for every $h\in\F_p^k$. Indeed, for every $u\in K$,
\[
    \Delta_hG(x+u)
    =
    \ol{G(x+u+h)}{G(x+u)} =
    \ol{G(x+h)}{G(x)}
    =
    \Delta_hG(x),
\]
and
\[
    |\Delta_hG(x)|
    =
    |G(x+h)|\,|G(x)|
    \le 1.
\]
The same therefore holds for any iterated derivative of $G$.
Apply \cref{lem:elem-cs} first with $G=G_1$ and $P=(H\circ F_J)G_2G_3$.
Since $G_1$ is $K_1$-invariant and $1$-bounded, we obtain
\begin{equation}
\label{eq:block-cs-one}
    |\Lambda|^2
    \le \E_{\substack{x\in\F_p^k\\h_1\in K_1}} \Delta_{h_1}\bigl((H\circ F_J)G_2G_3\bigr)(x) =\E_{\substack{x\in\F_p^k\\h_1\in K_1}}
    \Delta_{h_1}(H\circ F_J)(x)
    \Delta_{h_1}G_2(x)
    \Delta_{h_1}G_3(x).
\end{equation}

Now, for each fixed $h_1\in K_1$, the function $\Delta_{h_1}G_2$ is
$K_2$-invariant and $1$-bounded. Applying \cref{lem:elem-cs} with
$G=\Delta_{h_1}G_2$
and $ P=
    \Delta_{h_1}(H\circ F_J)\Delta_{h_1}G_3$, and Cauchy-Schwarz over $h_1$ yields
\begin{equation}
\label{eq:block-cs-two}
    |\Lambda|^4
    \le
    \E_{\substack{x\in\F_p^k\\
                  h_1\in K_1,\ h_2\in K_2}}
    \Delta_{h_2}\Delta_{h_1}(H\circ F_J)(x)
    \Delta_{h_2}\Delta_{h_1}G_3(x).
\end{equation}

For each fixed $h_1,h_2$, the function
$\Delta_{h_2}\Delta_{h_1}G_3$ is $K_3$-invariant and $1$-bounded.
Applying \cref{lem:elem-cs} one final time and Cauchy-Schwarz over $(h_1, h_2)$ gives
\begin{equation}
\label{eq:block-three-cs}
    |\Lambda|^8
    \le
    \E_{\substack{x\in\F_p^k\\
                  h_i\in K_i,\ i\in[3]}}
    \Delta_{h_3}\Delta_{h_2}\Delta_{h_1}
    (H\circ F_J)(x).
\end{equation}

It remains to identify the right-hand side with the $U^3$-norm of
$H$.  Since $F_J:\F_p^k\to\F_p^J$ and each restriction
$F_J|_{K_i}:K_i\to\F_p^J$ are surjective linear maps, independent
uniform $x\in\F_p^k$ and $h_i\in K_i$ give independent uniform
variables in $\F_p^J$ that we denote by $ z:=F_J(x)$ and $ u_a:=F_J(h_a)$ for $a\in [3]$. Moreover, by linearity of $F_J$,
$\Delta_{h_3}\Delta_{h_2}\Delta_{h_1}
    (H\circ F_J)(x)=\Delta_{u_3}\Delta_{u_2}\Delta_{u_1}H(z)$.
Therefore, by the definition of the $U^3$-norm, the right-hand side of \cref{eq:block-three-cs} is
\begin{equation}\label{eq:block-rhs}
    \E_{z,u_1,u_2,u_3\in\F_p^J}
    \Delta_{u_3}\Delta_{u_2}\Delta_{u_1}H(z)
    =
    \norm{H}_{U^3(\F_p^J)}^8,
\end{equation}
Finally, since
the cube average factors over the coordinates $j\in J$, we see that
\[
    H(z)=\prod_{j\in J}L_j(z_j)\implies \norm{H}_{U^3(\F_p^J)}^8
    =
    \prod_{j\in J}\norm{L_j}_{U^3(\F_p)}^8.
\]
Combining with \cref{eq:block-three-cs,eq:block-rhs}, and taking eighth roots gives
$|\Lambda|
    \le
    \prod_{j\in J}\norm{L_j}_{U^3}$.
This proves the linear, $1$-bounded case. Applying the initial
reduction proves the general statement.
\end{proof}

For \cref{thm:random-opi}, we need the following higher-order version
of \cref{lem:block-u3}. Its proof is the same simultaneous
Cauchy--Schwarz argument, with $s$ blocks in place of three. For a complete proof, see \cite[Lemma~10.1]{HMSY26}.

\begin{fact}[Higher-order MDS inequality]\label{fact:block-ul}
Let $C\subset\F_p^n$ be an $[n,k]$ MDS code, let $v\in\F_p^n$, and
let $s\ge2$. If $J\subset[n]$ satisfies
$n-|J|\le s(k-|J|)$, then for all $1$-bounded functions
$L_i:\F_p\to\C$,
\[
 \left|\E_{x\in v+C}\prod_iL_i(x_i)\right|
 \le\prod_{i\in J}\|L_i\|_{U^s}.
\]
\end{fact}

\subsection{Computation Model and Subsampling}
\label{subsec:subsampling}
We describe the computation model for our  algorithmic OPI statements, e.g. \cref{thm:main-algo}. Recall that to specify the OPI instance, we need to specify the Reed-Solomon code $C=\on{RS}_p(\alpha_1, \dots, \alpha_n;k)$  as well as input lists $S_i$ with $|S_i|\ge \rho p$. We adopt the same model as \cite{yz,hy,sjo}.
\begin{definition}[Computation Model for OPI]
\label{def:input-model}    
The Reed-Solomon code $C$ is passed in via inputs
$(p,k,\alpha_1,\ldots,\alpha_n)$. We also give as input an integer $\rho p\in[p]$ satisfying $|S_i|\ge \rho p$ for every $i$.
The lists $S_i$ accessed through the coherent membership oracles: for $i\in[n], x\in\F_p, b\in\{0, 1\}$,
\[
    \ket{i}\ket{x}\ket{b}
    \longmapsto
    \ket{i}\ket{x}\ket{b\oplus\one_{S_i}(x)}
\]
The \emph{efficient} classical and quantum algorithms we consider here have running time $\poly(n, \log p)$, which is polynomial whenever $\log p \le \poly(n)$.
When $p\le \poly(n)$, explicit membership truth tables $\{\one_{S_i}(x):i\in [n], x\in \F_p\}$ suffice,
since look-ups are in polynomial time. When $p\ge n^{\omega(1)}$ but $\log p \le \poly(n)$, this is not efficient, so the lists are supplied through membership
oracles above.
\end{definition}

Most of our arguments for near-optimal OPI discrepancy and quantum algorithms first treat the case where $|S_i|=\rho p$ for every $i$, and where $\rho$ is known. In the rest of this section, we give the ingredients to extend such results to
arbitrary denser and heterogeneous input lists satisfying $|S_i|\ge \rho p$. Moreover, as discussed in \cref{rem:no-need-rho}, these ideas can remove the need for the algorithm to receive $\rho p$ as an input. 
Many of the observations below were made in \cite{hy}; we prove a slightly different reduction for membership-oracle input setting suitable to our algorithms.

The first observation is containment: since solutions are preserved if we enlarge the input lists, any existential satisfaction guarantee, including existence of a
perfect solution, transfers from lists of common size $\rho p$ to lists with $|S_i|\ge \rho p$.
In particular, a result proved in the balanced case (when $|S_i| = 1/2 + o(1)$ for all $i$) will apply to any case where $\min_i|S_i| \geq 1/2 + o(1)$.

To make this observation algorithmic, we can use \emph{algorithmic subsampling} to reduce the general case to the case where $|S_i|$ all have the same density.
The idea is to choose $\rho p$ elements of each input list and run a common-density algorithm.  If the output of this algorithm has satisfaction ratio $s'$ constraints, then it have satisfaction ratio $s \geq s'$ for the original lists.  This idea works immediately when the lists are given as truth tables, which requires $p = \poly(n)$.  If $p$ is larger and the lists are given via oracle access, then \cref{prop:oracle-thinning} below shows how to make this idea work.

Before we get to that, however, we note that the monoticity above (that is, that making the lists larger can only make a solution satisfy more constraints) does not by itself transfer the \emph{number} of
solutions, as required for near-optimal OPI
discrepancy. This is handled by the second subsampling idea.

\begin{proposition}
\label{prop:subsample}
Let $C\subset\F_p^n$, let $S_i\subset\F_p$ with density $\rho_i$, and let
$1\le t\le\min_i|S_i|$.  If the $T_i$ are independent uniformly random
$t$-subsets of $S_i$, then
\[
\E_T\left[\frac{|C\cap(T_1\times\cdots\times T_n)}
{|C|\left({t}/{p}\right)^n}\right]
=
\frac{|C\cap(S_1\times\cdots\times S_n)|}
{|C|\prod_i\rho_i}.
\]
Consequently, a uniform relative-discrepancy bound for input lists of size $t$ transfers without loss to every instance with $|S_i|\ge t$.
\end{proposition}

\begin{proof}
For every fixed
$c\in C\cap(S_1\times\cdots\times S_n)$, independence gives
\[
\P_T(c\in T_1\times\cdots\times T_n)
=\prod_{i=1}^n\frac{t}{|S_i|}.
\]
Summing this identity over $c$ and dividing by $|C|(t/p)^n$ proves the
displayed formula. The second bound follows by averaging.
\end{proof}

For coherent membership access, we can transfer an algorithm's satisfaction guarantee without enumerating the input lists. This is inspired by the oracle thinning procedure in \cite{hy,zhandry-oracle}.
\begin{proposition}
\label{prop:oracle-thinning}
Fix $\rho \in (0, 1)$ and $s\in[0,1]$. Suppose a $\poly(n,\log p)$ time quantum algorithm with coherent membership access achieves a satisfaction ratio at least $s$
with inverse-polynomial probability on every OPI instance with density
$\rho_i=\rho$.  Then there is a $\poly(n,\log p,\rho^{-1})$ time quantum algorithm in the same oracle model that achieves the same satisfaction ratio $s$ with high probability whenever $\rho_i\ge\rho$.
\end{proposition}
\begin{proof}
Let $A$ be the given algorithm on fixed density $\rho$, so it has at least inverse-polynomial success probability  $\sigma\in(0,1]$ and polynomial runtime with a polynomial $Q$ many queries of the membership oracles for the input lists when they have size exactly $\rho p$ each. We impose these same time and query bounds on every execution, aborting if they are exceeded, so that $A$ is also defined on inputs outside its density promise. Let 
\[
p_0 := \max\left\{\frac{16}{\delta}, \frac{32n}{\sigma\delta^2}\right\}\quad\text{where}\quad \delta:=\frac{\rho\sigma^2}{64Q^2}.
\]
If $p\le p_0$, then $p=\poly(n)$, so we can simply compute the truth table $(\mbf{1}\{a\in S_i\})_{i\in [n], a\in\F_p}$ with polynomially many queries, arbitrarily throw out elements $S_i$ to obtain $U_i$ with $|U_i|=\rho p$, and run $A$ on $(U_1, \dots, U_n)$. Hence, assume $p\ge p_0$ for the rest of the proof.

First, we efficiently implement membership oracles for
subsets $T_i\subset S_i$ such that for every $i\in [n]$:
\begin{equation}
\label{eq:thinning-size-window}
    \rho\le\frac{|T_i|}{p}\le\rho+\delta.
\end{equation}
For each coordinate $i$, independently draw
\[
    M:=
    \left\lceil
      \frac{128}{\delta^2}
      \log\frac{16n}{\sigma}
    \right\rceil
\]
uniform field elements $X_{i,1},\ldots,X_{i,M}$, and query
their memberships in $S_i$. Define
\[
    \widetilde\rho_i
    :=
    \frac1M\sum_{\ell=1}^M
       \mathbf1\{X_{i,\ell}\in S_i\}\quad\text{and}\quad \widehat\rho_i := \max(\rho,  \widetilde\rho_i).
\]
For fixed $i$, since $\rho_i\in[\rho,1]$, clipping cannot
increase the estimation error. By Chernoff's inequality
\[
    \mb{P}\left(
       |\widehat\rho_i-\rho_i|>\frac{\delta}{16}
    \right)\le \mb{P}\left(
       |\widetilde\rho_i-\rho_i|>\frac{\delta}{16}
    \right)
\le
    2\exp\left(
       -2M\left(\frac{\delta}{16}\right)^2
    \right)\le
    2\exp\left(-\log\frac{16n}{\sigma}\right)
     =\frac{\sigma}{8n}.
\]
Union bounding over $i$, the event $\mc{G}$ that $|\widehat\rho_i -\rho_i|\le \delta/16$ occurs with probability at least $1-\sigma/8$. Let
\[
    \theta_i:=
    \min\left\{1,\frac{\rho+\delta/2}{\widehat\rho_i}\right\}\quad.
\]
Independently sample uniform $a_i,b_i\in\mathbb F_p$ and define $T_i\subset S_i$ via
\[
    T_i:=
    \left\{
       x\in S_i:
  a_ix+b_i\in \{0, \dots, h_i-1\}
    \right\}\quad\text{where}\quad     h_i:=\lfloor p\theta_i\rfloor.
\]
On event $\mc{G}$, we claim that
\eqref{eq:thinning-size-window} holds with high probability.
If $\theta_i=1$, then $\widehat{\rho_i}\le \rho+\delta/2 $ and $T_i=S_i$, so
\[
    \rho\le\frac{|T_i|}{p}=\rho_i
    \le\widehat\rho_i+\frac{\delta}{16}
    \le\rho+\frac{9\delta}{16}
    <\rho+\delta,
\]
as claimed. If $\theta_i<1$, then let $Z_i=|T_i|/p$. Since $a_ix+b_i$ is uniform on
$\mathbb F_p$ for each $x$, $\mb{E}[Z_i]=\rho_i h_i/p$.
Since $ (a_i,b_i)\mapsto(a_ix+b_i, a_iy+b_i)$
is a bijection of $\mathbb F_p^2$ for distinct $x,y$,
\[
    \on{Var}(Z_i)\le \frac{|S_i|}{4p^2}
     \le\frac1{4p}.
\]

Since $h_i=\lfloor p\theta_i\rfloor$
and $\theta_i\widehat\rho_i=\rho+\delta/2$, 
\[
    \left|\mb{E}[Z_i]-\left(\rho+\frac{\delta}{2}\right)\right| = \left|\frac{\rho_i h_i}{p}-\left(\rho+\frac{\delta}{2}\right)\right|
    \le \theta_i|\rho_i-\widehat\rho_i|+\frac1p
    \le\frac{\delta}{8},
\]
where the last step uses $p\ge p_0$.
Hence, on event $\mc{G}$, Chebyshev gives
\[
    \mb{P}\left(Z_i\not\in[\rho,\rho+\delta]\middle| \mc{G}\right)
    \le
    \mb{P}\left(\left|Z_i-\frac{\rho_i h_i}{p}\right|
       >\frac{\delta}{4}\middle| \mc{G}\right)
    \le\frac4{p\delta^2}.
\]
Union bounding over $i\in [n]$ and combining with $p\ge p_0$ and the bound on $\mc{G}$, we see that
\begin{equation}
\label{eq:thinning-good-prob}
    \mb{P}\left(
      \frac{|T_i|}{p}\in[\rho,\rho+\delta]\,\forall\, i\in [n]
    \right)
    \ge
    1-\frac{\sigma}{4}.
\end{equation}

Note that algorithmically, we never enumerate $T_i$. We simply sample and store $(a_i, b_i, h_i)$. To simulate membership oracles for $T_i$, we query that for $S_i$, compute $\one\{(a_ix+b_i)\in \{0, \dots, h_i\}\}$, and compute the AND of the two bits coherently. This costs polynomial time overhead and allows us to run the algorithm $A$ on input lists $(T_1, \dots T_n)$.
On the high probability event in \cref{eq:thinning-good-prob}, these $T_i$ to have density in $[\rho, \rho+\delta]$. However, they are not yet all equal to $\rho$, so we cannot yet apply the guarantee of satisfaction ratio at least $s$ with probability at least $\sigma$. Nevertheless, we run $A$ on $T_i$ and we show for the rest of the proof why the output has satisfaction ratio at least $s$ with probability at least $\sigma/2$, upon which we boost to high probability with polynomially many repetitions since $\sigma$ is at least inverse polynomial.

Fix any $(T_1, \dots, T_n)$ and condition on the high probability event in \cref{eq:thinning-good-prob}. For analysis only, independently choose $U_i$ uniformly
among the $\rho p$-element subsets of $T_i$.
For every $(i,x)$,
\[
    \mb{P}_U\left(
      \one_{T_i}(x)\ne\one_{U_i}(x)
    \right)
    =
\frac{|T_i|-\rho p}{|T_i|}  \one_{T_i}(x) \le\frac{\delta p}{\rho p}
     =\frac{\delta}{\rho}.
\]
We argue that running $A$ with corresponding membership oracles $O_T$ and $O_U$ produce close-by quantum states.
Write its state just before query $j$ in the execution using $T$
as
\[
    \ket{\psi_j^T}
    =
    \sum_{i,x}\ket{i,x}\ket{\omega_{j,i,x}},
    \qquad
    \sum_{i,x}\|\omega_{j,i,x}\|^2=1.
\]
Let $\Delta:= \{(i, x): \one_{T_i}(x)\ne \one_{U_i}(x)\}$ 
and let $\Pi_{\Delta}$ project onto $\on{span}\{\ket{i,x}:(i,x)\in \Delta\}$ which depends on the random $U$, so
\[
\mb{E}_U
       \|\Pi_{\Delta}\ket{\psi_j^T}\|^2
    =
    \sum_{i,x}\|\omega_{j,i,x}\|^2
       \mb{P}_U\bigl((i,x)\in\Delta\bigr)
    \le
    \frac{\delta}{\rho}
    \sum_{i,x}\|\omega_{j,i,x}\|^2
     =\frac{\delta}{\rho}.
\]
Now, let us consider applying queries $O_T$ and $O_U$ to state $\ket{\psi_j^T}$, the outcome states overlap outside $\Delta$, and on $\Delta$ the difference $O_T-O_U$ has operator norm at most two, so by Cauchy-Schwarz,
\[
    \mb{E}_U
       \|(O_T-O_U)\ket{\psi_j^T}\|\le \mb{E}_U \|\Pi_{\Delta_U}\ket{\psi_j^T}\|
    \le
    2\sqrt{
       \mb{E}_U
       \|\Pi_{\Delta_U}\ket{\psi_j^T}\|^2
    }
    \le 2\sqrt{\frac{\delta}{\rho}}.
\]

Coupling the two executions, of $A$ with $O_U$ and $O_T$, we bound the error by telescope summing $j$ over the at most $Q$ queries, so the final outcome states $|\Psi_T\rangle,|\Psi_U\rangle$ satisfy
\[
    \mb{E}_U
       \bigl\||\Psi_T\rangle-|\Psi_U\rangle\bigr\|
    \le2Q\sqrt{\frac{\delta}{\rho}} = \frac{\sigma}{4},
\]
Every tuple $U$ has exact density $\rho$ and  $U_i\subset S_i$. By the hypothesis on $A$, the outcome of measuring $\ket{\Psi_U}$ has satisfaction ratio at least $s$ with probability at least $\sigma$ for every $U$, since it holds under the fixed event that the output satisfies at least $s$ of the original constraints $S$. Since probability changes by at most the distance between the two output states, the above display shows that measuring $\ket{\Psi_T}$ has satisfaction ratio at least $s$ with probability at least $3\sigma/4$, conditioned on the event in \cref{eq:thinning-good-prob}. Together, we succeed on $S$ with probability
\[
    \left(1-\frac{\sigma}{4}\right)\frac{3\sigma}{4}
    =
    \frac{3\sigma}{4}-\frac{3\sigma^2}{16}
    \ge\frac{9\sigma}{16},
\]

Now, we amplify this to high probability by repetition, evaluating each output's satisfaction on the original lists $S_i$ with $n$ membership queries and returning the best output. In total, the running time and query complexity overhead is polynomial in $(n,\log p,\rho^{-1},\sigma^{-1},Q)$. Since $\sigma^{-1}$ and $Q$ are polynomial in $n$ and $\log p$, we proved the desired proposition. In particular, $U$ is a proof device, not part of the algorithm.
\end{proof}

Thus our guarantees $s(Q)\ge\operatorname{SCL}_\rho(R)-\eps$ and
$s(Q)=1$ when $R+\rho>1$ extend to unequal list densities
$\rho_i\ge\rho$, with no additional loss in satisfaction. 

\begin{remark}
\label{rem:no-need-rho}
Treating the list size lower bound $\rho p$ as a known input in \cref{def:input-model} is only a convenient convention.
If it is unknown, use $\widetilde{O}(\delta^{-2})$ membership queries per coordinate to estimate every $\rho_i=|S_i|/p$ by $\wh{\rho_i}$ to additive error $\delta$ as in the proof of \cref{prop:oracle-thinning}. Then, apply \cref{prop:oracle-thinning} with
$\underline\rho = \min_i \wh{\rho_i}-\delta$ suffices with high probability. This preprocessing costs
polynomial overhead and our applications have
arbitrarily small constant slack in the density parameter, so choosing
$\delta$ sufficiently small absorbs the resulting change into the stated
accuracy or strict parameter gap. Thus the density lower bound $\rho$ need not
actually be supplied.
\end{remark}

\section{Fourier Proxies and Thresholding}
\label{sec:proxies}
In this section, we present variants of the Fourier proxies we need to bound LLR and OPI and discuss the implications among them. To emphasize the connections and distinctions with previous proxies in \cite{bdir,hy,kk23}, we state them for arbitrary coordinate densities $\rho_i :=|S_i|/n$. This generality is needed by the above-half OPI theorem and the uniform one-bit LLR theorem.

We record the exact centered proxies and their positive, or
``phase-forgotten'', relaxations. We usually fix $C$ and $\mbf{S}:=S_1\times\dots\times S_n$ (and hence the functions $g_i$), and suppress them from notation. For sets
$S_i\subset\F_p$ of densities $\rho_i\in(0,1)$, define
\begin{equation}
\label{eq:centered-local-functions}
        g_i(x):=
        \frac{\one_{S_i}(x)-\rho_i}{\sqrt{\rho_i(1-\rho_i)}},
        \qquad
        \lambda_i:=\sqrt{\frac{1-\rho_i}{\rho_i}}.
\end{equation}
By centering in \cref{eq:def-gi} and Parseval's identity, $\E g_i=\wh{g_i}(0)=0$ and $\|g_i\|_2=1$. We also frequently make use of \cref{cor:L-infty-bound}, which says that for any $\xi \neq 0$, $|\widehat{g_i}(\xi)| \leq \gamma(\rho_i) + O(p^{-2})$ where $\gamma(\rho_i) = \frac{\sin(\pi \rho_i)}{\pi\sqrt{\rho_i(1-\rho_i)}}$.

\subsection{The OPI Discrepancy Proxies}
\label{subsec:opi-proxy}

We begin by proving a slightly generalized version of \cref{eq:FO} %
that introduces the OPI proxy.
\begin{lemma}[Translated OPI Proxy]
\label{lem:translated-opi-proxy}
If $S_i$ has density $\rho_i$ and $v\in\F_p^n$,  then
\begin{equation}\label{eq:translated-EO}
\FO(v)=\FO(v;C,\mbf{S}):=\left|
\frac{|(v+C)\cap \mbf{S}|}
{|C|\prod_i\rho_i}-1
\right|=\left|
\sum_{0\ne y\in C^\perp}
e_p(-\langle v,y\rangle)
\prod_{i\in\supp(y)}
\lambda_i\wh g_i(y_i)
\right|.
\end{equation}
\end{lemma}
\begin{proof}
Fourier inversion and orthogonality on the coset $v+C$ give
\begin{equation}\label{eq:l31}
\frac{|(v+C)\cap \mbf{S}|}{|C|}
=
\sum_{y\in C^\perp}
e_p(-\langle v,y\rangle)
\prod_{i=1}^n\wh{\one_{S_i}}(y_i).
\end{equation}
Now, $\wh{\one_{S_i}}(0)=\rho_i$, and for $\xi\ne0$,
we know by \cref{eq:centered-local-functions} that
\[ \wh{\one_{S_i}}(\xi)
=
\sqrt{\rho_i(1-\rho_i)}\wh{g_i}(\xi).\]
The $y=0$ term in \cref{eq:l31} is $\prod_i\rho_i$.  Subtracting this term and dividing
by $\prod_i\rho_i$ proves \cref{eq:translated-EO}.
\end{proof}
By the triangle inequality, we know that $\FO(v)\le \FO^+$ where we recall the positive relaxation $\FO^+$ obtained by forgetting the phase term:
\begin{equation}
\label{eq:lambda-EO+}
\FO^+=\FO^+(C,\mbf{S})
:=
\sum_{0\ne y\in C^\perp}
\prod_{i\in\supp(y)}
\lambda_i
\abs{\wh g_i(y_i)}
\end{equation}
Note that $\FO^+$ is an $L^1$ proxy. This proxy was introduced in \cite{bdir} and the same exact and positive proxies were
used for OPI in \cite{us,hy}.

For later use in \cref{subsec:rrs}, we also record
the same proxy for fractional lists on an arbitrary linear
code $\widetilde C \subset \prod_{i=1}^m V_i$, so that the $i$'th symbol of $\widetilde{C}$ lies in some vector space $V_i$ over $\F_p$.  
For $v \in \prod_i V_i$ and for $y \in \prod_i V_i^\star$,
write
$\langle v,y\rangle=\sum_i y_i(v_i)$. %

\begin{corollary}[Fractional OPI proxy]
\label{cor:soft-opi-proxy}
Let $\widetilde C\subset \prod_{i=1}^n V_i$ be any $\F_p$-linear code of length $n$, where each
$V_i$ is a finite-dimensional vector space over $\F_p$. For $f_i:V_i\to[0,1]$
with mean $\rho_i>0$, let
\[
 \FO(v;\widetilde C,\mbf f)
 :=\left|\frac{\E_{x\in v+\widetilde C}\prod_i f_i(x_i)}{\prod_i\rho_i}-1\right|\quad\text{and}\quad \FO^+(\widetilde C;\mbf f)
 :=\sum_{0\ne y\in\widetilde C^\perp}
       \prod_i\frac{|\widehat f_i(y_i)|}{\rho_i} 
\]
for any  $v\in\prod_i V_i$ and $\mbf{f}=(f_1, \dots, f_n)$. Then
\begin{equation}\label{eq:soft-opi-proxy}
 \FO(v;\widetilde C,\mbf f)
 =\left|\sum_{0\ne y\in\widetilde C^\perp}e_p(-\langle v,y\rangle)
             \prod_i\frac{\widehat f_i(y_i)}{\rho_i}\right| \le\FO^+(\widetilde C;\mbf f).
\end{equation}
\end{corollary}
\begin{proof}
The proof of \cref{lem:translated-opi-proxy} applies verbatim:
Fourier inversion and orthogonality on $v+\widetilde C$ give the
expansion, and its zero-frequency term is $\prod_i\rho_i$.
Subtracting this term, dividing by $\prod_i\rho_i$, and applying the
triangle inequality yields the desired result.
\end{proof}
For indicator functions on $\F_p$, this recovers the Fourier proxy from \cref{lem:translated-opi-proxy}. More
generally, for $0<\rho_i<1$, retain
$g_i=(f_i-\rho_i)/\sqrt{\rho_i(1-\rho_i)}$ and the same $\lambda_i$;
then $\|g_i\|_2\le1$ and
$|\widehat f_i(\xi)|/\rho_i=\lambda_i|\widehat g_i(\xi)|$
for $\xi\ne0$.

\subsection{The One-Bit LLR Proxies}
\label{subsec:llr-proxy}
For LLR, fix one-bit leakages $\ell_i$, put
$S_i:=\ell_i^{-1}(1)$, and let $\rho_i:=|S_i|/p$. By deleting constant leakages, we can assume without loss of generality $\rho_i\in (0, 1)$ for every $i$.
Let $C$ be any linear code in $\F_p^n$, and recall from \cref{def:llr} that $L_v$ is the distribution of transcripts when the leakage functions $\ell_1, \ldots, \ell_n$ are applied to a uniformly random codeword in $v + C$. 
Recall the LLR proxy and its relaxation from \cref{eq:FL,eq:FL+}:
\begin{equation}
\begin{aligned}
\FL(v,v')=\FL(v,v';C,\mbf S)
&:=\frac14
\sum_{\emptyset\ne A\subset[n]}
\abs{
\sum_{\substack{y\in C^\perp\\\supp(y)=A}}
\left[e_p(-\langle v,y\rangle)-e_p(-\langle v',y\rangle)\right]
\prod_{i\in A}\wh g_i(y_i)
}^2,\\
\FL^+=\FL^+(C,\mbf{S})&:=
\sum_{\emptyset\ne A\subset[n]}
\left(
\sum_{\substack{y\in C^\perp\\\supp(y)=A}}
\prod_{i\in A}\abs{\wh g_i(y_i)}
\right)^2.
\end{aligned}
\end{equation}
We now prove \cref{eq:llr-proxy}.

\begin{lemma}[Refined LLR Proxy]
\label{lem:restated-llr-proxy}
For every $v,v'\in\F_p^n$, we have
\begin{equation}
\label{eq:llr-TwA-form}
\TV(L_v,L_{v'})^2\le\FL(v,v')\le\FL^+.
\end{equation}
In particular,
$\TV(L_{s\one},L_{s'\one})
\le\sqrt{\FL(s\one,s'\one)}\le\sqrt{\FL^+}$.
\end{lemma}
 The key idea is to Fourier expand around the product Bernoulli
distribution
\begin{equation}
\label{eq:prod-ber}
\nu:=\bigotimes_{i=1}^n\on{Ber}(\rho_i),
\qquad
\nu_i(z_i):=\rho_i^{z_i}(1-\rho_i)^{1-z_i},
\end{equation}
using the orthonormal functions $\{\psi_A:A\subset [n]\}$ in $L^2(\nu)$ given by
\begin{equation}
\label{eq:bernoulli-basis}
\psi_A(z):=\prod_{i\in A}\psi_i(z_i)\quad\text{where}\quad
\psi_i(z_i):=\frac{z_i-\rho_i}{\sqrt{\rho_i(1-\rho_i)}}.
\end{equation}
\begin{proof}[Proof of \cref{lem:restated-llr-proxy}]
For every $x\in\F_p^n$ and $z\in\{0,1\}^n$, we can check the identity
\[
\frac{\one\{\ell_i(x_i)=z_i\}}{\nu_i(z_i)}
=
1+\psi_i(z_i)g_i(x_i),
\]
for every $i\in[n]$.
Consequently,
\begin{equation}
\label{eq:llr-likelihood-expansion}
\frac{\P(L_v=z)}{\nu(z)}=
\E_{x\in v+C}\left[
\prod_{i=1}^n
\left(1+\psi_i(z_i)g_i(x_i)\right)\right]=
\sum_{A\subset[n]}
\psi_A(z)
\E_{x\in v+C}\left[\prod_{i\in A}g_i(x_i)\right].
\end{equation}
Applying Fourier inversion on $v + C$, we see that
\begin{equation}
\label{eq:llr-code-moment-expansion}
\E_{x\in v+C}\left[\prod_{i\in A}g_i(x_i)\right]=
\sum_{\substack{y\in C^\perp\\\supp(y)=A}}
e_p\left(-\langle v,y\rangle\right)
\prod_{i\in A}\wh g_i(y_i),
\end{equation}
where above we have used the fact that $\wh g_i(0) = 0$ for all $i$.
Now, subtracting \cref{eq:llr-likelihood-expansion} for $v$ and $v'$ and
using \cref{eq:llr-code-moment-expansion} gives
\begin{equation}
\label{eq:llr-proxy-key}
\frac{\P(L_v=z)-\P(L_{v'}=z)}{\nu(z)}
=
\sum_{A\subset[n]}
\psi_A(z)
\sum_{\substack{y\in C^\perp\\\supp(y)=A}}
\left[e_p(-\langle v,y\rangle)-e_p(-\langle v',y\rangle)\right]
\prod_{i\in A}\wh g_i(y_i).
\end{equation}
Therefore,
\begin{equation}
\label{eq:ortho-llr-proxy}
\begin{aligned}
4\TV(L_v,L_{v'})^2
&=
\left(
\E_{Z\sim\nu}
\left|
\frac{\P(L_v=Z)-\P(L_{v'}=Z)}{\nu(Z)}
\right|
\right)^2\\
&\le
\E_{Z\sim\nu}
\left|
\frac{\P(L_v=Z)-\P(L_{v'}=Z)}{\nu(Z)}
\right|^2\\
&=
\E_{Z\sim\nu}\left|
\sum_{A\subset[n]}\psi_A(Z)
\sum_{\substack{y\in C^\perp\\\supp(y)=A}}
\left[e_p(-\langle v,y\rangle)-e_p(-\langle v',y\rangle)\right]
\prod_{i\in A}\wh g_i(y_i)
\right|^2\\
&=
4\FL(v, v')
\end{aligned}
\end{equation}
where we use Cauchy--Schwarz and orthonormality of $\psi_A$. The empty-support contribution cancels. The bound $\FL(v,v')\le\FL^+$ follows from
the triangle inequality.
\end{proof}

We recall the following lemma from~\cite{kk23}, with normalized $L^2$ and Fourier
conventions as in \cref{subsec:fourier}. {\cite[Lemma 5.4]{kk23}} states the lemma with fixed choice $\gamma=2/\pi$, while we replace it by some parameter and observe the proof only requires $\gamma <1$.

\begin{theorem}[{\cite[Lemma 5.4]{kk23}}]
\label{thm:kk-black-box}
Let $C\subset\F_p^n$ be the secret-zero Shamir code from \cref{eq:secret-zero-code} of dimension $k$,
and assume
$k+1\ge n/2$.  Fix $\gamma\in(0,1)$ and for $i \in [n]$, let $u_i:\F_p\to\R$ satisfy
\[
        \E u_i=0,
        \qquad
        \|u_i\|_2\le1,
        \qquad
        \max_{\xi\ne0}|\wh u_i(\xi)|\le \gamma.
\]
Then, for every $\omega\ne0$, we have that
\[
   \sum_{A\subset[n]}
   \abs{\sum_{\substack{y\in C^\perp\\\supp(y)=A\\\langle y,\one\rangle=\omega}}
   \prod_{i\in A}\wh u_i(y_i)}^2
   \le
   \gamma^{2(k+1)}(1+\gamma^{-2})^{n-k-1}.
\]
\end{theorem}

For one-bit leakages, applying \cref{thm:kk-black-box,cor:L-infty-bound} to the centered functions $g_i$ and summing
over $\omega\ne 0$ immediately gives the following bias-dependent improvement to the benchmark in \cite{kk23} for one-bit LLR of Shamir schemes.
While this is not our strongest theorem on one-bit LLR of Shamir schemes (see \cref{thm:main-llr,sec:llr}), we include it to demonstrate the motivation for centering the proxy this way. Recall that by complementing leaked bits if necessary, we can assume density $\rho_i\ge 1/2$ for all one-bit leakage functions. 

\begin{corollary}[Bias-dependent improvement of \cite{kk23}]
\label{cor:bias-aware-kk}
Fix $\tau\in[1/2,1)$ and assume that $p>n$ is an odd prime and
$\log p\ll n$. Suppose that $\rho_i\ge\tau$ for all $i$, and define
\begin{equation}
\label{eq:bias-aware-kk-rate}
\RK(\tau)
:=
\frac{\log(1+\gamma(\tau)^{-2})}
{\log(1+\gamma(\tau)^{-2})-2\log\gamma(\tau)}.
\end{equation}
Every Shamir scheme of fixed rate $R>\RK(\tau)$ is one-bit locally
leakage resilient against every such leakage tuple. 
\end{corollary}
\begin{proof}
For the centered functions $g_i$, using that $\gamma$ is decreasing on
$[1/2,1)$, we have
\[
    \E g_i=0,\qquad \E|g_i|^2=1,
    \qquad
    \max_{\xi\ne0}|\wh g_i(\xi)|
    \le \gamma(\rho_i)+O(p^{-2})
    \le \gamma_n:=\gamma(\tau)+O(p^{-2})<1.
\]
Since $R>\RK(\tau)>1/2$,
\cref{thm:kk-black-box} holds for sufficiently large $n$.
Partitioning \cref{eq:FL} according to
$\omega=\langle y,\one\rangle$, the $\omega=0$ contribution vanishes,
and Cauchy--Schwarz gives
{\small \[
    \FL(s\one,s'\one;C,\mbf S)
    =
    \frac14\sum_A
    \left|
      \sum_{\omega\ne0}
      \bigl(e_p(-s\omega)-e_p(-s'\omega)\bigr)
     \sum_{\substack{y\in C^\perp\\
                    \supp(y)=A\\
                    \langle y,\one\rangle=\omega}}
    \prod_{i\in A}\wh g_i(y_i)
    \right|^2\le
    (p-1)\sum_{\omega\ne0}\sum_A \left|\sum_{\substack{y\in C^\perp\\
                    \supp(y)=A\\
                    \langle y,\one\rangle=\omega}}
    \prod_{i\in A}\wh g_i(y_i)\right|^2.\]}
Applying
\cref{thm:kk-black-box} separately for each $\omega\ne0$ and then summing
over the $p-1$ frequencies gives
\[
      \FL(s\one,s'\one;C,\mbf S)  \le
    (p-1)^2\gamma_n^{2(k+1)}
    \bigl(1+\gamma_n^{-2}\bigr)^{n-k-1} \le  e^{-\Omega(n)}
\]
uniformly over $s,s'\in\F_p$. Note that the rate in the exponent is strictly negative precisely when $R>\RK(\tau)$, with the fixed gap absorbing the
$o(1)$ error terms.
Finally, \cref{lem:restated-llr-proxy} and the above displayed equation shows $ \TV(L_{s\one},L_{s'\one})$ is exponentially small,
proving one-bit LLR.
\end{proof}
\begin{remark}\label{rem:bias-aware-kk}
At $\tau=1/2$, \cref{eq:bias-aware-kk-rate} is exactly the threshold
$\RK(1/2)=0.579\ldots$ 
that was established by \cite{kk23} for balanced leakage functions.
The bias-aware centering in 
\cref{eq:FL}
removes the need for a reduction to pass from the balanced to the unbalanced case,
and \cref{fig:R-KK-benchmark} demonstrates the resulting improvement in the unbalanced case.
\end{remark}
\begin{figure}[H]
\centering

%

\providecolor{KKPurple}{HTML}{7B2CBF}
\providecommand{\RK}{R_{\mathrm{KK}}}

\pgfmathdeclarefunction{biasgamma}{1}{%
  \pgfmathparse{sin(deg(pi*#1))/(pi*sqrt(#1*(1-#1)))}%
}
\pgfmathdeclarefunction{biasrate}{1}{%
  \pgfmathparse{
    ln(1+1/(biasgamma(#1)*biasgamma(#1)))
    /
    (ln(1+1/(biasgamma(#1)*biasgamma(#1)))-2*ln(biasgamma(#1)))
  }%
}

\begin{tikzpicture}
\begin{axis}[
  width=8.1cm,
  height=5.8cm,
  scale only axis,
  axis lines=box,
  axis line style={black!60},
  tick align=outside,
  tick style={black!60},
  grid=major,
  major grid style={black!10},
  xmin=0.5,
  xmax=1,
  ymin=0.498,
  ymax=0.582,
  xtick={0.5,0.6,0.7,0.8,0.9,1},
  ytick={0.50,0.52,0.54,0.56,0.58},
  xlabel={$\tau$},
  ylabel={$\RK(\tau)$},
  label style={font=\small},
  tick label style={font=\footnotesize},
  scaled ticks=false,
  /pgf/number format/fixed,
  /pgf/number format/precision=2,
  clip=false,
]
  \addplot[black!45,densely dashed,line width=.7pt,domain=.5:1]
    {0.5};

  \addplot[KKPurple,solid,line width=1.6pt,domain=.5:.9995,samples=300]
    {biasrate(x)};

  \addplot[KKPurple,only marks,mark=*,mark size=1.8pt]
    coordinates {(0.5,0.57925197399)};
  \addplot[KKPurple,only marks,mark=o,mark size=2.0pt,thick]
    coordinates {(1,0.5)};
\end{axis}
\end{tikzpicture}
\caption{The improvement over \cite{kk23} from
\cref{cor:bias-aware-kk}. The rate $\RK(\tau)$ required for one-bit LLR, where $\tau\le \max(\rho_i, 1-\rho_i)$ for all $i$. In particular, $\RK(1/2) = 0.579\ldots$ is the worst (largest) bound in the range, and coincides with the result of \cite{kk23} when $\tau=1/2$; however, $\RK(\tau)$ decreases as $\tau > 1/2$ increases, improving on \cite{kk23} (which obtained \emph{worse} results for $\tau>1/2$).}
\label{fig:R-KK-benchmark}
\end{figure}
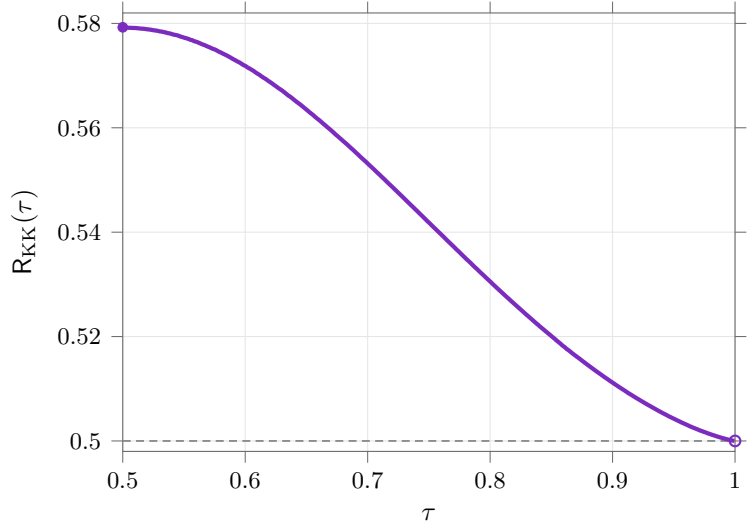

\subsection{Implications Between OPI and LLR}
\label{subsec:opi-harder}
In this section, we extend the techniques that derived our Fourier proxies to clarify the relation between them beyond the trivial relaxation bounds. The main message, as stated in \cref{thm:informal-opi-harder}, is that near-optimal OPI discrepancy is harder than one-bit LLR. More precisely, we elaborate on 
the three levels at which this main message is realized.  The first statement below is the statement of \cref{thm:informal-opi-harder}.

\begin{theorem}[OPI is harder than one-bit LLR]
\label{thm:opi-harder}
For any linear code $C\subset \F_p^n$ and a one-bit leakage tuple $(\ell_1,\ldots,\ell_n)$, let $\mbf{S}:=\ell_1^{-1}(1)\times\dots \times \ell_n^{-1}(1)$ and $\mbf{S}_z=\ell_1^{-1}(z_1)\times\dots\times \ell_n^{-1}(z_n)$ for $z\in\{0, 1\}^n$. Then,

\begin{enumerate}
    \item Suppose that, for all $s \in \F_p$, and for all $z \in \{0,1\}^n$, $s\one + C$ has near-optimal OPI discrepancy for input lists $S_i = \ell_i^{-1}(z_i)$.  Then $C$ is $1$-bit LLR.
    \item Suppose that, for all $v \in \F_p^n$, and for all $z \in \{0,1\}^n$,  $v+C$ has  near-optimal OPI discrepancy for input lists $S_i = \ell_i^{-1}(z_i)$. Then $\TV(L_v,L_{v'})\le e^{-\Omega(n)}$ uniformly over $v,v'\in\F_p^n$.
\item At the level of the Fourier proxies and the phase-forgotten Fourier proxies,
\begin{equation}
\label{eq:proxy-reductions}
\begin{aligned}
\sqrt{\FL(v,v'; C, \mbf{S})}
& \le\max_{\substack{u\in\{v,v'\}\\z\in\{0, 1\}^n}}
\FO(u;C,\mbf{S}_z)\\
\sqrt{\FL^+( C, \mbf{S})} &
\le\max_{z\in\{0, 1\}^n}\FO^+(C, \mbf{S}_z).
\end{aligned}
\end{equation}
\end{enumerate}

\end{theorem}

\begin{proof}
We first recall the product Bernoulli distribution $\nu$ as well as the orthonormal basis $\{\psi_A:A\subset [n]\}$ of $L^2(\nu)$ from \cref{eq:prod-ber,eq:bernoulli-basis}.
Then to prove (2), note that
$\FO(v;C,\mbf S_z)
=\abs{\frac{\P(L_v=z)}{\nu(z)}-1}$ for $z\in\supp(\nu)$.
By the triangle inequality,
\begin{equation}
\label{eq:opi-harder-1}
\TV(L_v,L_{v'})
\le\frac12\sum_{z\in\supp(\nu)}\nu(z)
\left(\FO(v;C,\mbf S_z)+
\FO(v';C,\mbf S_z)\right)\le
\max_{\substack{u\in\{v,v'\}\\z\in\supp(\nu)}}
\FO(u;C,\mbf S_z).
\end{equation}
The proof of (1) is similar, replacing $v,v'$ with $s\one$, $s'\one$ for $s \neq s' \in \F_p$.

 To prove (3), recall from \cref{eq:ortho-llr-proxy} that 
\[
\FL(v,v')=\frac14\E_{Z\sim\nu}
\abs{\frac{\P(L_v=Z)-\P(L_{v'}=Z)}{\nu(Z)}}^2\le\frac12\E_{Z\sim\nu}
\left(\FO(v;C,\mbf S_Z)^2+
\FO(v';C,\mbf S_Z)^2\right),
\]
so the first inequality of \cref{eq:proxy-reductions} follows by bounding the expectations with the maximum. For the second inequality, observe that by definition of $\FO^+(C;\mbf{S}_z)$ and $\psi_i$, we have that the positive OPI proxy for $\mbf S_z$ is
\[
\FO^+(C,\mbf S_z)
=
\sum_{\emptyset\ne A\subset[n]}
\left(\prod_{i\in A}|\psi_i(z_i)|\right)
\sum_{\substack{y\in C^\perp\\\supp(y)=A}}
\prod_{i\in A}\abs{\wh g_i(y_i)}.
\]
Choose $z^\star\in\supp(\nu)$ coordinatewise so that
$|\psi_i(z_i^\star)|=\max\{\lambda_i,\lambda_i^{-1}\}\ge1$.
For this instance $\mbf{S}_{z^\star}$,
\[
\FO^+(C,\mbf S_{z^\star})
\ge
\sum_{\emptyset\ne A\subset[n]}
\sum_{\substack{y\in C^\perp\\\supp(y)=A}}
\prod_{i\in A}\abs{\wh g_i(y_i)}
\ge
\left(
\sum_{\emptyset\ne A\subset[n]}
\left(
\sum_{\substack{y\in C^\perp\\\supp(y)=A}}
\prod_{i\in A}\abs{\wh g_i(y_i)}
\right)^2
\right)^{1/2}
=\sqrt{\FL^+(C;\mbf S)}.
\]
this proves the second inequality of \cref{eq:proxy-reductions} upon observing $\FO^+(C,\mbf S_{z^\star})\le \max_z \FO^+(C,\mbf S_{z})$.
\end{proof}

The converse implications of \cref{thm:opi-harder} can fail by an exponential factor even for
Shamir scheme aka fixed-rate GRS codes, as shown in the following example.
The gap arises because LLR averages deviations over transcripts, whereas
uniform OPI discrepancy demands relative control on every individual
transcript, including exponentially rare ones.
\begin{example}[The converse to \cref{thm:opi-harder} is not true]
\label{ex:opi-llr-strict}
Let
$C\subset\F_p^n$ be an $[n,k]$ MDS code with $k=Rn$, where
$R\in(0,1)$ is fixed and $p\ge n$, and corresponding input lists and leakage functions via $ S_i:=\F_p^\times$ and $\ell_i(a):=\mathbf 1\{a\ne0\}$.
For $z\in\{0,1\}^n$, uniformly over $v,v'\in\F_p^n$, 
\begin{equation}
\label{eq:strict-1} 
|\mb{P}(L_v=z)-\mb{P}(L_{v'}=z)|\le 2^{n+1} p^{-k}
\end{equation}
To see this, write
\[ \one\{L_v=z\} = \prod_{i:z_i=0}\one\{(L_v)_i=0\}\prod_{i:z_i\ne 0}\left(1-\one\{(L_v)_i=0\}\right) = \sum_{T} (-1)^{|T\cap \on{supp}z|} \one\{L_v\vert_T=0\},\]
summed over all $T\subset [n]$ that contains $\{i:z_i=0\}$. Now, $L_v\vert_T=0$ imposes $|T|$ zeros on coordinates of $v+c$ for $c\sim\on{Unif}(C)$. If $|T|\le k$, $\mb{P}(L_v\vert_T=0)$ is independent of $v$ by the MDS property of $C$, so the contributions to \cref{eq:strict-1} cancels for $v$ and $v'$. If $|T|>k$, there is at most one $c\in C$ with the prescribed zeros, so $\mb{P}(L_v\vert_T=0)\le p^{-k}$. There are at most $2^{n}$ such $T$, which gives the bound in \cref{eq:strict-1}. 
Hence, we can crudely bound $\TV(L_v,L_{v'})\le 2^{2n}p^{-k}$.
Moreover, since $\nu(z)\ge p^{-n}$ for every transcript $z$,
\[
    \FL(v,v';C,\mbf S)
   \le
    \frac{p^n}{4}
    \left(\sum_z|\mb{P}(L_v=z)-\mb{P}(L_{v'}=z)|\right)^2
    =
    p^n\TV(L_v,L_{v'})^2.
\]
Therefore, $\FL(v,v';C,\mbf S)= e^{-\Omega(n\log n)}$ by \cref{lem:restated-llr-proxy},
whenever $k=Rn$, $R>1/2$, and $p\ge n$.
Nevertheless, for $z=0^n$ we have $S_z=\{0\}^n$ and
$C\cap S_z=\{0\}$, so
\[
    \FO(0;C,S_z)
    =
    \left|\frac{p^{-k}}{p^{-n}}-1\right|
    =
    p^{n-k}-1.
\]
Thus both $\TV(L_v,L_{v'})$ and the LLR proxy can be
superexponentially small even though the worst-case OPI
discrepancy is superexponentially large.
\end{example}

\section{Local Leakage Resilience}
\label{sec:llr}

In this section, we prove the first item of the following theorem.

\mainllr*

\subsection{Many-bit LLR Above Rate 1/2}
\label{subsec:llr-above}
To prove the \cref{thm:main-llr}(1), we use the ``tests'' of \cite{mpsw21}.
These tests will bound the leakage resilience error $\eps$, i.e. the maximum total variation of the transcript distributions for distinct secrets. We state the result below for the Shamir scheme only, whereas \cite{mpsw21} proves it for more general Massey schemes corresponding to replacing the generalized Reed-Solomon code $C$ with any MDS code of equal length and rate.

We state these tests. Fix
$\sigma\in[0,1]$, $\ell\in\N$, and $a>1$. Every test is indexed by a
tuple $\boldsymbol V=(V_1,\ldots,V_n)$ of $\ell$-element subsets of
$\F_p$ and a set $J\subset[n]$ of size $(1-\sigma)n$. A word
$y\in\F_p^n$ fails the test $(\boldsymbol V,J)$ if
$y_j\in V_j$ for every $j\in J$. The code $C^\perp$ fails the test $(\mathbf{V}, J)$ if
at least $a^n$ of the elements of $C^\perp$ fail it.  We say that $C^\perp$ passes
$\on{TEST}_{C^\perp}(\sigma,\ell,a)$ if it fails none of these tests  $(\mathbf{V}, J)$. 

\begin{lemma}[Test Criterion {\cite[Lemma~1]{mpsw21}}, specialized to Shamir sharing]
\label{lem:mpsw-fixed-support}
Fix $b\in\N$. There is a constant $C_b>1$ such that the following
holds. Let $C\subset\F_p^n$ be the secret-zero share code of a Shamir scheme of rate $R$,
and choose $0<\sigma\le 2R-1$, $a>1$, and positive integers $T$ and $\ell=2^bT^2$.
If $C^\perp$ passes $\on{TEST}_{C^\perp}(\sigma,\ell,a)$, then the
scheme is $b$-bit locally leakage resilient with leakage error
\[
\varepsilon\le 
2^{-[R\log_2 C_b-\log_2 a-H_2(\sigma)]n}+2^{-[\sigma\log_2 T-(\sigma+1-R)b-H_2(\sigma)]n}.
\]
\end{lemma}

We observe that list recovery (see \cref{def:list-recovery}) of $C^\perp$ certifies
every one of these tests
simultaneously.  Indeed, if $C^\perp$ is
$(\sigma,\ell,L)$-list recoverable
 with $L<a^n$, every word failing
a test agrees with its list tuple on all $(1-\sigma)n$ coordinates of
$J$. Hence at most $L<a^n$ words fail any test.  This observation means that we can combine \cref{lem:mpsw-fixed-support} with recent breakthroughs in the list-recovery of Reed--Solomon codes (\cref{thm:rs-lr}) to prove part (1) of \cref{thm:main-llr}, which we do next.

\begin{proof}[Proof of \cref{thm:main-llr}(1)]
Fix the constant $C_b$ from \cref{lem:mpsw-fixed-support}. First choose
$0<\sigma<2R-1$ so that
$H_2(\sigma)<R\log_2 C_b$. Next choose $a>1$ sufficiently close to
one that
\[
R\log_2 C_b-\log_2a-H_2(\sigma)>0,
\]
and finally choose $T$ sufficiently large that
\[
\sigma\log_2T-(\sigma+1-R)b-H_2(\sigma)>0.
\]
Put $\ell:=2^bT^2$. Our choice $0<\sigma<2R-1<R$ gives agreement $1-\sigma>1-R$, strictly above the dual rate, so
\cref{thm:rs-lr} implies that $C^\perp$ is
$(\sigma,\ell,L)$-list recoverable with
$L=n^{O_{R,b}(1)}\ll a^n$ as $n\to\infty$.
Hence $C^\perp$ passes
$\on{TEST}_{C^\perp}(\sigma,\ell,a)$, and
the two exponents in the square brackets in \cref{lem:mpsw-fixed-support} are positive
constants. This proves exponentially small leakage error.
\end{proof}

\subsection{Weighted Fourier Thresholding}
\label{subsec:wftl}
In this section, we isolate the key mechanism by which we leverage list recoverability of RS codes (see \cref{thm:rs-lr}) in the proofs of \cref{thm:main-llr,thm:main-opi,thm:main-algo,thm:main-existential}. Recall the Fourier proxies $\FO^+$ and $\FL^+$ from \cref{subsec:llr-proxy,subsec:opi-proxy} that we wish to upper bound: Their bounds reduce to weighted sums of the form
\begin{equation}\label{eq:form}
\sum_{0\ne y\in C^\perp}\prod_{i:y_i\ne 0}w_i(y_i)
\end{equation}
where $w_i$ is the norm of some Fourier coefficients of $f_i$ or $g_i$. The lemma below converts $L^2$-norm bounds (via Parseval's identity) and pointwise bounds of these Fourier coefficients into an exponentially small bound on the quantity in \cref{eq:form}, whenever $C$ has good list recovery properties.

\wftl*

\begin{proof}
First, we claim that it suffices to show the case of exact norm bounds, i.e. without the uniform
$\varepsilon=o(1)$ error. Scale  $\widetilde w_i:=\lambda^{-1}w_i$ so that $ \|\widetilde w_i\|_\infty\le\gamma$ and $ \|\widetilde w_i\|_2^2\le\Gamma$ for every $i$, via
\[
    \lambda:=
    \max\left\{
      \frac{\gamma+\varepsilon_n}{\gamma},
      \sqrt{\frac{\Gamma+\varepsilon_n}{\Gamma}}
    \right\}
    =1+o(1).
\]
Then, since $|y|\le n$, the result for $w_i$ follows from that of $\widetilde w_i$ as
\[
\sum_{\substack{0\ne y\in C^\perp\\|y|\le t}}
\prod_{i\in\supp(y)}w_i(y_i)
\le
\lambda^n
\sum_{\substack{0\ne y\in C^\perp\\|y|\le t}}
\prod_{i\in\supp(y)}\widetilde w_i(y_i),
\]
and
$\lambda^n=e^{o(n)}$. Therefore assume throughout the proof that $ \|w_i\|_\infty\le\gamma$ and $ \|w_i\|_2^2\le\Gamma$ for every $i$.

Now, If $t\le k$, then the sum is empty by the minimum distance $k+1$ of the dual code $C^\perp$.  
For $t> k$,  by assumption, $\sigma:=\lfloor k-t/2\rfloor /n$ is bounded away from $0$. Choose a constant
$\theta\in(0,\gamma)$ such that
\begin{equation}
\label{eq:quant-llr-small}
3(1+\Gamma) \theta^\sigma\le \gamma.
\end{equation}
For each coordinate $i \in [n]$, define the \emph{heavy list} $H_i
:=
\{0\}\cup
\left\{\xi\ne0:w_i(\xi)>\theta\right\}$.
Markov's inequality implies
 that $\abs{H_i}\le1+\Gamma\theta^{-2}$.  Since
$\sigma<R-t/(2n) = R-\Omega(1)$, \cref{thm:rs-lr} shows that only
$L=n^{O(1)}$ elements of $C^\perp$ agree with the lists $H_i$ on at least
$(1-\sigma)n$ coordinates.
Write $J(y):=\{i:y_i\in H_i\}$ and split weighted sum as
\begin{equation}
\label{eq:low-high-gi}
\sum_{\substack{0\ne y\in C^\perp\\|y|\le t}}
\prod_{i\in\supp(y)}w_i(y_i)
=
\sum_{\substack{0\ne y\in C^\perp\\|y|\le t\\
|J(y)|\ge(1-\sigma)n}}
\prod_{i\in\supp(y)}w_i(y_i)
\\
+
\sum_{\substack{0\ne y\in C^\perp\\|y|\le t\\
|J(y)|<(1-\sigma)n}}
\prod_{i\in\supp(y)}w_i(y_i).
\end{equation}
We bound the two summands separately.
\begin{itemize}
\item For nonzero $y\in C^\perp$ with $\abs{J(y)}\ge (1-\sigma)n$, as $\sigma < R$, list recovery (\cref{thm:rs-lr}) of $C^\perp$ of distance $k+1$, and the pointwise weight bound give
\[
\sum_{\substack{0\ne y\in C^\perp\\|y|\le t\\
|J(y)|\ge(1-\sigma)n}}
\prod_{i\in\supp(y)}w_i(y_i)
\le
L\gamma^{k+1} = n^{O(1)}\gamma^{k}.
\]

\item For nonzero $y\in C^\perp$ with
$\abs{J(y)}<(1-\sigma)n$, because $0\in H_i$,
every coordinate $i$ outside $J(y)$ belongs to $\supp(y)$.  For any such $i \not\in J(y),$ the weight $w_i(y_i)$ 
 is at most $\theta$ by the definitions of $H_i$ and $J(y)$.
Since $|y|\le t$, choose $T\in\binom{[n]}t$ containing
$\supp(y)$. Counting by $T$ and
the possible sets $M\in\binom{T}{\sigma n}$ of $\sigma n$ coordinates outside $J(y)$
gives
\[
\sum_{\substack{0\ne y\in C^\perp\\
|y|\le t,\ |J(y)|<(1-\sigma)n}}
\prod_{i\in\supp(y)}w_i(y_i)
\le
\sum_{T\in\binom{[n]}t}
\sum_{M\in\binom{T}{\sigma n}}
\sum_{\substack{y\in C^\perp\\
\supp(y)\subset T\\
y_i\notin H_i\ \forall i\in M}}
\prod_{i\in\supp(y)}w_i(y_i).
\]
Fix $T$ and one such $M$.
Since $2k-t>0$, we have
\[
    \sigma n=\left\lfloor k-\frac t2\right\rfloor
    \le k-\frac t2
    \le2k-t,
\]
and so $|T\setminus M|=t-\sigma n\ge2(t-k)$.
Hence, $T\setminus M$ contains two
disjoint sets $J_1,J_2$ of size $t-k$.
Extracting the $\theta^{|M|}$ contribution on $M$
and applying Cauchy--Schwarz gives
\begin{align*}
\sum_{\substack{y\in C^\perp\\
\supp(y)\subset T\\
y_i\notin H_i\ \forall i\in M}}
\prod_{i\in\supp(y)}w_i(y_i)
&\le
\theta^{|M|}
\left(
\sum_{\substack{y\in C^\perp\\\supp(y)\subset T}}
\prod_{i\in J_1\cap\supp(y)}w_i(y_i)^2
\right)^{1/2}
\left(
\sum_{\substack{y\in C^\perp\\\supp(y)\subset T}}
\prod_{i\in J_2\cap\supp(y)}w_i(y_i)^2
\right)^{1/2}
\\
&=
\theta^{\sigma n}
\left(
\prod_{i\in J_1}
\left[1+\sum_{a\ne0}w_i(a)^2\right]
\right)^{1/2}
\left(
\prod_{i\in J_2}
\left[1+\sum_{a\ne0}w_i(a)^2\right]
\right)^{1/2}
\\
&\le
\theta^{\sigma n}(1+\Gamma)^{t-k},
\end{align*}
where the crucial observation is that the shortening of
$C^\perp$ to $T$ is an $[t,t-k]$ MDS code, so projection
onto either $J_1$ or $J_2$ is a bijection.
The number of choices of $(T, M)$ is at most $3^n$ as $M\subset T\subset [n]$, so by our choice \cref{eq:quant-llr-small}, we see that
\[
\sum_{\substack{0\ne y\in C^\perp\\
|y|\le t,\ |J(y)|<(1-\sigma)n}}
\prod_{i\in\supp(y)}w_i(y_i)\le 
3^n
\theta^{\sigma n}(1+\Gamma)^{t-k}
\le
\left[3(1+\Gamma)\theta^\sigma\right]^n
\le\gamma^n.
\]
\end{itemize}
Combining the two estimates in the bullet points proves the lemma.
\end{proof}
\begin{remark}
\label{rem:wftl-param}
We remark on the requirement in \cref{lem:wftl} that $2k -t = \Omega(n)$, and its implications.  First, this requirement comes from the fact that
the proof of \cref{lem:wftl} requires $\sigma <R$ in order to apply \cref{thm:rs-lr}, and requires $\sigma n\le2k-t$  so that $|T \setminus M|$ is sufficiently large; the existence of such a $\sigma$ is guaranteed when $2k-t = \Omega(n)$, hence the hypothesis.  

As for why this matters, it means that if we want to apply \cref{lem:wftl} for every $t\in [n]$, including $t=n$---which, for example, is what we do in the proof of a more quantitative statement of~\cref{thm:main-llr}(1) below, and in the proof of \cref{thm:main-opi}(1)---we will then need $2k-n\gtrsim n$, i.e. rate $R>1/2$.  This is one of the two apparent barriers at rate $1/2$ discussed in \cref{sec:intro} that we will need to circumvent for our results with $R < 1/2$.
\end{remark}

As a corollary of \cref{lem:wftl}, we show a quantitative version of \cref{thm:main-llr}(1). Although the proof above using the tests in \cite{mpsw21} already gives the full arbitrary-$b$
bit LLR above one half, we next prove a quantitative one-bit version of \cref{thm:main-llr}(1) for two reasons.
\begin{itemize}
    \item First, the proof via \cref{lem:wftl} gives a self-contained illustration of why dual list recovery is the natural property and why \cref{lem:wftl} is the key mechanism of it we will repeatedly use. We note that this mechanism is implicit in \cref{lem:mpsw-fixed-support}.
\item Second, we will need the quantitative result in \cref{cor:one-bit-above} as an input for the conditioning argument to extend our results to arbitrary biases in \cref{sec:app-llr}.
\end{itemize}

Recall from \cref{def:llr} the leakage functions $\ell_1, \dots, \ell_n:\F_p\to\{0, 1\}$ and  their
transcript $L_v$ on a uniformly random word of $v+C$.
\begin{corollary}
\label{cor:one-bit-above}
Fix $R>1/2$. Let $C$ be an $[n,k]$ GRS code over a prime
field $\F_p$, with $p\ge n$ and $k/n\ge R$.
 Then, uniformly over the leakage functions $\ell_1, \ldots, \ell_n$  and $v,v'\in\F_p^n$,
\[
\TV(L_v,L_{v'})\le\sqrt{\FL(v,v')}\le\sqrt{\FL^+}
\le n^{O_R(1)}\left(\frac{2}{\pi}\right)^k.
\]
\end{corollary}
\begin{proof}
By \cref{lem:restated-llr-proxy} and the definition of $\FL^+$
\[
\TV(L_v,L_{v'})
\le\sqrt{\FL(v,v')}\le\sqrt{\FL^+}
\le\sum_{0\ne y\in C^\perp}
\prod_{i\in\supp(y)}|\wh g_i(y_i)|.
\]
Apply \cref{lem:wftl} to $w_i(\xi)=|\wh g_i(\xi)|$ with parameters $t=n$, $\Gamma=1$, and
$\gamma=2/\pi+O(p^{-2})$ from \cref{cor:L-infty-bound}.
Since $2k-n\ge(2R-1)n$, we obtain the bound
\[
\sum_{0\ne y\in C^\perp}
\prod_{i\in\supp(y)}|\wh g_i(y_i)|
\le n^{O_R(1)}\left(\frac{2}{\pi}+O(p^{-2})\right)^k
\le n^{O_R(1)}\left(\frac{2}{\pi}\right)^k.
\]
Combining the two equations above gives the corollary.
\end{proof}

\subsection{One-bit LLR Below Rate 1/2}
\label{subsec:llr-below}
In this section, we prove \cref{thm:main-llr}(2). We begin by proving it in the balanced case where $\rho_i=\rho\sim 1/2$, so $\|g_i\|_\infty=1+o(1)$ uniformly. The analogous statement for the arbitrary biased case is given in \cref{sec:app-llr}.
Fix an $[n,k]$ GRS code $C$ of rate $R\in[1/2-\delta,1/2]$ where $\delta\in (0, 1/6)$ is some small constant chosen later. Our goal is to bound the Fourier proxy $\FL(v,v')$, which amounts to bounding expressions of the form 
\begin{equation}\label{eq:tobound1} \left| \sum_{\substack{y \in C^\perp \\ \supp(y) = A}} e_p(-\langle v,y \rangle)\prod_{i \in A}\widehat{g_i}(y_i) \right|.\end{equation} We will bound this term in two ways, depending on how the size of $A$ compares with the cutoff 
$m=\lfloor(2-\delta)k\rfloor$ from \cref{eq:puncture-size}; notice that with this value of $m$, it follows from $1/2 - \delta \leq R$ that
\[
m\ge\left(1-\frac52\delta+\delta^2\right)n-O(1).
\]
When $|A| > m$, then we will bound \cref{eq:tobound1} using \cref{lem:high-support-energy}, which we state and prove below; if $|A| \leq m$, we will ignore the phases and use \cref{lem:wftl}.  
For sufficiently small $\delta$, Stirling's formula bounds the number of supports above this cutoff by
\begin{equation}
\label{eq:high-support-count}
\abs{\left\{A\subset[n]:|A|>m\right\}}
\le
2^{\left[H_2\left(\frac{n-m}{n}\right)+o(1)\right]n}
\le
2^{\left[H_2\left(\frac52\delta-\delta^2\right)+o(1)\right]n}.
\end{equation}
For the supports $A$ above the cutoff, we use \cref{lem:u3-gap,lem:block-u3} to obtain the following bound.
\begin{lemma}
\label{lem:high-support-energy}
There exists an absolute constant $q\in (0, 1)$ such that for every $v\in\F_p^n$, every balanced leakage tuple and its centered functions $g_i$, and every $A\subset [n]$ with $|A|>m$,

\begin{equation}
\label{eq:high-support-energy}
\abs{
\sum_{\substack{y\in C^\perp\\\supp(y)=A}}
e_p(-\langle v,y\rangle)
\prod_{i\in A}\wh g_i(y_i)
}
=\abs{\E_{x\in v+C}\left[\prod_{i\in A}g_i(x_i)\right]}
\le 2^{o(n)}q^{(3k-n)/2}.
\end{equation}
\end{lemma}

\begin{proof}
The first equality is Fourier inversion, since $\wh g_i(0)=0$.
In the balanced case, $\rho_i\sim 1/2$ so centered $g_i$ has mean zero and takes values $\pm 1+o(1)$. Thus,
\cref{lem:u3-gap} gives constant $q\in (0, 1)$ such that
$\|g_i\|_{U^3}\le q$ 
for every $i$.
Let $d:=\lfloor(3k-n)/2\rfloor$.
Since $\delta\in (0, 1/6)$ and $k\le n/2$, we see that $ 3k-n\ge(1/2-3\delta)n>0$ and $  0\le d \le k/2\le m<|A|$.
Thus we may choose set $J\subset A$ of size $d$.

We next apply \cref{lem:block-u3}, where the functions $F_i$ map a message $x$ to the $i$-th symbol of its encoding in $v+C$, and where $L_i=g_i$ for $i\in A$ and $L_i=1$ for $i\not\in A$, to obtain
\[
\abs{\E_{x\in v+C}\prod_{i\in A}g_i(x_i)}
\le
\prod_{i\in A\setminus J}\|g_i\|_\infty
\prod_{j\in J}\|g_j\|_{U^3}
\le
2^{o(n)}q^d
=
2^{o(n)}q^{(3k-n)/2},
\]
where various $o(1)$ error at each coordinate is absorbed into the $2^{o(n)}$ term.
\end{proof}

To show \cref{thm:main-llr}(2), it suffices by \cref{lem:restated-llr-proxy} to bound the corresponding Fourier proxy.

\begin{lemma}
\label{lem:balanced-llr-below}
There exists $\delta>0$ such that for every
$[n,k]$ generalized Reed-Solomon code $C$ of rate $R\ge1/2-\delta$ over a prime
field $\F_p$ with $p\ge n$, every balanced one-bit
leakage tuple, and every $v, v'\in \F_p^n$,
\begin{equation}
\label{eq:balanced-llr-proxy-below}
\FL(v, v')
\le e^{-\Omega(n)}.
\end{equation}
\end{lemma}
\begin{proof}
For $R>1/2$, this is true by \cref{cor:one-bit-above}. We assume $R\in [1/2-\delta, 1/2]$, where we choose $\delta>0$
sufficiently small that
\begin{equation}
\label{eq:llr-below-condition}
H_2\left(\frac52\delta-\delta^2\right)
+\left(\frac12-3\delta\right)\log_2q<0.
\end{equation}
Such a choice exists because the left-hand side is strictly
negative at $\delta=0$.
Use the cutoff $m=\lfloor(2-\delta)k\rfloor$ from
\cref{eq:puncture-size}. Forgetting phases 
and using $\abs{a-b}^2\le2\abs a^2+2\abs b^2$, we have
\begin{equation}
\label{eq:llr-below-low-high}
\FL(v,v')\le 
\sum_{0<|A|\le m}
\left(
\sum_{\substack{y\in C^\perp\\\supp(y)=A}}
\prod_{i\in A}\abs{\wh g_i(y_i)}
\right)^2+
\frac12\sum_{u\in\{v,v'\}}\sum_{|A|>m}
\abs{
\sum_{\substack{y\in C^\perp\\\supp(y)=A}}
e_p(-\langle u,y\rangle)
\prod_{i\in A}\wh g_i(y_i)
}^2.
\end{equation}
We bound the two summands separately.
\begin{itemize}
\item For the first term, apply \cref{lem:wftl} with $w_i=|\wh{g_i}|$,
 $t=m$, and $\gamma =2/\pi$ via \cref{cor:L-infty-bound}, and $\Gamma=1$ via Parseval's identity. We check that $2k-m\sim \delta k\asymp n$, so
\begin{equation}
     \sum_{0<|A|\le m}
\left(
\sum_{\substack{y\in C^\perp\\\supp(y)=A}}
\prod_{i\in A}\abs{\wh g_i(y_i)}
\right)^2\le
\left(
\sum_{\substack{0\ne y\in C^\perp\\|y|\le m}}
\prod_{i\in\supp(y)}\abs{\wh g_i(y_i)}
\right)^2\le
n^{O(1)}\left(\frac{2}{\pi}\right)^{2k}
=e^{-\Omega(n)}.
\end{equation}
\item For the second term, we apply \cref{lem:high-support-energy} to each term and count the number of $A$ via \cref{eq:high-support-count}:
\begin{equation}
\begin{aligned}
\frac12\sum_{u\in\{v,v'\}}\sum_{|A|>m}
\abs{
\sum_{\substack{y\in C^\perp\\\supp(y)=A}}
e_p(-\langle u,y\rangle)
\prod_{i\in A}\wh g_i(y_i)
}^2 & \le |\{A\subset [n]:|A|>m\}| \cdot 2^{o(n)}q^{3k-n}
\\ & \le 2^{\left[
H_2\left(\frac52\delta-\delta^2\right)
+\left(\frac12-3\delta\right)\log_2q+o(1)
\right]n}
\\ & =2^{-\Omega(n)}.
\end{aligned}
\end{equation}
where the last inequality uses $R\ge1/2-\delta$, $\log_2q<0$, and \cref{eq:llr-below-condition}.
\end{itemize}
Combining the two estimates proves \cref{lem:balanced-llr-below}.
\end{proof}

\begin{remark}[Where was balance needed?]
We remark that balance is used only for the large-support estimate: it makes the centered functions nearly $1$-bounded while retaining a fixed $U^3$ gap. The low-support estimate already holds for arbitrary biases.
\end{remark}

By \cref{lem:restated-llr-proxy}, \cref{lem:balanced-llr-below} implies that
$\TV(L_v,L_{v'})\le e^{-\Omega(n)}$ uniformly over the translates $v,v'$.
Taking $v=s\one$ and $v'=s'\one$ proves the balanced case of
\cref{thm:main-llr}(2).
The arbitrary-bias extension is proved in
\cref{lem:arbitrary-llr-below}. We now collect the two parts of the
main theorem on LLR, formally completing the proof of \cref{thm:main-llr}.

\begin{proof}[Proof of \cref{thm:main-llr}]
Part~(1) is proved in \cref{subsec:llr-above} by combining
\cref{lem:mpsw-fixed-support} with \cref{thm:rs-lr}.
For part~(2), \cref{lem:arbitrary-llr-below} bounds
$\TV(L_v,L_{v'})$ exponentially for every one-bit leakage tuple.
Taking $v=s\one$ and $v'=s'\one$ proves the theorem.
\end{proof}

\section{Near-Optimal OPI Discrepancy}
\label{sec:opi}

In this section, we prove \cref{thm:main-opi,thm:random-opi}. We first re-state the former for the reader's convenience.

\mainopi*

We prove \cref{thm:main-opi} in its three regimes. Above rate 1/2, weighted Fourier thresholding controls the full positive proxy. Near rate 1/2, the balanced proof uses the same low-support and high-support split as \cref{subsec:llr-below}. The regime where $R > 1-\rho$ is treated separately in \cref{subsec:opi-dense}.

\subsection{OPI Above and Near Rate 1/2}
\label{subsec:opi-above}
In this section, we prove the first two cases of \cref{thm:main-opi}, similar to the corresponding proofs of one-bit LLR in \cref{sec:llr}.
For the first conclusion, the Fourier weights in \cref{eq:lambda-EO+} have bounded squared $\ell^2$ mass and a pointwise bound strictly below one whenever the densities have a fixed positive lower bound, so \cref{lem:wftl} immediately gives exponential decay, much like the proof of \cref{cor:one-bit-above}.
\begin{lemma}
\label{lem:positive-opi-above}
For every $[n,k]$ generalized Reed--Solomon code $C$ of fixed rate
$R>1/2$ over a prime field $\F_p$ with $p\ge n$, and every OPI
instance $\mbf S$ of fixed density $\rho\in(0,1)$,
\[
\FO^+(C;\mbf S)
\le n^{O(1)}
\left(\frac{\sin(\pi\rho)}{\pi\rho}\right)^k
=e^{-\Omega(n)}.
\]
\end{lemma}
\begin{proof}
By Parseval and \cref{cor:L-infty-bound}, with
$\lambda=\sqrt{(1-\rho)/\rho}$,
\[
\sum_{\xi\ne0}\left(\lambda|\wh g_i(\xi)|\right)^2
=\frac{1-\rho}{\rho},
\qquad
\lambda|\wh g_i(\xi)|
\le\frac{\sin(\pi\rho)}{\pi\rho}+O_\rho(p^{-2}).
\]
Since $\sin(\pi\rho)/(\pi\rho)<1$, apply \cref{lem:wftl}
to $w_i(\xi)=\lambda|\wh g_i(\xi)|$ with $t=n$,
$\Gamma=(1-\rho)/\rho$, and
$\gamma=\sin(\pi\rho)/(\pi\rho)$. We obtain
\[
\FO^+
=\sum_{0\ne y\in C^\perp}
\prod_{i\in\supp(y)}\lambda|\wh g_i(y_i)|
\le n^{O(1)}
\left(\frac{\sin(\pi\rho)}{\pi\rho}\right)^k
=e^{-\Omega(n)}.\qedhere
\]
\end{proof}
By \cref{lem:translated-opi-proxy}, $\FO(v)\le\FO^+$ uniformly
over translates. This proves \cref{thm:main-opi}(1).
The same bound for heterogeneous lists with a fixed positive lower
density follows from \cref{prop:subsample}.

Next, we prove the balanced case of \cref{thm:main-opi}(2). As in \cref{subsec:llr-below}, we split the support $A$ of codewords $y \in C^\perp$  at cutoff $m\sim (2-\delta)k$, use \cref{lem:wftl} to control the small-support case, and use \cref{lem:u3-gap,lem:block-u3} to control the large-support case.

\begin{lemma}
\label{lem:balanced-opi-below}
There exists $\delta>0$ such that for every $[n,k]$
generalized Reed--Solomon code $C$ of fixed rate $R\ge1/2-\delta$
over a prime field $\F_p$ with $p\ge n$, every balanced OPI instance
(that is, where $\rho_i=\rho=1/2+o(1)$),
and every $v\in\F_p^n$,
\[
\FO(v)\le e^{-\Omega(n)}.
\]
\end{lemma}
\begin{proof}
When $R>1/2$, the statement follows from \cref{lem:positive-opi-above}, so we assume $R \leq 1/2$. Otherwise, fix the absolute constant $q<1$ from \cref{lem:high-support-energy}, choose $\delta>0$ sufficiently small, take $R\in[1/2-\delta,1/2]$, and use the same cutoff
$m=\lfloor(2-\delta)k\rfloor$ from \cref{eq:puncture-size}.
Choose $\delta$ so that
\begin{equation}
\label{eq:opi-below-condition}
H_2\left(\frac52\delta-\delta^2\right)
+\left(\frac14-\frac32\delta\right)\log_2q<0.
\end{equation}
Such a choice exists because the left-hand side is strictly negative at $\delta=0$.
Write $\rho_i=\rho=1/2+o(1)$ and $\lambda=\sqrt{(1-\rho)/\rho}$. For every support $A$, $\prod_{i\in A}\lambda_i=\lambda^{|A|}=2^{o(n)}$
uniformly in $A$, so we may ignore the $\lambda$'s. Applying \cref{eq:translated-EO} and taking the triangle inequality over supports gives
\begin{equation}
\label{eq:opi-below-low-high}
\FO(v)\le
2^{o(n)}
\sum_{\substack{0\ne y\in C^\perp\\|y|\le m}}
\prod_{i\in\supp(y)}\abs{\wh g_i(y_i)}
+2^{o(n)}
\sum_{\substack{A\subset[n]\\|A|>m}}
\abs{\sum_{\substack{y\in C^\perp\\\supp(y)=A}}
e_p(-\langle v,y\rangle)\prod_{i\in A}\wh g_i(y_i)}.
\end{equation}
We bound the two summands separately.
\begin{itemize}
\item For $|y|\le m$, we apply \cref{lem:wftl} exactly as in the small-support part of the proof of \cref{lem:balanced-llr-below}:
\[
\sum_{\substack{0\ne y\in C^\perp\\|y|\le m}}
\prod_{i\in\supp(y)}\abs{\wh g_i(y_i)}
\le n^{O(1)}\left(\frac2\pi\right)^k=e^{-\Omega(n)}.
\]
\item For $|A|>m$, we apply \cref{lem:high-support-energy} to each term and count supports using \cref{eq:high-support-count}:
\begin{equation}
\begin{aligned}
\sum_{\substack{A\subset[n]\\|A|>m}}
\abs{\sum_{\substack{y\in C^\perp\\\supp(y)=A}}
e_p(-\langle v,y\rangle)\prod_{i\in A}\wh g_i(y_i)}
&\le
2^{\left[H_2\left(\frac52\delta-\delta^2\right)+\frac{3R-1}{2}\log_2q+o(1)\right]n}
\\ &\le
2^{\left[H_2\left(\frac52\delta-\delta^2\right)+\left(\frac14-\frac32\delta\right)\log_2q+o(1)\right]n}
\\ & =2^{-\Omega(n)}.
\end{aligned}
\end{equation}
The last inequality uses $R\ge1/2-\delta$, $\log_2q<0$,
and \cref{eq:opi-below-condition}.
\end{itemize}
Combining the two estimates proves \cref{lem:balanced-opi-below}upon multiplying by $2^{o(n)}$.
\end{proof}
By \cref{lem:translated-opi-proxy}, this proves the balanced case of \cref{thm:main-opi}(2).
The extension to fixed density $\rho\ge1/2-\delta$, after decreasing the constant $\delta$ if necessary, is deferred to \cref{sec:app-opi}.
\subsection{Recap of the Satisfaction Expansion in \cite{us}}
\label{subsec:expansion-recap}

In this section, we recall the framework in \cite{us} to analyze the satisfaction ratio obtained by
measuring the DQI state $\ket{\varphi_\ell}$ from
\cref{eq:dqi-u}. These tools will help us prove \cref{thm:main-opi}(3) and the algorithmic results in \cref{sec:algo}.
To set up\footnote{Note that \cite{us} handles general $[m, n]$ MDS codes and defines functions on the message space $\F_p^k$. We specialize to $[n, k]$ GRS codes for simplicity, but define our functions on the code space to handle LLR and OPI. More precisely, our functions are defined
on the ambient space containing $C$, and we identify the message space with $C$.  We have translated the framework of \cite{us} to this setting in our exposition.}, fix an $[n,k]$ GRS code $C\subset\F_p^n$
and input lists of common 
density $\rho=|S_i|/p\in(0,1)$. 
Recall the centered functions $g_i$ and the $j$-wise discrepancies
from the overview:
\begin{equation}
\label{eq:sw-layer-definition}
q_j(x):=\sum_{A\in\binom{[n]}j}\prod_{i\in A}g_i(x_i),
\qquad
\beta_\rho:=\frac{1-2\rho}{\sqrt{\rho(1-\rho)}}.
\end{equation}
We begin with the key Fourier lemma about discrepancies.
\begin{lemma}[{\cite[Lemma~4.1]{us}}]
\label{lem:key-fourier}
For any $0\le t\le n$ and every $x\in\F_p^n$
\begin{equation}
\label{eq:key-fourier}
\E_{x\in C}[q_t(x)]
=
\sum_{\substack{y\in C^\perp\\|y|=t}}
\prod_{i\in\supp(y)}\wh g_i(y_i).
\end{equation}
\end{lemma}
In particular, the expectation $\E_{x \in C}[q_t(x)]$ vanishes for $1\le t\le k$, as $C^\perp$ has distance $k+1$. The expectation  of $q_0$ is $\E_{x \in C}[q_0(x)] = 1$.
Write $s(x):=n^{-1}\sum_i\one_{S_i}(x_i)$ for the satisfaction ratio, and for $\ell \geq 0$ and $a \in \mathbb{R}^{\ell+1}$, define the quantum state
\begin{equation}\label{eq:general-u}
    \ket{\varphi_\ell(a)}\propto \sum_{x\in C}\sum_{j=0}^{\ell}a_jq_j(x)\ket{x}.
\end{equation}
Observe that $s(x)=\rho+\sqrt{\rho(1-\rho)} q_1(x)/n$ and that sampling from $\ket{\varphi_\ell(a)}$
produces $x \in C$ with probability proportional to $|\sum_{j}a_jq_j(x)|^2$. Therefore,  the expected satisfaction ratio is
\begin{equation}
\label{eq:sw-score-ratio}
\rho+\frac{\sqrt{\rho(1-\rho)}}n
\frac{\sum_{j,j'=0}^{\ell}a_ja_{j'}
\E_{x\in C}[q_1q_jq_{j'}]}
{\sum_{j,j'=0}^{\ell}a_ja_{j'}
\E_{x\in C}[q_jq_{j'}]}
\end{equation}
when we sample from $\ket{\varphi_\ell(a)}$.
 This is the main expansion in \cite[(4.3), (4.19)]{us}. Unlike \cite{us}, we state it with general variables $a_j$ (rather than pre-specified weights) to allow the flexibility of different choices in \cref{subsec:opi-dense} versus \cref{sec:algo}. Namely, the notation $u_j$ remains reserved for the fixed DQI weights in \cref{eq:dqi-u}; taking $a_j=u_j$ gives $F_\ell$, while taking $a_j=\rho^{(n-j)/2}(1-\rho)^{j/2}$ gives $H_\ell$ in \cref{eq:def-H-ell}.

Thus, to use \cref{lem:key-fourier}, we need to
expand the products $q_jq_{j'}$ and $q_1q_jq_{j'}$ back into the
functions $q_t$. To do so, we note that, because $g_i$ takes only two values,
\cite[Lemma~4.2]{us} gives
\begin{equation}
\label{eq:sw-lin}
g_i^2=1+\beta_\rho g_i.
\end{equation}
This enables the following expansion of $q_jq_{j'}$ by considering repeated indices and counting.
\begin{lemma}[{\cite[Proposition~4.3]{us}}]
\label{lem:layer-recursion}
For $0\le j,j'\le n$ and any $x\in \F_p^n$,
\begin{equation}
q_jq_{j'}=\sum_{t=0}^nN_\rho(j,j';t)q_t
\quad\text{and}\quad
q_1q_j =(j+1)q_{j+1}+(n-j+1)q_{j-1}
+\beta_\rho j q_j,
\end{equation}
where $ N_\rho(j,j';t)$ and its upper bound $ \widetilde{N_\rho}(j,j';t)\ge |N_\rho(j, j';t)|$ are defined as follows:
\begin{equation}
\label{eq:def-N}
\begin{aligned}
 N_\rho(j,j';t)&:=[X^jY^{j'}]
(1+XY)^{n-t}(X+Y+\beta_\rho XY)^t\\
\widetilde{N_\rho}(j,j';t)&:= [X^jY^{j'}]
(1+XY)^{n-t}(X+Y+|\beta_\rho| XY)^t.\end{aligned}
\end{equation}
\end{lemma}
Here, for $P\in \R[X, Y]$, let $[X^jY^{j'}]P(X, Y)$ denote the coefficient of the $X^jY^{j'}$ term.
Now,
$N_\rho(j,j';0)=\binom nj\one\{j=j'\}$, and
$N_\rho(j,j';t)=0$ for $t>j+j'$. Since $\E_{\F_p^n}q_t=0$ for $t>0$, this also gives
\begin{equation}
\label{eq:sw-local-and-orthogonality}
\E_{x\in\F_p^n}[q_jq_{j'}]=\binom nj\one\{j=j'\}.
\end{equation}

Hence, once we expand \cref{eq:sw-score-ratio} by the recurrence in \cref{lem:layer-recursion},
we see that the numerator and denominator of  involve $q_t$ for $t\le 2\ell+1$ only. 
We now specialize the DQI state with weights $u_j$ chosen as in \cref{eq:dqi-u}: pick $\sigma = \lfloor\log\log n\rfloor$ and let
\begin{equation} 
\label{eq:dqi-state}
\ket{\varphi_\ell}\propto \sum_{x\in C}F_\ell(x)\ket{x}
    \quad\text{where}\quad
F_\ell(x):=
\sum_{j=\ell-\sigma}^{\ell}
\frac{q_j(x)}{\sqrt{\binom nj}}.
\end{equation}
The following main theorem of \cite{us} states a sufficient condition for this DQI state $\ket{\varphi_\ell}$ to attain a satisfaction ratio of $\on{SCL}_\rho(\ell/n)$  asymptotically. This sufficient condition is especially convenient for us to check since it is very similar to the condition that $\FO^+$ is exponentially small, which we used to prove near-optimal OPI discrepancy in \cref{subsec:opi-above}.

\begin{theorem}[{\cite[Proposition 4.7 and (4.23)]{us}}]
\label{thm:sw-main}
Fix $\alpha\in (0, 1-\rho]$ and let $\ell =\lfloor \alpha n\rfloor$.
Suppose that, for some fixed $c>0$,
\begin{equation}
\label{eq:key-decay}
\frac{\widetilde N_\rho(j,j';t)}
{\sqrt{\binom nj\binom n{j'}}}
\abs{\E_{x\in C}q_t(x)}
\le e^{-cn}
\end{equation}
uniformly for $k<t\le n$ and
$\ell-\sigma-1\le j,j'\le\ell+1$. Then, for some $c'>0$,
\begin{equation}
\label{eq:sw-main}
\begin{aligned}
\E_{x\in C}\left[F_\ell(x)^2\right]
&=\sigma+1+O(e^{-c'n}),\\
\E_{x\in C}[s(x)F_\ell(x)^2]
&=(\sigma+1)\left[
\SCL_\rho(\alpha)
+O_{\rho,\alpha}\left(\frac{\sigma}{n}+\frac1\sigma\right)
\right]+O(e^{-c'n}).
\end{aligned}
\end{equation}
Consequently, measuring the normalized state $\ket{\varphi_\ell}$ gives expected satisfaction ratio $\SCL_\rho(\alpha)+o(1)$.
\end{theorem}

\subsection{OPI with Rate $R > 1 - \rho$}
\label{subsec:opi-dense}
In this section, we leverage the framework in \cref{subsec:expansion-recap} to prove \cref{thm:main-opi}(3). Fix density $\rho\in(0,1)$, so $|S_i|=\rho p$ for all $i$. In \cref{eq:def-N}, the generating function of
$\widetilde N_\rho(j,j';t)$ has nonnegative coefficients, so for all $z>0$,
\begin{equation}
\label{eq:N-bound}
\widetilde{N}_\rho(j,j';t)
\le z^{-j-j'}(1+z^2)^{n-t}
(2z+|\beta_\rho|z^2)^t.
\end{equation}
We begin with the main Fourier proxy bound in the capacity regime.
\begin{lemma}
\label{lem:N-proxy-bound}
If $t\in [n]$ satisfies $2k-t=\Omega(n)$, then
uniformly over the input lists,
\begin{equation}
\label{eq:N-proxy-bound}
\sum_{\substack{0\ne y\in C^\perp\\|y|\le t}}
b_\rho^{|y|}
\prod_{i\in\supp(y)}|\wh g_i(y_i)|
\le e^{-\Omega(n)} \quad\text{where}\quad b_\rho:=\sqrt{\frac {1-\rho}{\rho}} (2\rho+|1-2\rho|)
\end{equation}
\end{lemma}

\begin{proof}
We apply \cref{lem:wftl} to
$w_i(\xi):=b_\rho|\wh g_i(\xi)|$, so we first check the conditions on $w_i$:
\begin{itemize}
\item If $\rho\le1/2$, then
$b_\rho=\sqrt{(1-\rho)/\rho}$. Thus
\cref{lem:indicator-L-infty-bound} gives, for $\xi\ne0$,
\[
w_i(\xi)
=\frac{|\wh{\one_{S_i}}(\xi)|}{\rho}
\le\frac{\sin(\pi\rho)}{\rho p\sin(\pi/p)}
=\frac{\sin(\pi\rho)}{\pi\rho}+O_\rho(p^{-2}).
\]
The limiting bound is strictly less than one for fixed $\rho\in (0,1)$, uniformly over $\xi\ne 0$.
\item If $\rho\ge1/2$, let $x:=2\rho-1\in[0,1)$. Then $b_\rho\le\sqrt{2}$ as
\[
b_\rho^2
=\frac{1+3x-4x^3}{1+x}=2-\frac{1-x+4x^3}{1+x}\le 2.
\]
By \cref{cor:L-infty-bound},
$w_i(\xi)\le2\sqrt2/\pi+O(p^{-2})<1$
for all sufficiently large $p$.
\end{itemize}
In both cases there is a fixed $\gamma_\rho<1$ bounding all
$w_i(\xi)$ for sufficiently large $n$.
Parseval's identity implies that
\[
\sum_{\xi\ne0}w_i(\xi)^2
=
b_\rho^2\sum_{\xi\ne0}|\widehat g_i(\xi)|^2
=
b_\rho^2\E|g_i|^2
=
b_\rho^2
\le
\max\left\{\frac{1-\rho}{\rho},2\right\}=:\Gamma.
\]
If $t\le k$, the sum on the left of \cref{eq:N-proxy-bound} vanishes by the dual distance.
Otherwise, by assumption that $2k-t=\Omega(n)$, \cref{lem:wftl} with $\gamma = \gamma_\rho$ and $\Gamma$ as defined above upper bounds the left hand side of \cref{eq:N-proxy-bound} by $e^{-\Omega(n)}$.
\end{proof}
Now, assuming $R+\rho>1$, then we there exists $\ell\in [n]$ and $\alpha:=\ell/n$ such that 
\begin{equation}
\label{eq:def-ell-alpha}
(1-\rho)n+\Omega(n)\le \ell=\alpha n\le k-\Omega(n).
\end{equation}
We next implement the truncation idea outlined in Step 5 of \cref{subsec:overview}. First, define
\begin{equation}
\label{eq:def-H}
H(x):=
\prod_{i=1}^n\left(\sqrt\rho+\sqrt{1-\rho}\cdot g_i(x_i)\right)
=\rho^{-n/2}\one_{\mbf S}(x).
\end{equation}
The equality holds %
since
\begin{align*}
    \prod_{i=1}^n \left( \sqrt{\rho} + \sqrt{1 - \rho} \cdot g_i(x_i) \right) &= \prod_{i=1}^n \left( \sqrt{\rho} + \sqrt{1 - \rho} \left(\frac{\one_{S_i}(x_i) - \rho}{\sqrt{\rho(1 -\rho)} } \right)\right) \\
    &= \prod_i \left( \sqrt{\rho} + \left(\frac{\one_{S_i}(x_i)}{\sqrt{\rho}} - \sqrt{\rho} \right) \right) \\
    &= \rho^{-n/2} \mathbf{1}_{\mathbf{S}}(x).
\end{align*}
We can also expand 
\[
H(x)= \sum_{j=0}^n (1-\rho)^{j/2} \rho^{(n-j)/2} \sum_{S\in\binom{[n]}{j}} \prod_{i\in S} g_i(x_i) = \sum_{j=0}^n (1-\rho)^{j/2} \rho^{(n-j)/2} q_j(x).
\]
This motivates the truncation in the sum from $n$ to $\ell$ to obtain
\begin{equation}
\label{eq:def-H-ell}
H_\ell(x):=\sum_{j=0}^{\ell}
\rho^{(n-j)/2}(1-\rho)^{j/2}q_j(x).
\end{equation}
Let $B \sim \on{Bin}(n, 1 - \rho)$.
For convenience, we define the density functions
\begin{equation}
\label{eq:def-pi-j}
\pi_j:=\mb{P}(B=j)=\binom nj\rho^{n-j}(1-\rho)^j
\quad\text{and}\quad
\pi_{\le j}:=\mb{P}(B\le j)=\sum_{i=0}^{\ell}\pi_i.
\end{equation}
By \cref{eq:sw-local-and-orthogonality} and a Chernoff bound, we see that 
\begin{equation}
\label{eq:truncation-0}
\mb{E}_{x\in\F_p^n}\left[H_\ell(x)^2\right] = \pi_{\le \ell} = 1-e^{-\Omega(n)}.
\end{equation}

While $H_\ell$ lowers the degree to $\ell$ from $n$, for $H_\ell$ to still control the Fourier proxy $\FO(v)$, we need to ensure it preserves two properties of $H(x)$: it has comparable total mass on $C$, and it almost all its mass is on perfect solutions.  We establish this next.
\begin{lemma}
\label{lem:key-truncation}
With $H_\ell$ defined in \cref{eq:def-H-ell}, we have that
\begin{equation}
\left|\E_{x\in C}[H_\ell(x)^2]-1\right|\le e^{-\Omega(n)}\quad\text{and}\quad  0\le\E_{x\in C}[(1-s(x))H_\ell(x)^2]\le e^{-\Omega(n)}.
\end{equation}
\end{lemma}
\begin{proof}
Let $a_j:=\rho^{(n-j)/2}(1-\rho)^{j/2}$ and
$z:=\sqrt{(1-\rho)/\rho}$, so $a_j=\rho^{n/2}z^j$. 
Expanding
$H_\ell^2$ gives
\begin{align*}
\E_{x\in C}[H_\ell(x)^2]
=
\sum_{j,j'=0}^{\ell}a_ja_{j'}
\sum_{t=0}^{n}N_\rho(j,j';t)\E_{x\in C}[q_t(x)]
=
\pi_{\le\ell}
+
\sum_{t=1}^{\min(n,2\ell)}
\left(
\sum_{j,j'=0}^{\ell}a_ja_{j'}N_\rho(j,j';t)
\right)\E_{x\in C}[q_t(x)],
\end{align*}
by \cref{lem:layer-recursion}.
The $t=0$ term is $\pi_{\le\ell}$ because
$N_\rho(j,j';0)=\binom nj\one\{j=j'\}$ and $q_0=1$. Now, since $\widetilde{N}_\rho$ are nonnegative coefficients, \cref{eq:def-N}  gives for every $1\le t\le \min(n, 2\ell)$ that
\begin{equation}
\label{eq:truncation-proxy}
\left|
\sum_{j,j'=0}^{\ell}a_ja_{j'}N_\rho(j,j';t)
\right|
\le
\sum_{j,j'=0}^{\ell}a_ja_{j'}\widetilde N_\rho(j,j';t)\le
\sum_{j,j'=0}^{n}a_ja_{j'}\widetilde N_\rho(j,j';t)=
\rho^n(1+z^2)^{n-t}(2z+|\beta_\rho|z^2)^t.
\end{equation}
We claim that the last quantity is exactly $b_\rho^t$. This follows from $1+z^2=1/\rho$ and from \cref{eq:sw-layer-definition}:
\begin{equation}
\label{eq:b-rho}
   \beta_\rho:=\frac{1-2\rho}{\sqrt{\rho(1-\rho)}} \implies  b_\rho = 2\sqrt{\rho(1-\rho)}+(1-\rho)|\beta_\rho| = \rho(2z+z^2|\beta_\rho|).
\end{equation}
Therefore, by \cref{lem:N-proxy-bound}, we have that
\begin{equation}
\label{eq:key-trunction-2}
\left|\E_{x\in C}[H_\ell(x)^2]-\pi_{\le\ell}\right|
\le
\sum_{t=1}^{\min(n,2\ell)}
b_\rho^t\left|\E_{x\in C}[q_t(x)]\right|
\le
\sum_{\substack{0\ne y\in C^\perp\\|y|\le\min(n,2\ell)}}
b_\rho^{|y|}
\prod_{i\in\supp(y)}|\wh g_i(y_i)|
\le e^{-\Omega(n)},
\end{equation}
where the last inequality follows from \cref{lem:key-fourier} upon checking that, by \cref{eq:def-ell-alpha}, 
\begin{equation}
\label{eq:truncation-gap}
2k-\min(n,2\ell+1)\ge2(k-\ell)-1=\Omega(n).
\end{equation}
Finally, the first bound in \cref{lem:key-truncation} follows triangle inequality and \cref{eq:truncation-0}:
\[
\left|\E_{x\in C}[H_\ell(x)^2]-1\right|
\le
\left|\E_{x\in C}[H_\ell(x)^2]-\pi_{\le\ell}\right|
+(1-\pi_{\le\ell})
\le e^{-\Omega(n)}.
\]

For the second bound, recall \cref{eq:sw-score-ratio}, apply 
\cref{lem:layer-recursion}, and use $a_j=\rho^{(n-j)/2}(1-\rho)^{j/2}$ to expand
\begin{align*}
(1-s)H_\ell
&=
(1-\rho)\sum_{j=0}^{\ell}a_jq_j
-\frac{\sqrt{\rho(1-\rho)}}{n}
\sum_{j=0}^{\ell}a_j
\left[
(j+1)q_{j+1}+(n-j+1)q_{j-1}+\beta_\rho j q_j
\right]\\
&=
\sum_{j=0}^{\ell}a_j
\left[
1-\rho-
\frac{j\rho+(n-j)(1-\rho)+(1-2\rho)j}{n}
\right]q_j+
\frac{\sqrt{\rho(1-\rho)}}{n}
\left[
(n-\ell)a_{\ell+1}q_\ell
-(\ell+1)a_\ell q_{\ell+1}
\right]\\
&=
\frac{a_\ell}{n}
\left[
(1-\rho)(n-\ell)q_\ell
-\sqrt{\rho(1-\rho)}(\ell+1)q_{\ell+1}
\right],
\end{align*}
where we observe in the last step that the square bracket coefficient vanishes. Multiplying by $H_\ell$ and using
$\sqrt{\rho(1-\rho)}a_\ell=\rho a_{\ell+1}$ gives
\begin{align*}
\E_{x\in C}[(1-s(x))H_\ell(x)^2]
&=
\frac{(1-\rho)(n-\ell)}{n}
\sum_{j=0}^{\ell}a_\ell a_j\E_{x\in C}[q_\ell(x)q_j(x)]-
\frac{\rho(\ell+1)}{n}
\sum_{j=0}^{\ell}a_{\ell+1}a_j
\E_{x\in C}[q_{\ell+1}(x)q_j(x)].
\end{align*}
Now, we expand both products $q_{\ell}(x) q_j(x)$ and $q_{\ell+1}(x)q_j(x)$ using \cref{lem:layer-recursion}, which says that
$q_jq_{j'}=\sum_{t=0}^nN_\rho(j,j';t)q_t$.
We get
\begin{align*}
    \E_{x \in C} [ (1 - s(x)) H_\ell(x)^2 ] & = \frac{(1 - \rho)(n - \ell)}{n} \sum_{j=0}^\ell a_\ell a_j  \sum_{t=0}^n N_\rho(\ell,j;t) \E_{x \in C}[q_t(x)] \\
    &\ \ \ - \frac{ \rho(\ell + 1)}{n} \sum_{j=0}^\ell a_{\ell + 1} a_j \sum_{t=0}^n N_\rho(\ell + 1, j;t) \E_{x \in C}[q_t(x)].
\end{align*}

Since $N_\rho(j,j';0)=\binom nj\one\{j=j'\}$, the summands with $t=0$ and any $j$ together contribute $a_\ell^2\binom n\ell=\pi_\ell$ to the first sum, ignoring the prefactor of $(1 - \rho)(n-\ell)/n$, and vanishes in the second sum.
For the summands with $t>0$, we follow the same proof as the first inequality: for $r\in\{\ell,\ell+1\}$, we have that
\[
\sum_{j=0}^{\ell}a_ra_j|N_\rho(r,j;t)|
\le
\sum_{j,j'=0}^{n}a_ja_{j'}\widetilde N_\rho(j,j';t)
=b_\rho^t.
\]
similar to \cref{eq:truncation-proxy}. Just as \cref{eq:key-trunction-2}, we now have by \cref{lem:key-fourier} and \cref{eq:truncation-gap} that
\begin{align*}
\left|
\E_{x\in C}[(1-s(x))H_\ell(x)^2]
-\frac{(1-\rho)(n-\ell)}{n}\pi_\ell
\right|\le
2\sum_{t=1}^{\min(n,2\ell+1)}
b_\rho^t\left|\E_{x\in C}[q_t(x)]\right|\le e^{-\Omega(n)}.
\end{align*}
Finally, the second bound in \cref{lem:key-truncation} follows \cref{eq:truncation-0} and that $(1-s(x))H_\ell(x)^2\ge0$, since
\[
0\le
\E_{x\in C}[(1-s(x))H_\ell(x)^2]
\le
\frac{(1-\rho)(n-\ell)}{n}\pi_\ell+e^{-\Omega(n)}
\le e^{-\Omega(n)},
\]
and $\pi_\ell \le e^{-\Omega(n)}$. This concludes the proof of \cref{lem:key-truncation}.
\end{proof}
With these properties, we prove the following lemmas.
\begin{lemma}
\label{lem:terminal-dense-count}
Let $p\ge n$ be prime, let $R+\rho>1$, and let $C$ be a  $[n, k]$ GRS code of rate $R$ over $\F_p$.  Fix $\rho \in (0,1)$ and input lists $S_i$ with $|S_i|=\rho p$ for every $i$. Then uniformly over $v\in\F_p^n$,
\begin{equation}
\FO(v)=
\abs{\frac{|(v+C)\cap\mbf S|}{|C|\rho^n}-1}
\le e^{-\Omega_{R,\rho}(n)}
\end{equation}
\end{lemma}
\begin{proof}
First take $v=0$. Every imperfect point has
$1-s(x)\ge1/n$. By \cref{lem:key-truncation},
\[
\E_{x\in C}[\one_{\mbf S^c}(x)H_\ell(x)^2]
\le
n\E_{x\in C}[(1-s(x))H_\ell(x)^2]
\le e^{-\Omega(n)}.
\]
Now, observe that
$H_\ell(x)=\rho^{-n/2}\pi_{\le\ell}$ 
for every $x\in\mbf{S}$, as 
$g_i(x_i)=\sqrt{(1-\rho)/\rho}$ for every $i$. Therefore,
\begin{align*}
\pi_{\le\ell}^2\frac{|C\cap\mbf S|}{|C|\rho^n}
=\E_{x\in C}[\one_{\mbf S}(x)H_\ell(x)^2]
=\E_{x\in C}[H_\ell(x)^2]
-\E_{x\in C}[\one_{\mbf S^c}(x)H_\ell(x)^2]=1+e^{-\Omega(n)}.
\end{align*}
by \cref{lem:key-truncation} and \cref{eq:truncation-0}, and the desired conclusion for $v=0$ follows.
For general $v$, apply the same argument to the translated
lists $S_i-v_i$. Their densities are unchanged, and
\[
\FO(v;C,\mbf S)
=\FO(0;C,\mbf S-v).\qedhere
\]
\end{proof}
Now, applying this lemma with $v=0$ as well as \cref{lem:restated-llr-proxy} gives near optimal OPI discrepancy. At this point, we have proven all parts of \cref{thm:main-opi}, which we formally collect below.

\begin{proof}[Proof of \cref{thm:main-opi}]
It suffices by \cref{prop:subsample}, applied also to translates
$v+C$, to establish a uniform relative-count bound for common-size
sublists. We collect the three cases.

For (1) and (3), take $t=\lceil\rho p\rceil$. Its actual density
$t/p=\rho+O(1/p)$ preserves the fixed density and rate gaps.
Part~(1) follows from
\cref{lem:positive-opi-above,lem:translated-opi-proxy}, and part~(3)
follows from \cref{lem:terminal-dense-count}.

For (2), choose the theorem's constant $\delta$ smaller than the
constant $\delta_0$ of \cref{lem:near-balanced-opi-below}, and take
\[
    t=\left\lceil p\min\{\rho,1/2\}\right\rceil.
\]
For sufficiently large $p$,
    $t/p\in[1/2-\delta_0,1/2+\delta_0]$,
so that lemma applies.

In every case, $t\le\min_i|S_i|$, and \cref{prop:subsample} transfers
the exponentially small relative error to the original normalization
$|C|\prod_i\rho_i$. Taking $v=0$ proves all three assertions.
\end{proof}

\subsection{Near-Optimal OPI Discrepancy for Random Reed-Solomon Codes}
\label{subsec:rrs}
In this subsection, we prove \cref{thm:random-opi}, which states that that random Reed-Solomon codes of any positive rate have near-optimal OPI discrepancy with high probability. Throughout, let
$C=\on{RS}_p(\alpha_1,\ldots,\alpha_n;k)$ be the random RS code
obtained from a uniformly random ordered tuple of distinct nonzero evaluation points in $\F_p$.
We write $\E_C$ and $\P_C$ for expectation and probability over
the choice of evaluation points; $\E_{x\in C}$ always means an average over a
codeword of the fixed code $C$.

We begin by outlining the proof of one-bit LLR of random Reed-Solomon codes from \cite{HMSY26}.
The first step is to approximate each local indicator as
$\one_{S_i}=G_i+H_i$, where $G_i$ is structured and $H_i$ has
small Gowers $U^s$-norm for some large constant $s$. More precisely, the following decomposition holds.

\begin{fact}[Structured decomposition, {\cite[Section~10.1]{HMSY26}}]
\label{fact:hmsy-approx}
For fixed $s\ge4$ and $\eta>0$, there is a class $\mc G$ of
functions $\F_p\to[0,1]$, such that for every $S\subset\F_p$, the 
indicator function $\mbf{1}_S$ has a decomposition
\[
\mbf{1}_S=G+H,\qquad G\in\mc G,\qquad |H|\le1,
        \qquad |\E H|\le\eta\le \|H\|_{U^s}\le\eta.
\]
Moreover, $\mc{G}$ is closed under invertible affine precomposition, i.e. for every $a\in\F_p$ and $b\in\F_p^\times$ and $G\in\mc{G}$, $Gz\mapsto G(a+bz)$ is also in $\mc{G}$.
\end{fact}

Upon expanding $\prod_i\one_{S_i}(x_i)$ in the LLR proxy, \cite{HMSY26} considers a set $J\subset [n]$ indexing products of $H_i$. If $|J|$ is large, the small Gowers $U^s$-norm gives exponential decay. To handle the terms with small $|J|$, 
\cite{HMSY26} uses the structure of $G_i\in\mc{G}$ for $i\not\in J$ to construct a good evaluation event $\mathsf{Good}$ via higher order Fourier analysis. In more detail, \cite{HMSY26} shows the following:

\begin{enumerate}
    \item Fix any Reed-Solomon code $C$ that satisfies the event $\mathsf{Good}$. Upon revealing $O(1)$-many coordinates values $u$, \cite{HMSY26} constructs an auxiliary linear code $\widetilde{C}$ whose Fourier proxy is close to that of $C$. Moreover, $\widetilde{C}^\perp$ has dimension $O(1)$ and distance $k-O(1)$, where $k=\dim C$.
    \item In contrast, $C^\perp$ has dimension $n-k$ and distance $k$. The reduction of dimension but not distance in $\widetilde{C}^\perp$ enables the Fourier and Gowers-norm bounds to show one-bit LLR at any positive rate $R=k/n$, bypassing the rate-$1/2$ barrier. This is discussed in \cref{rem:like-bdir}.
    \item $\mathsf{Good}$ occurs with high probability for a random
    Reed--Solomon code. %
\end{enumerate}

We import Items 1 and 3 above as a black box, which we state as \cref{fact:hmsy-model} below.  Then we adapt Item 2 above to use \cref{fact:hmsy-model} to apply to OPI instead of LLR; this is done in \cref{prop:hmsy-relative}.

\begin{restatable}{fact}{hmsyfact}
\label{fact:hmsy-model}\emph{(Implicit in {\cite[Proof of Proposition 10.3]{HMSY26}}).}
Fix any $s\ge 4$ and $\eta>0$, and let $\mc G$ be the class of functions guaranteed in \cref{fact:hmsy-approx}.  Also, fix any rate $R=k/n\in (0, 1)$, and $\nu, L>0$.

For sufficiently large $n$ and $p$ where $\log p\ge n^{\Omega(1)}$, there exists an event $\mathsf{Good}$
with $\P_C(\mathsf{Good})\ge 1-\eps$ for $\eps:=e^{-Ln}$, on which the
following holds simultaneously for every $\mbf G=(G_1, \dots, G_n)\in\mc G^n$ with $\nu_i:=\mb{E}G_i\ge \nu$ for every $i\in [n]$:
there exists $I\subset [n]$, finite dimensional vector spaces $V_i$ over
$\F_p$ for each $i\in I$, a linear code
$\widetilde C\subset\prod_{i\in I}V_i$, and functions
$\widetilde G_i:V_i\to[0,1]$, such that:
\begin{enumerate}
    \item $\widetilde{C}^\perp$ has length $|I|=n-O(1)$, dimension $\widetilde{k}=O(1)$, and minimum distance $\widetilde{d}\ge k-O(1)$,
    \item $\mb{E}\widetilde{G}_i\le \nu_i$ for every $i\in I$, and \begin{equation}\label{eq:hmsy-fourier-reduction}
 \left|\E_{x\in C}\prod_{i=1}^nG_i(x_i)-\prod_{i=1}^n\nu_i\right|
 \le 2\sum_{0\ne y\in\widetilde C^\perp}
       \prod_{i\in I}|\widehat{\widetilde{G}}_i(y_i)|+\eps.
\end{equation}
\end{enumerate}
\end{restatable}
For completeness, we include a proof sketch of \cref{fact:hmsy-model} in \cref{app:hmsyProofSketch}.

To use \cref{fact:hmsy-model} to prove near-optimal OPI discrepancy instead of one-bit LLR, we adapt Item 2 in the high-level outline of \cite{HMSY26} above. The main idea is to apply the density-normalized OPI Fourier proxy to the auxiliary code $\widetilde{C}$. This
gives relative counting error in place of absolute leakage error. Combining this with the Gowers norm bounds \cref{fact:block-ul}  yields near-optimal OPI discrepancy. The key step is the following control of the Fourier proxy for structured leakages $\mbf{G}\in \mc{G}^n$, analogous to \cref{lem:wftl}.  We prove this next as \cref{prop:hmsy-relative}.  After the proof of the proposition, we will put the pieces together as described above to formally prove \cref{thm:random-opi}.

\begin{proposition}\label{prop:hmsy-relative}
Fix any rate $R=k/n, \nu\in(0,1)$ and $\mc G$. There exists $\sigma_\nu\in (0, 1)$ such that
for sufficiently large $n$ and prime $p$ with $\log p= n^{\Omega(1)}$, we have that
\begin{equation}\label{eq:hmsy-relative}
 \E_C\left[
   \sup_{\substack{\mbf G\in\mc G^n\\\E G_i\ge\nu}}
   \FO(0;C,\mbf G)\right]=O(\sigma_\nu^k),
\end{equation}
\end{proposition}
\begin{proof}
Define $h:(0,1]\to\R_{\ge0}$ by
\[
 h(t):=\frac{\sin(\pi\min\{t,1/2\})}{\pi t}.
\]
Observe that $h$ is decreasing and $h(t)<1$ for every $t>0$.
Fix $\sigma_\nu\in(h(\nu),1)$ and $L>\log(1/\nu)+\log(1/\sigma_\nu)$.
Apply \cref{fact:hmsy-model} with $\mc G$, rate $R$, density
lower bound $\nu$, and $L$, and adopt the objects and notation in \cref{fact:hmsy-model}. Fix any $C$ satisyfing the event $\mathsf{Good}$ and
$\mbf G$ with $\nu_i:=\E G_i\ge\nu$ for every $i\in I$. 
For each $i\in I$, define
\[
 w_i(\xi):=\frac{|\widehat{\widetilde{G}}_i(\xi)|}{\nu_i}.
\]
Now, $\mb{E}\widetilde{G}_i\le \nu_i$ so $w_i(0)\le 1$, and $|\widehat{\widetilde G}_i(\xi)|\le\E\widetilde G_i$ gives $w_i(\xi)\le1$ for all $\xi$. As $0\le\widetilde{G}_i\le1$ so $\widetilde{G}_i^2\le \widetilde{G}_i$,
\[
 \sum_{\xi\in V_i^\star}w_i(\xi)^2 = \nu_i^{-2}\mb{E}  \widetilde{G}_i^2 
 \le\nu_i^{-1}\le\nu^{-1},
\]
by Parseval's identity. Next, we show that every nonzero character $\xi$ gives a decay factor, analogous to \cref{cor:L-infty-bound}.
For $\xi\ne0$, maximizing \cref{cor:soft-arc} over means $\mb{E}\widetilde{G}_i\le \nu_i$ and using $\nu_i\ge \nu$ gives
\[
 w_i(\xi)
 \le
 \frac{\sin(\pi \widetilde{G}_i)}
      {\nu_i p\sin(\pi/p)}
 \le  \frac{\sin(\pi\min\{\nu_i,1/2\})}
      {\nu_i p\sin(\pi/p)}=
 \frac{\pi h(\nu_i)}{p\sin(\pi/p)}\le \frac{\pi h(\nu)}{p\sin(\pi/p)}= h(\nu)+O(p^{-2})
<\sigma_\nu,
\]
for sufficiently large $p$.

Recall that $\widetilde{C}^\perp$ has dimension $\widetilde{k}=O(1)$ and distance $ \widetilde{d}\ge k-O(1)$ by \cref{fact:hmsy-model}.
We next run the support splitting argument underlying the proof of \cref{lem:wftl} to $\widetilde{C}^\perp$ instead of $C^\perp$.
For sufficiently large $n$, $\widetilde{d}>2\widetilde{k}$.
Choose $J_1\subset I$ with $|J_1|\le\widetilde{k}$ such that
projection $\pi_1$ onto $J_1$ is injective on $\widetilde{C}^\perp$.
Since $\widetilde{d}>|J_1|$, projection onto $I\setminus J_1$ is
also injective. Thus this complement contains a second set $J_2$ with $|J_2|\le\widetilde{k}$ on which the projection $\pi_2$ is injective on $\widetilde{C}^\perp$. Now,
every $0\ne y\in\widetilde{C}^\perp$ has support size at least
$\widetilde{d}-2\widetilde{k}$ outside $J_1\cup J_2$. By
Cauchy--Schwarz and that $w_i\in [0, 1]$,
\begin{equation}
\label{eq:like-bdir}
\begin{aligned}
\sum_{0\ne y\in\widetilde C^\perp}
   \prod_{i\in I}w_i(y_i) & \le
 \sigma_\nu^{\widetilde d-2\widetilde k}
 \sum_{0\ne y\in\widetilde C^\perp}
 \left(\prod_{i\in J_1}w_i(y_i)\right)
 \left(\prod_{i\in J_2}w_i(y_i)\right)\\
&
\le
 \sigma_\nu^{\widetilde d-2\widetilde k}
 \left(
   \sum_{y\in\widetilde C^\perp}
   \prod_{i\in J_1}w_i(y_i)^2
 \right)^{1/2}
 \left(
   \sum_{y\in\widetilde C^\perp}
   \prod_{i\in J_2}w_i(y_i)^2
 \right)^{1/2}\\
&=
 \sigma_\nu^{\widetilde d-2\widetilde k}
 \left(
   \sum_{z\in\pi_{1}(\widetilde C^\perp)}
   \prod_{i\in J_1}w_i(z_i)^2
 \right)^{1/2}
 \left(
   \sum_{z\in\pi_{2}(\widetilde C^\perp)}
   \prod_{i\in J_2}w_i(z_i)^2
 \right)^{1/2}\\
&\le
 \sigma_\nu^{\widetilde d-2\widetilde k}
 \left(
   \sum_{z\in\prod_{i\in J_1}V_i^\star}
   \prod_{i\in J_1}w_i(z_i)^2
 \right)^{1/2}
 \left(
   \sum_{z\in\prod_{i\in J_2}V_i^\star}
   \prod_{i\in J_2}w_i(z_i)^2
 \right)^{1/2}\\
&=
 \sigma_\nu^{\widetilde d-2\widetilde k}
 \left(
   \prod_{i\in J_1}
   \sum_{\xi\in V_i^\star}w_i(\xi)^2
 \right)^{1/2}
 \left(
   \prod_{i\in J_2}
   \sum_{\xi\in V_i^\star}w_i(\xi)^2
 \right)^{1/2}\\
&\le
 \sigma_\nu^{\widetilde d-2\widetilde k}
 \nu^{-|J_1|/2}
 \nu^{-|J_2|/2}\\
&\le
 \sigma_\nu^{\widetilde d-2\widetilde k}
 \nu^{-\widetilde k}.
\end{aligned}
\end{equation}
Now, as $\widetilde{k}=O(1)$, this is $O(\sigma_\nu^k)$. 
Dividing \cref{eq:hmsy-fourier-reduction} by
$\prod_{i=1}^n\nu_i\ge\nu^n$, we obtain
\[
 \FO(0;C,\mbf G)\le
 2\left(\prod_{i\notin I}\nu_i^{-1}\right)
 \sum_{0\ne y\in\widetilde{C}^\perp}\prod_{i\in I}w_i(y_i)
 +e^{-Ln}\nu^{-n}\le
 O(\sigma_\nu^k)+e^{-Ln}\nu^{-n},
\]
since $n-|I|=O(1)$. On the event $\mathsf{Good}^c$, which happens with probability at most $e^{-Ln}$, the crude bound
$\FO(0;C,\mbf G)\le\nu^{-n}$ holds. Since every bound so far is uniform over $\mbf G\in\mc{G}^n$ with $\mb{E}G_i\ge \nu$, 
\[
 \E_C\left[
   \sup_{\substack{\mbf G\in\mc G^n\\\E G_i\ge\nu}}
   \FO(0;C,\mbf G)
 \right]
 \le O(\sigma_\nu^k+e^{-Ln}\nu^{-n})
 =O(\sigma_\nu^k).
\]
The last equality follows from our choice of $L$ and $k\le n$, as desired.
\end{proof}
\begin{remark}[Overcoming the rate-$1/2$ barrier; comparison between \cite{HMSY26} and our earlier approach]
\label{rem:like-bdir}
The key method to bounding the Fourier proxy in \cref{eq:like-bdir} is familiar: we fit two information sets in the coordinates and gain on the support outside. This is exactly the approach to the Fourier proxy bounds in the original leakage resilience works \cite{bdir,mpsw21}, the connection to OPI in \cite{us}, and the proof of \cref{eq:wftl}.

The key insight of \cite{HMSY26} in \cref{fact:hmsy-model} is to leverage the structured leakage $\mbf{G}\in\mc{G}^n$ to construct the auxiliary code $\widetilde{C}^\perp$ on which we apply the bound. The Fourier proxies remain close via \cref{eq:hmsy-fourier-reduction} with high probability, and fitting two information sets $J_1, J_2$ into $\widetilde{C}^\perp$ is much easier since $\dim \widetilde{C}^\perp=\widetilde{k}=O(1)$. 
Indeed, $C^\perp$ has length $n$ and dimension $n-k$, so to fit two information sets, we need $2(n-k)\le n$, which gives exactly the rate-$1/2$ barrier.

From this perspective, we view the overcoming of this barrier in \cite{HMSY26} as analogous to our approach in  \cref{subsec:llr-below,subsec:opi-above} via the $U^3$ Gowers norm that achieved rate $R=1/2-\delta$, except \cite{HMSY26} uses $U^s$ norms for $s>1/R$ to handle any positive rate $R$. On the other hand, the construction in \cite{HMSY26} of the key auxiliary $\widetilde{C}^\perp$ holds only for typical evaluations $\alpha$.
\end{remark}

We now prove the main result of this subsection, which we restate for the reader's convenience.  The proof idea is similar to that of the analogous statement for one-bit LLR in \cite{HMSY26}. We give a short self-contained proof, using the Fourier-analytic framework and the ingredients developed above.

\randomopi*

\begin{proof}
Without loss of generality, we assume $|S_i|=\rho p$ for every $i$. Indeed, the general unequal density case follows via \cref{prop:subsample}. Now, fix any integer $s> \max\{4, 1/R\}$ and fix $\sigma:=\sigma_{\rho/2}$ as guaranteed by \cref{prop:hmsy-relative}. Choose sufficiently small positive constants $\tau=\tau(R, \rho)\in (0, R/2)$ and $\eta=\eta(R, \rho)\in (0, \rho/2)$ to be determined. Without loss of generality assume $t:=\tau n$ is an integer.

By \cref{fact:hmsy-approx}, we can write $\one_{S_i} = G_i+H_i$ for $G_i\in\mc{G}$ and $\nu_i:=\mb{E}G_i$ such that 
\begin{equation}
\label{eq:decomp-param}
 |\nu_i-\rho|\le\eta,\qquad \nu_i\ge\rho/2,
 \qquad \E|H_i|\le\E f_i+\E G_i\le2\rho_p+\eta.
\end{equation}
Since $\rho=\E G_i+\E H_i$, we define $ \Delta_J(C)$ for each $J\subset[n]$ and expand $\FO(0;C,\mbf S)$ as follows:
\begin{equation}\label{eq:rrs-outer-expansion}
\begin{aligned}
 \rho^n\FO(0;C,\mbf S)
 &=
 \left|
   \E_{x\in C}\prod_i(G_i(x_i)+H_i(x_i))
   -\prod_i(\E G_i+\E H_i)
 \right| \le \sum_{J\subset[n]}|\Delta_J(C)|,
\\ \Delta_J(C) & :=
 \E_{x\in C}
   \prod_{i\in J}H_i(x_i)\prod_{i\notin J}G_i(x_i)
 -
 \prod_{i\in J}\E H_i\prod_{i\notin J}\nu_i.
\end{aligned}
\end{equation}
Therefore, it suffices to show that $\rho^{-n}\sum_J \mb{E}_C[\sup_{\mbf{S}}|\Delta_J(C)|]$ is exponentially small for any appropriate functions $G_i$ and $H_i$. Indeed, suppose that, uniformly over every $\mbf{S}$ with $|S_i|=\rho p$, we had
\begin{equation}
\label{eq:rrs-exp-obj}
 \E_C\left[
   \sup_{\mbf{S}}\FO(0;C,\mbf S)
 \right]\le\rho^{-n}\sum_{J\subset [n]} \mb{E}_C\left[\sup_{\mbf{S}}|\Delta_J(C)|\right] \le e^{-\Omega(n)},
\end{equation}
Markov's inequality then gives $
   \sup_{\mbf S}
   \FO(0;C,\mbf S)>e^{-\Omega(n)}$ with probability at most $e^{-\Omega(n)}$, which proves near-optimal OPI discrepancy by \cref{lem:translated-opi-proxy}.
Thus, we need to prove the second inequality in~\cref{eq:rrs-exp-obj}.
To do so, we break the proof down into cases, depending on $|J|$.
\begin{itemize}
\item Suppose $j:=|J|\ge t$, and choose $T\subset J$ with $|T|=t$.
We will choose $\tau$ sufficiently small that
\begin{equation}
\label{eq:tau-req-1}
 n-t\le s(k-t).
\end{equation}
By \cref{fact:block-ul} applied to the factors indexed
by $T$, and by the fact that all remaining factors are $1$-bounded, we have
\[
 \left|
 \E_{x\in C}
   \prod_{i\in J}H_i(x_i)\prod_{i\notin J}G_i(x_i)
 \right|
 \le
 \prod_{i\in T}\|H_i\|_{U^s}
 \le\eta^t.
\]
The other term in the definition \cref{eq:rrs-outer-expansion} of $\Delta_J(C)$ satisfies the same bound as $|\E H_i|\le\eta$ and $\nu_i\le1$.
Therefore, for every fixed $C$ and $J\subset [n]$ with $|J|=j\ge t$,
\begin{equation}\label{eq:rrs-many-residuals}
 |\Delta_J(C)|\le2\eta^t.
\end{equation}
    \item Suppose $j:=|J|<t$.  We consider the smaller Reed-Solomon code $C_J:=\on{RS}_p((\alpha_i)_{i\notin J};k-j)$. Then, $C_J$ has rate at least $R-\tau \ge R/2$. The idea is that on the coordinates $i\not \in J$, we have structured terms $G_i$, which we handle via \cref{prop:hmsy-relative}. Let
    \[
D_J(C):=
 \sup_{\substack{(F_i)_{i\notin J}\in\mc G^{n-j}\\
                 \E F_i\ge\rho/2}}
 \FO(0;C_J,(F_i)_{i\notin J}).
\]
 First, we bound $\Delta_J(C)$ by $D_J(C)$ for each fixed $C$. Sample $x\sim\on{Unif}(C)$ and consider the restriction $x_J$ to coordinates $J$. Since $j=|J|<t<k$ and $C$ is MDS, $x_J$ is uniform on $\F_p^J$. Conditioned on $x_J=u$, the remaining coordinates $(x_j)_{j\not\in J}$ lie on an affine transformed version of $C_J$.
 Namely, there exists $a_{i}(u)\in\F_p, b_i\in\F_p^\times$ such that we have the following equality of distribution:
 \[\{(x_j)_{j\not\in J}\mid x_J=u \}\overset{d}{=} (a_i(u)+b_iz_i)_{i\not\in J}\]
 for $z\sim\on{Unif}(C_J)$. This follows from simple polynomial interpolation. Now, let $G_{i}^{(u)}(z):=G_i(a_i(u)+b_iz)$. Since $\mc{G}$ is closed under affine precomposition by \cref{fact:hmsy-approx} and affine precompositions preserve the mean, then $G_{i}^{(u)}\in \mc{G}$ and $\mb{E}G_i^{(u)}= \mb{E}G_i=\nu_i\ge \rho/2$. Therefore, \[
 \left|
\mathbb E\left[\prod_{i\notin J}G_i^{(u)}(x_i)\mid x_J=u\right]
-\prod_{i\notin J}\nu_i
\right|
\le
D_J(C)\prod_{i\notin J}\nu_i,
 \]
by the definition of $D_J(C)$ and $\FO$.
Consequently,
\[
 |\Delta_J(C)|
 =
 \left|\E_{u\in\F_p^J}
 \left[
   \prod_{i\in J}H_i(u_i)
   \left(
     \E_{z\in C_J}\prod_{i\notin J}G_{i,u}(z_i)
     -\prod_{i\notin J}\nu_i
   \right)
 \right]\right|
 \le
 D_J(C)
 \prod_{i\in J}\E|H_i|
 \prod_{i\notin J}\nu_i.
\]

Using $\nu_i\le\rho+\eta$ and
$\E|H_i|\le2\rho+\eta\le2(\rho+\eta)$, we obtain
\begin{equation}\label{eq:rrs-few-residuals}
 \rho^{-n}|\Delta_J(C)|
 \le
 2^j(1+\eta/\rho)^nD_J(C)
 \le
 2^je^{\eta n/\rho}D_J(C).
\end{equation}
Since $C_J$ has rate at least $R/2$, \cref{prop:hmsy-relative} implies that there is a constant
$\sigma=\sigma_{\rho/2}\in (0, 1)$ and a constant $K$ independent of $J$ and $n$, such that for every $J\subset [n]$ with $|J|=j<t$,
\begin{equation}\label{eq:rrs-shortened-average}
 \E_C D_J(C)\le K\sigma^{k-j}.
\end{equation}
\end{itemize}

In both cases, 
the bounds \cref{eq:rrs-shortened-average,eq:rrs-few-residuals,eq:rrs-many-residuals} hold uniformly over all lists $S_i$ with 
$|S_i|=\rho p$. Taking their supremum in
\cref{eq:rrs-outer-expansion}, and then averaging over $C$ (i.e. over the evaluation points $\alpha$), gives
\begin{equation}\label{eq:rrs-outer-bound}
\begin{aligned}
 \E_C\left[
   \sup_{\mbf{S}}
   \FO(0;C,\mbf S)
 \right]
 &\le
\rho^{-n}\cdot 2\eta^t\cdot 2^n +e^{\eta n/\rho}
 \sum_{\substack{J\subset[n]\\|J|<t}}
 2^{|J|}\E_C D_J(C)
 \\
 &\le
2(2/\rho)^n\eta^t+ K\sigma^k e^{\eta n/\rho}
 \sum_{j<t}\binom nj(2/\sigma)^j.
\end{aligned}
\end{equation}
By \cref{eq:rrs-exp-obj}, it suffices to show that both terms are exponentially small, by choosing our parameters $\tau$ and $\eta$. We first choose $\tau$ sufficiently small to satisfy \cref{eq:tau-req-1} and that binary natural log entropy $H(z):=-z\log z-(1-z)\log(1-z)$, we have that
\begin{equation}
\label{eq:tau-req-2}
H(\tau)+\tau\log(2/\sigma)<\frac{R}{4}\log(1/\sigma).
\end{equation}
Next, choose $\eta\in(0,\rho/2)$ sufficiently small that
\begin{equation}
\label{eq:eta-req}
\eta <\frac{\rho R}{4}\log(1/\sigma)\quad\text{and}\quad
 \ln(2/\rho)+\tau\log\eta <-\frac{R}{2}\log(1/\sigma).
\end{equation}
Now, \cref{eq:tau-req-2,eq:eta-req} shows that both terms in \cref{eq:rrs-outer-bound} are
$e^{-\Omega(n)}$, via the entropy estimate
\[
 \sum_{j<t}\binom nj(2/\sigma)^j
 \le
 \exp\left(
   n\left[H(\tau)+\tau\log(2/\sigma)+o(1)\right]\right).
\]
Note that the parameter dependencies are noncircular: $\sigma$ depends only on $\rho/2$,
so we choose $s,\tau,\eta$ first, then fix the class $\mc G$ from
\cref{fact:hmsy-approx}, and the large constant in the requirement $\log p\ge n^{\Omega(1)}$.
\end{proof}

\section{Quantum Algorithms for OPI}
\label{sec:algo}
In this section, we leave random Reed-Solomon codes aside and develop quantum algorithms for worst-case OPI for arbitrary Reed-Solomon codes. Throughout, we use the machinery and notation developed in \cref{subsec:expansion-recap} and the quantum algorithm model in \cref{subsec:subsampling}.
For convenience, the state analysis first assumes common known density $\rho$; the final proofs use \cref{prop:oracle-thinning} to handle unequal densities.

\subsection{Recap of the Algorithm in \cite{sjo}}
\label{sec:algo-recap}
We begin by recapping the algorithmic results of \cite{sjo}.  Recall from \cref{eq:dqi-state} that
\[ F_\ell(x) := \sum_{j=\ell-\sigma}^\ell \frac{q_j(x)}{\sqrt{{n \choose j}} },\]
where $\sigma = \lfloor \log\log n \rfloor$.
The algorithm of \cite{sjo} samples from the distribution proportional to $F_\ell(x)^2$ on $C$ and prepares the DQI state $|\varphi_\ell\rangle \propto \sum_{x \in C} F_\ell(x)\ket{x}$ as in \cref{eq:dqi-state}. First, define normalizations
\begin{equation}
\label{eq:sw-target-state}
Z_\ell:=\E_{x\in C}|F_\ell(x)|^2,
\quad\text{and}\quad
Z_0:=\E_{x\in\F_p^n}|F_\ell(x)|^2.
\end{equation}
We now state a main result of \cite{sjo} that we will use as a black-box.
\begin{theorem}[{\cite[Theorem~5.2 and Propositions~6.2--6.3]{sjo}}]
\label{thm:sjo-main}
Let $C \subset \F_p^n$ be a linear code.  Suppose some deterministic classical algorithm that takes as input a received word $r \in \F_p^n$ and returns the set $\mathcal{L}(r) = \{c \in C^\perp : |r-c| \leq \ell\}$, with running time $T = T(n, \log p)$. 

If $|\mathcal{L}(r)| \leq L$ for every $r \in \F_p^n$, then there is a quantum algorithm that prepares $\ket{\varphi_\ell}$ exactly under coherent membership oracle access to the input lists $S_i$ with known common density $\rho\in (0, 1)$, and with running time $\poly(T,n,\log p)$ and flagged success probability $\frac{Z_\ell}{LZ_0}$.
\end{theorem}

Before explaining how we apply \cref{thm:sjo-main}, we briefly walk through the algorithm behind it to explain how to prepare the state $\ket{\varphi_\ell}$ in a few steps.  We note that the details of the algorithm are not needed for our result, as we use \cref{thm:sjo-main} as a black box.  But we include them for completeness, and in particular because we have synthesized a few results of \cite{sjo} (Theorem 5.2 and Propositions 6.2 and 6.3 in that paper) in the statement of~\cref{thm:sjo-main}, and we wish to explain how they fit together to give \cref{thm:sjo-main}.  Our synthesis is very similar to the approach in~\cite[Theorem 7.1]{sjo}.

\paragraph{Step 1: Local Preparations.}
First, the algorithm prepares the Fourier contributions separately. Let
$B\in\F_p^{n\times k}$ generate $C$,
so $\ker B^\top=C^\perp$. With the coefficients $u_j$ from
\cref{eq:dqi-u},  so that $F_\ell(x)=\sum_j u_jq_j(x)$ and $\ket{\varphi_\ell}=\ket{\varphi_\ell(u)}$. Write its Fourier amplitudes as
\begin{equation}
\label{eq:sw-error-amplitudes}
\widehat F_\ell(y) :=u_{|y|}\prod_{i\in\supp(y)}\wh g_i(y_i),\quad\text{where}\quad g_i(x) = \frac{\mathbf{1}_{S_i}(x) - \rho}{\sqrt{\rho(1 - \rho)}}
\end{equation}
as usual, per the common density case where $\rho_i=\rho$.
For each coordinate, \cite[Proposition~6.3]{sjo} prepares
the coordinate $i$ state 
\[ \sum_{\xi\ne0}\wh g_i(\xi)\ket\xi.\]
Now, prepare a uniform superposition over support sizes
$j\in[\ell-\sigma,\ell]$ and, for each $j$, over all $j$-subsets,
and place the local Fourier states on that subset.
Since the resulting vector $y$
determines both $j=|y|$ and the chosen subset $\on{supp}y$,
we can erase auxiliary labels and compute syndrome $B^\top y$.
By \cite[Proposition~6.2]{sjo}, we obtain the state
\begin{equation}
\label{eq:sw-labeled-error-state}
\frac1{\sqrt{Z_0}}\sum_y \widehat F_\ell(y) \ket y\ket{B^\top y} =\frac1{\sqrt{Z_0}}\sum_s \left(\sum_{\substack{B^\top y =s \\ |y|\le \ell}} \widehat F_\ell(y) \ket y\right) \ket{s}.
\end{equation}
Here, Parseval's identity on the ambient space gives normalization
$Z_0=\sum_y|\widehat F_\ell(y) |^2$.
Note that the sum over $y$ versus over $y$ with $|y|\le \ell$ are identical since $u_{|y|}$ and $\widehat F_\ell(y) $ vanish otherwise.

\paragraph{Step 2: Applying the Decoder.}
This is the key step: we must add the contributions having the same syndrome
$s$ coherently and erase the states $\ket{y}$. For each $s$, choose $r_s$ with $B^\top r_s=s$, so
\[
\{y:B^\top y=s, |y|\le\ell\}
=\{r_s-c:c\in C^\perp, |r_s-c|\le\ell\}.
\]
Thus, a classical list decoder that finds all codewords within radius $\ell$ of $r_s$ will enumerate all errors $y$ that appear in the sum in \cref{eq:sw-labeled-error-state} for the same syndrome $s$. Now, the reversible indexing construction
of \cite[Lemma~5.1]{sjo} replaces the error register by the
word's index in a canonically ordered list padded to some common length $L$.  That is, for some $j_s(y)\in [L]$, from the state in \cref{eq:sw-labeled-error-state}, we can efficiently prepare the state 
\begin{equation}
\frac1{\sqrt{Z_0}}\sum_s \left(\sum_{\substack{B^\top y =s \\ |y|\le \ell}} \widehat F_\ell(y) \ket{j_s(y)}\right) \ket{s}.
\end{equation}
To erase the index, we take the projective measurement as in \cite[Theorem~5.2]{sjo} given by 
\[ \ket{+_L}\bra{+_L}\otimes I\quad\text{where}\quad \ket{+_L} :=\frac{1}{\sqrt{L}}\sum_{j\in [L]}\ket{j}.\]
As $\braket{+_L|j_s(y)} = L^{-1/2}$, the success probability and the unnormalized successful branch are 
\begin{equation}
\label{eq:success-branch}
\frac{1}{LZ_0}\sum_s \left|\sum_{B^\top y =s}\widehat F_\ell(y) \right|^2=\frac{Z_\ell}{LZ_0}\quad\text{and}\quad \frac1{\sqrt{LZ_0}}
\sum_s\left(\sum_{B^\top y=s}\widehat F_\ell(y) \right)\ket s,
\end{equation}
respectively. Here, we used Parseval's identity on the message space to obtain
\[
Z_\ell=\E_{z\in\F_p^k}|F_\ell(Bz)|^2
=\sum_s\left|\sum_{B^\top y=s}\widehat F_\ell(y) \right|^2.
\]
Note that padding ensures the common factor $L^{-1/2}$; using a different length normalization for each list
would change the target amplitudes.

\paragraph{Step 3: Inverse Fourier Transform.}
Finally, the inverse Fourier transform and the encoding
$z\mapsto Bz$ turn \cref{eq:success-branch} into the scaled version of $\ket{\varphi_\ell}$, with squared norm $Z_\ell/(LZ_0)$
\begin{equation}
\label{eq:sw-coherent-fiber-amplitude}
\frac1{\sqrt{LZ_0|C|}}\sum_{x\in C}F_\ell(x)\ket x \propto \ket{\varphi_\ell}.
\end{equation}
Indeed, the square-norm calculation follows \cref{eq:success-branch}. The inverse Fourier transform of the syndrome amplitudes is
exactly $F_\ell(Bz)/\sqrt{|C|}$, and encoding $z\mapsto Bz$
gives \cref{eq:sw-coherent-fiber-amplitude}. This is via Fourier inversion:
\[
F_\ell(Bz)
=
\sum_y \widehat F_\ell(y)  e_p(-\langle Bz,y\rangle)
=
\sum_s e_p(-\langle z,s\rangle)
\left(\sum_{B^\top y=s}\widehat F_\ell(y) \right).
\]

\subsection{Improved Semicircle Law}
\label{sec:algo-partial}
In this section, we prove
\cref{thm:main-algo}(1). We will reduce to the case where $|S_i|=\rho p$ for every $i$ via \cref{prop:oracle-thinning}, so fix $\rho \in (0,1)$.  Also fix a rate  $R \in (0, 1)$, and an arbitrarily small error parameter $\epsilon>0$. We wish use the algorithm in \cref{thm:sjo-main} to prepare $\ket{\varphi_\ell}$, where $\ell = \lfloor \alpha n\rfloor$ and we choose
\begin{equation}
\label{eq:def-partial-alpha}
\alpha \in (0,  \min(1-\rho, R)) \quad\text{such that}\quad \on{SCL}_\rho(\alpha)\ge \on{SCL}_\rho(R)-\frac{\eps}{10}.
\end{equation}

To apply \cref{thm:sjo-main} effectively and prove \cref{thm:main-algo}(1), we must first establish both that $LZ_0/Z_\ell$
is polynomially bounded (so that the success probability is not too small) and that the prepared state yields a
sufficiently good solution with nonnegligible probability. 
We verify these conditions in the following below.

\begin{lemma}
\label{lem:z-partial-check}
For any $0< \alpha \le 1-\rho$ with $\alpha <R$, let $\ell=\lfloor \alpha n\rfloor$ and define $\ket{\varphi_\ell}, F_\ell$ as in \cref{eq:dqi-state}. Then, $Z_0=\sigma+1$ and $Z_\ell = \sigma+1+e^{-\Omega(n)}$. Moreover, the expected satisfaction ratio of the outcome $x$ obtained by measuring the DQI state $\ket{\varphi_\ell}$ is $\on{SCL}_\rho(\alpha)+o(1)$.
\end{lemma}
\begin{proof}
By orthogonality in \cref{eq:sw-local-and-orthogonality}, we see that $Z_0 = \sigma + 1$.
We verify \cref{eq:key-decay} to apply \cref{thm:sw-main} to bound $Z_\ell$ and the expected satisfaction ratio. Let
$\pi_j(\alpha):=\mb{P}(B=j)=\binom nj\alpha^j(1-\alpha)^{n-j}$ for $B\sim \on{Bin}(n,\alpha)$. Note that $h(t):=2\sqrt{t(1-t)}+t|\beta_\rho|$ is nondecreasing for $t\in [0, 1-\rho]$, so it attains maximum at $h(1-\rho)=b_\rho$. 
Now, substituting $z=\sqrt{\alpha/(1-\alpha)}$ into
\cref{eq:N-bound} gives
\begin{equation}
\label{eq:b-rho-expansion}
\frac{\widetilde N_\rho(j,j';t)}
{\sqrt{\binom nj\binom n{j'}}}
\le
\frac{\left(2\sqrt{\alpha(1-\alpha)}+\alpha|\beta_\rho|\right)^t}
{\sqrt{\pi_j(\alpha)\pi_{j'}(\alpha)}}
\le\frac{b_\rho^t}{\sqrt{\pi_j(\alpha)\pi_{j'}(\alpha)}}=\Theta( b_\rho^t \sqrt{n}),
\end{equation}
where we justify the last step as follows: for $j, j'\in [\ell-\sigma-1, \ell+1]$, recalling that $\sigma = \lfloor \log\log n \rfloor$, we have that $|j-\alpha n|=O(\log\log n)$, so Stirling's formula gives $\pi_j(\alpha)=\Theta (n^{-1/2})$ and $\pi_{j'}(\alpha)=\Theta (n^{-1/2})$ uniformly for such pairs $j, j'$. Thus, by \cref{lem:key-fourier}, for every pair in
\cref{eq:key-decay} and every $k<t\le\min\{n,j+j'\}$,
\[
\frac{\widetilde N_\rho(j,j';t)}
{\sqrt{\binom nj\binom n{j'}}}
\abs{\E_{x\in C}q_t(x)}
\le O(\sqrt n)
\sum_{\substack{y\in C^\perp\\|y|=t}}
b_\rho^{|y|}\prod_{i\in\supp(y)}|\wh g_i(y_i)|
\le e^{-\Omega(n)},
\]
where the last inequality is by \cref{lem:N-proxy-bound}. Its support assumption holds uniformly here since $t\le\min(n,2\ell+2)$ and $2k-(2\ell+2)=2(R-\alpha)n-O(1)=\Omega(n)$. Now, \cref{thm:sw-main} gives
$Z_\ell=\sigma+1+O(e^{-\Omega(n)})$ and expected satisfaction ratio
$\SCL_\rho(\ell/n)+o(1)=\SCL_\rho(\alpha)+o(1)$.
\end{proof}

 Since $\ell/n\sim \alpha <R$, the list decoder in \cref{thm:rs-ld} satisfies the condition in \cref{thm:sjo-main} with $L=n^{O(1)}$. Next, we put this together to prove \cref{thm:main-algo}(1) via the algorithm in \cref{thm:sjo-main}. 

 \begin{proof}[Proof of \cref{thm:main-algo}(1)]
 First suppose that all lists have common exact density $\rho$. Since $C^\perp$ has rate $1-R$ and $\alpha<R$ with a fixed
gap, \cref{thm:rs-ld,thm:rs-lr} give a deterministic list decoder
with polynomial running time and a uniform list bound $L=n^{O(1)}$, as required by \cref{thm:sjo-main}.
By \cref{lem:z-partial-check,eq:def-partial-alpha}, measuring
$\ket{\varphi_\ell}$ gives expected satisfaction at least
$\SCL_\rho(R)-\varepsilon/2$ for sufficiently large $n$.
Since satisfaction lies in $[0,1]$, the probability of obtaining
satisfaction at least $\SCL_\rho(R)-\varepsilon$ is at least
$\varepsilon/2$ by applying Markov's inequality to $1-s(x)$.

 \cref{lem:z-partial-check} implies that $Z_\ell/Z_0\ge1/2$, so the algorithm from \cref{thm:sjo-main} prepares the state $\ket{\varphi_\ell}$ successfully with
probability at least $1/(2L)$. 
By the above reasoning, 
this returns a solution
with satisfaction at least $\SCL_\rho(R) - \eps$ with probability at least
$\varepsilon/(4L)$. The final algorithm then runs this algorithm $T = O(L\log n)$ times independently, 
 evaluates the successful outputs' scores exactly, and
retains the best. The fact that $L = n^{O(1)}$ gives
the claimed running time, and shows that the success probability is at least
\[ 1 - \left( 1 - \frac{ \eps }{4L } \right)^T \geq 1 - \frac{1}{n^{\Omega(\eps)}} = 1 - o(1).\]

Finally, if lists sizes satisfy $|S_i|\ge \rho p$, the common-size algorithm above succeeds on every tuple of $\rho p$-element lists. For truth-table input, trim each list to size $\rho p$; for oracle input, apply \cref{prop:oracle-thinning} to this common-size algorithm. Both reductions preserve satisfaction on the original lists and give the claimed high-probability guarantee in polynomial time.
\end{proof}

Finally, we leverage \cref{lem:z-partial-check} to prove \cref{thm:main-existential}, which we restate below.
\mainexistential*

Before proving \cref{thm:main-existential}, which give a quick overview of the proof.  Choose a set $T$ of $t$
coordinates and consider the DQI state $\ket{\varphi_\ell}$ on the OPI instance with constraints $T^c$. %
The state $\ket{\varphi_\ell}$ induces a distribution on codewords $x$ in the punctured code, and hence a distribution on codewords $x \in C$ (the original code of length $n$). To see this, lift each punctured word uniformly to its fiber in $C$ and observe that all fibers have the same size. Condition on the event that $x \in C$ from this distribution satisfies all constraints indexed by $T$, using the normalized indicator
\begin{equation}
    A_T(x):=\rho^{-t}\prod_{i\in T}\one\{x_i\in S_i\}.
\end{equation}
The resulting conditioned state has amplitude $\sqrt{A_T}F_\ell$.
We follow the proof of \cref{lem:z-partial-check} to show that this conditioned state satisfies every constraint in $T$ and satisfies a double semicircle law fraction of constraints on $T^c$ whenever $t+2\ell\le 2k-\Omega(n)$, for a total satisfaction ratio of 
\[
    \frac{\E_C[sA_TF_\ell^2]}{\E_C[A_TF_\ell^2]}
    =
    \frac tn+\frac{n-t}{n}
       \SCL_\rho\left(\frac{\ell}{n-t}\right)+o(1).
\]
Then, optimizing $t, \ell$ subject to $t+2\ell\le 2k-\Omega(n)$ gives satisfaction ratio arbitrarily close to $V_\rho(R)$. For simplicity, in the proof below we set $t$ and $\ell$ to be these optimizers from the outset.

\begin{proof}[Proof of \cref{thm:main-existential}]
By thinning to a common list size (see \cref{subsec:subsampling}), it suffices to consider a common
density  $\rho$.
Choose a fixed $\eta\in(0,R-\rho)$ sufficiently small relative to $\eps$, and set
\[
    t:=\left\lfloor
          \frac{2(R-\eta-\rho)}{1-2\rho} n\right\rfloor\]
Fix any set $T\in\binom{[n]}{t}$, let $N:=n-t$, and let $\ell:=\lfloor\rho N\rfloor$. Then, we see that
\begin{equation}
\label{eq:AT-gap}
N=\Theta(n)\quad\text{and}\quad 
t+2\ell=2(R-\eta)n\le 2k-\Omega(n).
\end{equation}

Now, let $C|_{T^c} \subset \F_p^{T^c}$ be the puncturing of $C$ to the coordinates indexed by $T^c$.
Consider the OPI instance given by $C|_{T^c}$ with the corresponding input lists $S_i$ for $i\in T^c$. We will consider the DQI state $\ket{\varphi_\ell}$ on this instance with parameters $\ell:=\lfloor\rho N\rfloor, \sigma:=\lfloor \log \log N\rfloor$.
That is,  let
\[ \ket{\varphi_\ell} \propto \sum_{x \in C|_{T^c}} F_\ell(x) \ket{x},\]
where
\begin{equation}
\label{eq:puncture-param}
F_\ell(x):=\sum_{j=\ell-\sigma}^\ell\frac{q_j(x)}{\sqrt{\binom Nj}}\quad\text{where}\quad q_j(x):=\sum_{B\in\binom{T^c}j}\prod_{i\in B}g_i(x_i)\quad\text{and}\quad  g_i(x_i):=
    \frac{\mathbf{1}_{S_i}(x_i)-\rho}
         {\sqrt{\rho(1-\rho)}}.
\end{equation}
Here, we abuse notation and overload $F_\ell, q_j, g_i$ to refer to the versions defined on the punctured instance for the rest of this proof. However, we distinguish between the satisfaction ratios \[
\widetilde{s}(x):=\frac{1}{N}\sum_{i\not\in T}\one_{S_i}(x_i)\quad\text{and}\quad s(x)=\frac{1}{n}\sum_{i=1}^n\one_{S_i}(x_i).
\]
We consider the distribution $\mu$ on $C$ where $\mu(x) \propto A_T(x)F_\ell(x)^2$, with $ A_T(x):=\rho^{-t}\prod_{i\in T}\one\{x_i\in S_i\}$.
We think of $\mu$ as the distribution on $x \in C$ (or equivalently, on polynomials of degree less than $k$) given by the DQI state $\ket{\varphi_\ell}$ for the punctured instance to coordinates $T^c$, conditioned on satisfying every constraint in $T$. 
For $h:\F_p^n\to \R$, we adopt shorthand notation $\E h:=\mb{E}_{x\in \F_p^n}[h(x)]$ and $\E_C h :=\mb{E}_{x\in C}[h(x)]$.
Then, the expected satisfaction $s(x)$ for $x\sim\mu$ is given by
\begin{equation}
\label{eq:def-tilt}
\mb{E}_{x\sim \mu} [s(x)]=
\frac{\mb{E}_{C}[sA_TF_\ell^2]}{\mb{E}_{C}[A_TF_\ell^2]}=\frac{t}{n}+\frac{n-t}{n}\frac{\mb{E}_{C}[\widetilde{s}A_TF_\ell^2]}{\mb{E}_{C}[A_TF_\ell^2]}
\end{equation}

To compute the numerator and denominator, we compute them first over $x\in\F_p^n$ rather than $x\in C$, and argue that passing from the uniform distribution on $\F_p^n$ to that on $C$ changes the expectations by an exponentially small amount. By \cref{eq:sw-local-and-orthogonality}, and by the fact that $q_j, q_j$ depend only on the coordinates $x_i$ for $i\in T^c$ while $A_T$ depends only on coordinates $x_i$ for $i\in T$, we see that
\begin{equation}
\label{eq:ambient-AT}
\mb{E}[A_Tq_jq_{j'}]=\mb{E}[q_jq_{j'}]=\binom Nj \one\{j=j'\}.
\end{equation}
Therefore, conditioning on $A_T$ preserves the expectations over $\F_p^n$, i.e.
\begin{equation}
\label{eq:ambient-scl}
\begin{aligned}
\E[A_TF_\ell^2]&=\E[F_\ell^2]=\sigma+1,\\
\E[\widetilde{s}A_TF_\ell^2]&=\E[\widetilde{s}F_\ell^2] =[\on{SCL}_\rho(\rho)+o(1)](\sigma+1).
\end{aligned}
\end{equation}
where the semicircle law is evaluated at $\rho$ since $\ell\sim\rho N$. The equalities follow
\cref{eq:sw-local-and-orthogonality} and the three-term recurrence in
\cref{lem:layer-recursion}.
Next, recall \cref{lem:N-proxy-bound}. As $\rho\le 1/2$ by assumption, $b_\rho =\lambda=\sqrt{(1-\rho)/\rho}$. Now, apply the proof of \cref{lem:z-partial-check} verbatim to the punctured instance and with $N$ in place of $n$ and $\alpha=\rho$ to obtain the following analog of \cref{eq:b-rho-expansion}: for $\ell-\sigma-1\le j, j'\le \ell+1$,
\[
    \frac{q_j(x)q_{j'}(x)}{\sqrt{\binom Nj\binom N{j'}}}
    =
    \sum_{B\subset T^c}c_B\prod_{i\in B}g_i(x_i),
\]
where $c_\emptyset=\one\{j=j'\}$, $c_B=0$ if $|B|>j+j'$, and otherwise $|c_B|=O(b_\rho^{|B|}\sqrt{N}) = O(\lambda^{|B|}\sqrt{N})$.
Observe that $A_T(x)
    =
    \prod_{i\in T}(1+\lambda g_i(x_i))$ by definition of $g_i$. Applying this and taking expectations over $x\in C$, 
\[
 \frac{\mathbb{E}_{C}[A_Tq_jq_{j'}]}
           {\sqrt{\binom Nj\binom N{j'}}} = c_\emptyset+\sum_{\emptyset\ne D\subset [n]}\lambda^{|D\cap T|}c_{D\cap T^c}\mb{E}_{x\in C}\left[\prod_{i\in D}g_i(x_i)\right].
\]
Now, we obtain the key comparison inequality that says switching distributions from $x\in\F_p^n$ to $x\in C$ incurs exponentially small error for the relevant expectations in \cref{eq:def-tilt}:
\begin{align*}
\left|\frac{\mathbb{E}_{C}[A_Tq_jq_{j'}]-\mb{E}[A_Tq_jq_{j'}]}
           {\sqrt{\binom Nj\binom N{j'}}} \right| & =\left|\frac{\mathbb{E}_{C}[A_Tq_jq_{j'}]}
           {\sqrt{\binom Nj\binom N{j'}}} -\one\{j=j'\}\right| 
           \\ & \le
  O(\sqrt{N})
    \sum_{\substack{0\ne y\in C^\perp\\|y|\le t+j+j'}}
       \prod_{i\in\operatorname{supp}(y)}
          \lambda|\widehat g_i(y_i)|
          \\ & \le e^{-\Omega(n)}.
\end{align*}
We justify the steps as follows: the first step holds by \cref{eq:ambient-AT}, the second step follows from \cref{lem:key-fourier} and our computation of $c_B$ above, and the last step is given by \cref{lem:N-proxy-bound} where we check its hypothesis holds by \cref{eq:AT-gap} and that $j, j'\le \ell+1\sim\rho N$, so $2k-(t+j+j')\ge2\eta n-O(1)$. Now, by expanding $F_\ell^2$ and $\widetilde{s}F_\ell^2$ in terms of $q_j$ and $q_{j'}$ via \cref{lem:layer-recursion}, we see that 
\[
\left|\E_C[\widetilde{s}A_TF_\ell^2] - \E[\widetilde{s}A_TF_\ell^2]\right|\le e^{-\Omega(n)} \quad\text{and}\quad
\left|\E_C[A_TF_\ell^2]-\E[A_TF_\ell^2]\right|\le e^{-\Omega(n)}.
\]
Combining with \cref{eq:ambient-scl} and \cref{eq:def-tilt}, we see that
\begin{align*}
\mb{E}_{x\sim \mu} [s(x)]& = \frac{t}{n}+\frac{n-t}{n}\frac{\mb{E}_{C}[\widetilde{s}A_TF_\ell^2]}{\mb{E}_{C}[A_TF_\ell^2]}
\\ & =  \frac tn+\frac {n-t}n \on{SCL}_\rho(\rho)+o(1)\\
    &=  1-\frac {n-t}n(1-2\rho)^2+o(1)\\
    &=
    1-(1-2\rho)(1-2R)-2(1-2\rho)\eta+o(1).
\end{align*}
where the computation follows our choice of parameters of $\ell$ and $t$. Hence, some $x\in C$ has score
\[
    s(x)\ge
    1-(1-2\rho)(1-2R)-\varepsilon,
\]
upon choosing $\eta$ sufficiently small and then $n$ sufficiently large. This proves the theorem.
\end{proof}

\begin{remark}
\label{rem:exist-no-algo}
The crucial idea in the proof is the puncturing to $T^c$ and conditioning via $A_T$. In comparison to the proof of \cref{thm:main-opi}(3) (near-optimal OPI discrepancy in the case when $R + \rho > 1$), which also uses \cref{lem:N-proxy-bound}, here $A_TF_\ell^2$ has support degree at most $t+2\ell$---rather than $2(t+\ell)$---by paying for the coordinates in $T$ once. We are then able to apply \cref{lem:N-proxy-bound} with $t+2\ell<2k-\Omega(n)$ even when $t+\ell>k$. This is the parameter regime in \cref{thm:main-existential}, and why the proof of \cref{thm:main-opi}(3) for $R+\rho>1$ does not apply.

This parameter regime also explains why the satisfaction ratio in \cref{thm:main-existential} cannot immediately be attained quantum algorithmically.
One might be tempted to use the algorithm in \cref{sec:algo-recap} to prepare the state whose corresponding distribution is $\mu$. This state has amplitudes given by $\sqrt{A_T}F_\ell$, and directly using that implementation would require list decoding of $C^\perp$ through radius $t+\ell$. The dual has rate $1-R$, so the list-decoding capacity theorem implies that it cannot be reliably list-decoded (with polynomial list sizes) beyond radius $R$.  On the other hand, our parameters give
\[
    \frac{t+\ell}{n}
    =
    R-\eta+\frac{t}{2n}>R
\]
for sufficiently small fixed $\eta$.
Therefore, the associated list sizes in the list-decoding steps are then exponentially large, and our implementation in \cref{sec:algo-recap} fails to prepare this state efficiently.
\end{remark}

\subsection{Perfect Solution at Capacity}
\label{sec:algo-perfect}
Lastly, we prove \cref{thm:main-algo}(2). In this case, assume $R+\rho>1$. Then, the guarantee in \cref{thm:main-algo}(1) gives an \emph{near} perfect solution, that is, a satisfaction ratio at least $1-\varepsilon$ for every fixed $\varepsilon>0$.
We refine the analysis to show that in this case we can find a perfect solution that satisfies every constraint. 

We run the same algorithm (\cref{thm:sjo-main}) with the capacity list decoder (\cref{thm:rs-ld}) with 
\begin{equation}
\label{eq:def-perfect-alpha}
\ell := \lfloor (1-\rho)n\rfloor
\end{equation}
to prepare $\ket{\varphi_\ell}$.
Both the theorems apply as $\ell<k-\Omega(n)$, as checked in \cref{lem:z-partial-check}. Now, we show that measuring $\ket{\varphi_\ell}$ gives decent probability of obtaining a perfect solution.

\begin{lemma}
\label{eq:perfect-prob}
Fix $\ell$ from \cref{eq:def-perfect-alpha} and $\ket{\varphi_\ell}$ from \cref{eq:dqi-state}. Then, the probability of obtaining a perfect solution upon measuring $\ket{\varphi_\ell}$ is
\begin{equation}
(1+o(1))\frac{\sigma+1}{\sqrt{2\pi n\rho(1-\rho)}}\ge n^{-O(1)}
\end{equation}
\end{lemma}

\begin{proof}
For every ``perfect'' codeword $x\in C\cap\mbf S$, all coordinates $i \in [n]$
satisfy $g_i(x_i)=\sqrt{(1-\rho)/\rho}$. Thus, with the
binomial probabilities $\pi_j$ defined in \cref{eq:def-pi-j},
\[
F_\ell(x)
=\sum_{j=\ell-\sigma}^{\ell}
\sqrt{\binom nj}\left(\frac{1-\rho}{\rho}\right)^{j/2}
=\rho^{-n/2}\sum_{j=\ell-\sigma}^{\ell}\sqrt{\pi_j}.
\]
Since the squared norm of $\sum_{x\in C}F_\ell(x)\ket{x}$
is $|C|Z_\ell$, the probability of a perfect measurement is
\[
\frac{|C\cap\mbf S|}{|C|\rho^n}
\frac1{Z_\ell}
\left(\sum_{j=\ell-\sigma}^{\ell}\sqrt{\pi_j}\right)^2.
\]
The first factor is $1+O(e^{-\Omega(n)})$ by
\cref{lem:terminal-dense-count}. Because $1-\rho<R$,
\cref{lem:z-partial-check} applies at $\alpha=1-\rho$ and gives
$Z_\ell=\sigma+1+O(e^{-\Omega(n)})$.
Each $j$ in the summand satisfies $|j-(1-\rho)n|=O(\log\log n)$, so Stirling's formula gives that, uniformly over all $j$
\[
\pi_j=\frac{1+o(1)}{\sqrt{2\pi n\rho(1-\rho)}}.
\]
The squared sum is therefore
$\sim(\sigma+1)^2/\sqrt{2\pi n\rho(1-\rho)}$, which proves the desired bound.
\end{proof}

Now, we can repeatedly prepare $\ket{\varphi_\ell}$ and measure until the outcome solution $x$ satisfies every constraint. This will prove \cref{thm:main-algo}(2); we give the formal proof below.
\begin{proof}[Proof of \cref{thm:main-algo}(2)]
 First suppose that all lists have common exact density $\rho$.
Since $R+\rho>1$, the choice
$\ell=\lfloor(1-\rho)n\rfloor$ satisfies $\ell<k-\Omega(n)$.
Thus \cref{thm:rs-ld} supplies the decoder required by
\cref{thm:sjo-main}, with $L=\poly(n)$.

By \cref{lem:z-partial-check}, $Z_\ell/Z_0\ge1/2$, so each
preparation succeeds with probability at least $1/(2L)$.
Conditioned on success, \cref{eq:perfect-prob} implies that measurement yields a perfect solution with probability
$n^{-O(1)}$.
Hence each attempt returns a perfect solution with
inverse-polynomial probability $n^{-O(1)}$.  Polynomially many independent
attempts, followed by exact verification of every constraint,
give a perfect solution with high probability and the claimed
running time.

Finally, if lists have different sizes but $|S_i|\ge \rho p$ for all $i$, the common-size algorithm described above above succeeds on every tuple of $\rho p$-element lists. For truth-table input, trim each list to size $\rho p$; for oracle input, apply \cref{prop:oracle-thinning} to this common-size algorithm. Both reductions preserve satisfaction on the original lists and give the claimed high-probability guarantee in polynomial time.
\end{proof}

\begin{remark}[An alternative algorithm]
As an alternative proof of \cref{thm:main-algo}(2), 
could instead sample the state with amplitudes $H_\ell(x)$ from the proof of \cref{thm:main-opi}(3), defined in \cref{eq:def-H-ell}, where  $(1-\rho)n+\Omega(n)\le\ell\le k-\Omega(n)$. By \cref{lem:key-truncation}, measuring it attains a perfect solution with high probability, since 
\[ \frac{n\E_C[(1-s)H_\ell^2]}{\E_C[H_\ell^2]}=e^{-\Omega(n)}.\]
The proof of \cref{thm:sjo-main} (adapted to use the coefficients defining $H_\ell$) implies that we can prepare this state efficiently.  The reason that we prove the theorem as we do (with coefficients from the original DQI state) is for simplicity, so that \cref{thm:sjo-main} applies in its stated form.
\end{remark}

\section{Quantum Advantage for OPI Relative to a Membership Oracle}
\label{sec:advantage}
In this section, we construct structured OPI ensembles for which efficient quantum algorithms find perfect solutions, but efficient classical algorithms cannot find decent solutions with non-negligible probability. The proof is based on the random-oracle search problem of
Yamakawa and Zhandry \cite{yz} (henceforth referred to as the \emph{YZ problem}). We give two constructions that encode that problem as an ensemble of OPI instances:
\begin{enumerate}
    \item The first construction uses a simple encoding of the YZ problem, and can have arbitrary rate and upper bound on classical algorithm performance, but the input lists do not have the same density. In fact, some input lists are exponentially sparse, i.e. $|S_i|=2^{-\Theta(n)}p$. 
    \item The second construction uses a more elaborate encoding of the YZ problem. All sets are balanced with $|S_i|=\lfloor p/2\rfloor$, but the code has inverse polylogarithmic rate. Efficient classical algorithms can provably only attain $1/2+o(1)$, i.e. Prange's algorithm is near optimal.
\end{enumerate}

We highlight two caveats in comparison to the quantum algorithms in \cref{sec:algo}. To enable the  two encodings described above, we require super-polynomial alphabet size $p$. The source construction of \cite{yz} likewise has a super-polynomial folded alphabet, although its underlying field has polynomial size. Then, we can no longer pass in the input lists $S_i$ via a truth table, so we assume membership oracles $F_i=\mbf{1}_{S_i}$ for $i\in [n]$ (see \cref{subsec:subsampling}).
This is in contrast to \cite{dqi} and the previous sections of our paper, where we may take $p = \poly(n)$ and avoid the need for a membership oracle in order to obtain efficient quantum algorithms.

Moreover, the general-purpose guarantees in \cref{sec:algo} do not provide
efficient perfect search for these ensembles, which have either
vanishing density or vanishing rate. Our quantum algorithms instead
exploit the encodings to simulate the Yamakawa--Zhandry source
problem, solve it using \cref{thm:advantage-yz}, and translate the
result back to a target polynomial.

\subsection{Recap of the Source Ensemble in \cite{yz}}
\label{subsec:yz-recap}
We build the ``source'' ensemble of our quantum advantage claim. This problem is related to OPI, but the underlying code is the Folded Reed-Solomon code $C_0$ rather than the Reed-Solomon code. We fix parameter choices in \cite[Section~4.3.2]{yz} and use the guarantees as a black box.

\begin{definition}[Source Ensemble]
\label{def:source-ensemble}
Let $m+1$ be a power of two, let $\F_q$ be the finite field with order $q=(m+1)^2$, and choose an alphabet $ \Sigma=\F_q^{m+2}$. The code $C_0\subset \Sigma^m$ is given by evaluating degree polynomials $Q\in\F_q[X]$ with $\deg Q\le9(q-1)/10$ at consecutive
powers of a primitive element of $\F_q$, and folding consecutive $m+2$
evaluations.
That is, if $\gamma$ is the chosen primitive element,
\[ C_0 = \left\{ \left( \begin{bmatrix} Q(1) \\ Q(\gamma) \\ \vdots \\ Q(\gamma^{m+1}) 
\end{bmatrix}, 
\begin{bmatrix} Q(\gamma^{m+2}) \\ Q(\gamma^{m+3}) \\ \vdots \\ Q(\gamma^{2m + 3})
\end{bmatrix}, \cdots,
\begin{bmatrix} Q(\gamma^{(m-1)(m+2)}) \\ Q(\gamma^{(m-1)(m+2) + 1}) \\ \vdots \\ Q(\gamma^{m(m+2) - 1}) 
\end{bmatrix}
\right) \,:\, \deg(Q) \leq 9(q-1)/10 \right\} \subset \Sigma^m.
\]
The ensemble is given by $C_0$ as well as by 
a tuple $O=(O_1,\ldots,O_m)$ of independent uniform
functions $O_j:\Sigma\to\{0,1\}$, interpreted as membership oracles for the input lists.  
\end{definition}
\noindent With this ensemble, the goal in the YZ problem is to find $x \in C_0$ so that $O_i(x_i) = 1$ for all $i$.

For convenience, let us consider the message space $\F_2^D$ of $C_0$ and the encoder 
\begin{equation}
\label{eq:advantage-encoder}
E:\F_2^D\to C_0\subset (\mb{F}_q^{m+2})^m,\qquad \text{where } 
D=\dim_{\F_2}C_0=\Theta(m^2\log m).
\end{equation}
We fix the field representations once, identifying 
binary messages in $\F_2^D$, integers $\{0,\dots, 2^D-1\}$, and coefficient vectors of polynomials $Q \in \F_q[X]$ of degree at most $9(q-1)/10$, without further notation. We note that applying $E$ to $a\in\F_2^D$ to get $E(a)\in C_0$ can be done efficiently. For every $j\in [m]$, let $E_j:\F_2^D\to \F_q^{m+2}$ be the $j$-th block of the encoding map $E$. Then, interpolation on the
$j$th block's $m+2$ evaluation points efficiently finds a message $a$
with $E_j(a)=z$.  While there may be many such $a$, we choose a canonical one by using the unique
interpolant of degree at most $m+1$. 
Denote the satisfaction ratio by
\begin{equation}
\label{eq:advantage-source-score}
r: \F_2^D\to [0, 1]\quad\text{where}\quad  r(a):=\frac1m\sum_{j=1}^m O_j(E_j(a)).
\end{equation}
If $r(a) = 1$, we refer to $a$ as a \emph{perfect source message}.  (That is, it is a perfect message for the source YZ problem).

We next state the separation guarantees of \cite{yz} on this ensemble. For technical reasons, we will extend the guarantee from the original statement in \cite{yz} with two trivial changes: first, we observe from the proofs in \cite{yz} that the $\on{negl}(m)$ error is in fact exponentially small; second, we show not only can no classical algorithm find a particularly good solution, it cannot find a particularly bad one. 

\begin{theorem}[\cite{yz}]
\label{thm:advantage-yz}
For the ensemble in \cref{def:source-ensemble} and all sufficiently large $m$:
\begin{enumerate}
\item There is a uniform polynomial-time quantum algorithm $\mc{Q}^O$ that finds a perfect source
message with probability $1-e^{-\Omega(m)}$ over $O$ and the algorithm.
\item Every randomized classical algorithm $\mc A^O$ making at most $2^{m^{2/3}}-m$ queries to $O$ satisfies
\begin{equation}
\label{eq:def-eta}
\mb{P}_{O,\mc A}\left(
r(\mc A^O)\in [0, \eta)\cup (1-\eta, 1]
\right)\le e^{-\Omega(m)}\quad\text{where}\quad \eta:=\frac{1}{80}.
\end{equation}
\end{enumerate}
\end{theorem}

\begin{proof}
The quantum algorithm is~\cite[Lemma~6.3]{yz}. For the classical claim, use~\cite[Lemma~4.2 and Section~4.3.2]{yz}
with list-recovery exponents $2/3$ and $3/4$ and agreement fraction $19/20$. We apply the proof
of \cite[Lemma~6.9]{yz} with two-sided Chernoff bound to handle both tails in \cref{eq:def-eta}.
In both statements, the failure probability bound of \cite{yz} is given as $\on{negl}(m)$, but a closer inspection of the proofs review it is actually $e^{-\Omega(m)}$.
\end{proof}

\subsection{The Sparse Target OPI Ensemble}
\label{subsec:sparse-ensemble}
In this section, we give the first target OPI ensemble of OPI instances over $\F_p$ with $\log p =\Theta(m^2\log m)$, with length $n=m$ and any dimension $k$. More specifically, let
\begin{equation}
\label{eq:def-p}
m2^D<p<2m 2^D.
\end{equation}
We fix the canonical encoding of $(j, a)\in [m]\times \F_2^D$ as a field element via
\begin{equation}
    \label{eq:advantage-embedding}
(j, a)\longleftrightarrow u(j, a):=(j-1)2^D+a \bmod{p}
\end{equation}
and for every $u\in\F_p$, we view it as an integer in $\{0, \dots, p-1\}$ and parse $(j(u), a(u))\in [m]\times \F_2^D$ via
\begin{equation}
\label{eq:advantage-parser}
u\bmod{m2^D}=(j(u)-1)2^D+a(u).
\end{equation}
Note that the encoding is injective by our choice of field size $p$, and parsing recovers $(j,a)$ on its image. This allows us to interpret each field symbol as a solution $a$ and block coordinate $j$ of the hard instance \cref{def:source-ensemble}. 

\begin{definition}[Sparse Target Ensemble]
\label{def:sparse-target}
Fix any $k\in [n]$ and any Reed-Solomon code $C\subset \mb{F}_p^n$ of dimension $k$. Sample an instance from the ensemble in \cref{def:source-ensemble}, given by oracles $O=(O_1, \dots, O_m)$ with $m=n$. 
Construct the OPI instance given by $C$ and membership oracles $F_1, \dots, F_n$ of the input lists as follows. Define functions $\varphi_O, \varphi_{\on{all}}:\F_p\to \{0, 1\}$ via
\begin{equation}
\label{eq:advantage-warmup-predicates}
\varphi_O(u)=O_{j(u)}(E_{j(u)}(a(u)))\quad\text{and}\quad 
\varphi_{\on{all}}(u)=\prod_{j=1}^m O_j(E_j(a(u))).
\end{equation}
Then, let $F_i=\varphi_{\on{all}}$ for every $i\in [n-1]$ and $F_n=\varphi_O$.
\end{definition}

Since each $O_j$ is uniformly random, so with high probability $F_n$ is roughly balanced while $F_i$ is exponentially sparse for $i\in [n-1]$. On the other hand, since $k$ is arbitrary, the code can have any rate. Let the evaluation points of $C$ be $\alpha_1, \dots, \alpha_n \in \F_p$. The intuition behind the above construction is that we have one ``easy'' constraint $F_n$ that we use to simulate oracle queries to $O_j(z_j)$, together with $n-1$ ``difficult'' constraints where satisfying any of them perfectly solves the hard Yamakawa-Zhandry instance. Here, solving it as a scalar is already hard, so we enforce no consistency on the evaluations $Q(\alpha_i)$; this means that the code can have any rate.
\begin{theorem}
\label{thm:advantage-warmup}
For sufficiently large $n$,
the ensemble in \cref{def:sparse-target} admits a uniform polynomial time quantum algorithm
that returns a perfect polynomial with probability
$1-e^{-\Omega(n)}$. 

Moreover, every classical polynomial-query algorithm returns
a polynomial satisfying more than one constraint with probability at most
$e^{-\Omega(n)}$.
\end{theorem}
\begin{proof}
For the quantum algorithm, we solve the underlying instance from \cref{def:source-ensemble} via the quantum algorithm in \cref{thm:advantage-yz}. To simulate oracle query $O_j(z_j)$ for $j\in [m]$ and some $z_j\in \Sigma$, we efficiently find via \cref{eq:advantage-encoder} some $a\in \F_2^D$ with $E_j(a)=z_j$, compute
$u=u(j, a)\in \F_p$ via \cref{eq:advantage-embedding}, and return the query outcome
\[ F_n(u) = O_{j(u)}(E_{j(u)}(a(u))) = O_j(z_j).\]
With probability $1-e^{-\Omega(n)}$, we find some perfect solution $z\in C_0$ to the corresponding instance in \cref{def:source-ensemble}, i.e. $O_j(z_j)=1$ for every $j$. Then, 
we find via \cref{eq:advantage-encoder} some $a\in \F_2^D$ with $E(a)=z$, and output the constant polynomial $Q(X)=u$ where $u=u(1, a)\in \F_p$  via \cref{eq:advantage-embedding}. We see that this polynomial satisfies every constraint $F_i$ by definition \cref{eq:advantage-warmup-predicates}.

Conversely, we efficiently simulate each classical target query $F_i(u)$ as follows: compute $j(u), a(u)$ via \cref{eq:advantage-parser}, and query $O_j(E_j(a(u)))$ for every $j\in [m]$ via \cref{eq:advantage-encoder}. 
Then, $F_n(u)$ is the $j=j(u)$ outcome and $F_i(u)$ is the product of the $m$ outcomes for every $i\in [n-1]$.
Hence, any efficient classical algorithm on the target ensemble of \cref{def:sparse-target} can be turned into an efficient classical algorithm on the source ensemble of \cref{def:source-ensemble}. If on some instances of the target ensemble the output polynomial $Q$ satisfies constraint $F_i(Q(\alpha_i))=1$ for some $i\in [n-1]$, find such an $i$ by querying these $n-1$ constraints, and compute $a=a(Q(\alpha_i))\in\F_2^D$ via \cref{eq:advantage-parser}. Then, $r(a)=1$ by \cref{eq:advantage-warmup-predicates}, i.e. $E(a)$ is a perfect solution for the corresponding source instance. Now, the classical hardness in \cref{thm:advantage-warmup} follows that in \cref{thm:advantage-yz}.
\end{proof}

\subsection{The Balanced Target OPI Ensemble}
\label{subsec:balanced-ensemble}
In this section, we consider an alternative encoding of the same source ensemble in \cref{def:source-ensemble} to obtain a target OPI ensemble with balanced input lists and prove \cref{thm:main-opi-advantage}, which we restate below.
\mainadvantage*

The intuition and comparison to \cref{def:sparse-target} is as follows: to ensure the input lists are balanced, we cannot encode ``perfectly solving the source instance in \cref{def:source-ensemble}'' into one constraint $F_i(Q(\alpha_i))=1$. Instead, we distribute this test into many constraints $F_i$, where each $F_i$ counts the parity of failures among a $t$-element subset of the $m$-constraints in the source instance. 

The main technical difficulty is ``consistency'', meaning that a polynomial $Q$ satisfying many local tests $F_i(Q(\alpha_i))=1$ must be the constant polynomial encoding a solution to the source instance with very high (or very low) satisfaction ratio. To arrange this, we construct an ensemble with randomness both from the source ensemble in \cref{def:source-ensemble} but also independently its own randomness in setup. With quasi-polynomial large field size, we can union bound over the polynomials $Q$.

To define the construction, recall the source ensemble in \cref{def:source-ensemble} parametrized by $m$. We fix the remaining parameters now, all asymptotically as $m\to\infty$. 
Let $\eps=1/\log_2m$ and let $\delta=\eps/100$. Recall that $\eta=1/80$ from \cref{eq:def-eta}. Choose the least even $t\ge 4$
such that
\begin{equation}
\label{eq:def-t}
(1-2\eta)^t\le\frac{\eps}{8},
\end{equation}
so $t=\Theta(\log\log m)$. We now specify the code parameters. The field size is some prime $p$ where
\begin{equation}
\label{eq:def-p-new}
\frac{m2^{D+2}}{\delta}<p<\frac{m2^{D+3}}{\delta}\implies \log p\asymp  m^2\log m,
\end{equation}
the Reed-Solomon code $C$ has length $n=N+2$ and dimension $k$ where
\begin{equation}
\label{eq:def-param}
\begin{aligned}
H&:=\left\lceil \frac{2^{\sqrt m}}{\binom mt}\right\rceil =  2^{\Theta(\sqrt{m})},\\
N&:=\binom mt H\sim 2^{\sqrt{m}},\\ 
k&:=\left\lfloor \frac{\delta^2N}{4\log_2p}\right\rfloor.
\end{aligned}
\end{equation}
While the prime $p$ here replaces the choice \cref{eq:def-p} in the previous section,
we keep the same encoding and parser from
\cref{eq:advantage-embedding,eq:advantage-parser}. We now construct our target OPI ensemble.

\begin{definition}[Balanced Target Ensemble]
\label{def:balanced-target}
Fix the parameters in
\cref{eq:def-p-new,eq:def-param}. We define an ensemble of OPI instances on code $C\subset \F_p^n$ of dimension $k$ as follows. Index the $n$ coordinates by
\begin{equation}
I:=\{O, B\}\cup \left( \binom{[m]}t \times [H]\right).
\end{equation}
We describe \emph{setup randomness} $S$: independently of the source oracle $O$, and independent of each other, sample the following:
\begin{itemize}
    \item Sample evaluation points $\boldsymbol{\alpha}:=\{\alpha_i:i\in I\}$ without replacement from $\F_p$, and let $C:=\on{RS}_p(\boldsymbol{\alpha};k)$.
    \item Sample random bits
$B(u)\sim\on{Ber}(\delta)$ for each $u\in \mc{U}:=\{0, \dots, (p-3)/2\}$. If $B(u) = 1$, we refer to $u$ as \emph{active}
\item Sample random bits
$C_{J,h}(u)\sim \on{Ber}(1/2)$ for each $u\in \mc{U}$ and $(J, h)\in \binom{[m]}{t}\times [H]$.
\end{itemize}

For each setup $S=(\boldsymbol{\alpha}, B, C)$, define $\varphi_i :\mc{U}\to\{0, 1\}$ for $i\in I$ as follows: let \begin{equation}
\label{eq:advantage-OB-predicates}
\varphi_O(u):=O_{j(u)}(E_{j(u)}(a(u))) \quad \text{and}\quad \varphi_B(u):=B(u),
\end{equation}
and for each $(J, h)\in\binom{[m]}{t}\times [H]$ let
\begin{equation}
\label{eq:advantage-label-predicates}
\varphi_{J,h}(u):=
\begin{cases}
\mbf{1}\{ \#\{j\in J: O_{j}(E_{j}(a(u)))=0\} \text{ is even}\},
&\text{if }B(u)=1,\\
C_{J,h}(u), &\text{if }B(u)=0.
\end{cases}
\end{equation}
For each $i\in I$, we define the input lists given by membership oracles $F_i:\F_p\to \{0, 1\}$ by
\begin{equation}
\label{eq:advantage-balanced-predicates}
F_i(2u+b):=\varphi_i(u)\oplus b,
\end{equation}
for $u\in \mc{U}$ and $b\in\{0, 1\}$, and let
$F_i(p-1)=0$.
\end{definition}
We would like to reduce the target ensemble defined above to and from instances of the source ensemble in \cref{def:source-ensemble}, with the same oracle realization $O_1, \dots, O_m$. To this end, we prove the following proposition, which says that
an active constant $u$ that encodes a perfect source message corresponds to a constant polynomial with perfect satisfaction in the target ensemble; and on the other hand any polynomial with perfect (or even good) satisfaction in the target ensemble must correspond to a $u$ whose encoding has extreme satisfaction in the source ensemble.

\begin{proposition}[Consistency]
\label{prop:consistency}
With probability $1-e^{-\Omega(m)}$ over the setup $S=(\boldsymbol{\alpha}, B, C)$, the following
holds simultaneously for every source oracle $O$ and every
polynomial $Q\in\F_p[X]$ with $\deg Q<k$: if $s(Q)>1/2+\eps$, then $Q(X)=2u$ is constant for some $u\in\mc{U}$ such that $B(u)=1$ and
\begin{equation}
\label{eq:advantage-consistency-conclusion}
r(a(u))\in [0, \eta)\cup (1-\eta, 1].
\end{equation}
Moreover, if $Q(X)=2u$ with $B(u)=1$ and $r(a(u))=1$, then $s(Q)=1$.
\end{proposition}

\begin{proof}
If $Q(X)=2u$ with $B(u)=1$ and $r(a(u))=1$, then $\varphi_O(u)=1$ as $O_j(E_j(a(u)))=1$ for every $j$, and so $\varphi_{J, h}(u)=1$ for every $(J, h)$. Moreover, $\varphi_B(u)=B(u)=1$. Hence, $F_i(2u)=\varphi_i(u)=1$ for every $i\in I$, so $s(Q)=1$. For the rest of the proof, we show that with high probability over setup $S=(\boldsymbol{\alpha}, B, C)$, every $Q$ with $s(Q)>1/2+\eps$ has the form asserted in \cref{eq:advantage-consistency-conclusion}, simultaneously for every realization $O$ in the source ensemble. We will then union bound over $Q$, agnostic over $O$.

Suppose that $Q$ is nonconstant, for each ``main coordinate'' $i\in \binom{[m]}{t}\times[H]$, write $Q(\alpha_i)$ uniquely as $Q(\alpha_i)=2u_i+b_i$ for some $u_i\in\mc{U}$ and $b_i\in\{0, 1\}$. First, we show that it is unlikely for there to be at most $(1-\delta)N$ many distinct $u_i$'s. That is, it is unlikely that $|\{ u_i: i \in [N] \}| \leq (1 - \delta)N$.
Fix any set $U\subset \mc{U}$ of size $(1-\delta)N$ and any nonconstant polynomial $Q$. Then, for every $i\in \binom{[m]}{t}\times[H]$, $Q(\alpha_i)$ must be one of at most $2|U|+1\le 3N$ fixed values in $\F_p$ since each $u\in U$ corresponds $2u$ and $2u+1$, with possibly additionally $p-1$. Now, for random evaluations $\boldsymbol{\alpha}$, by counting roots, 
\[
\mb{P}_S\left(u_i\in U\,\forall\, i\in \binom{[m]}{t}\times[H]\right)\le \frac{(2Nk)^N}{p(p-1)\dots (p-N+1)}.
\]
Now, we union bound over $Q, U$ and apply our parameter choices \cref{eq:def-p-new,eq:def-param} to obtain
\begin{align*}
\mb{P}_S\left(\exists \text{ nonconstant $Q$ with }\le (1-\delta)N\text{ distinct }u_i\right) & \le \frac{p^{k+(1-\delta)N}(2Nk)^N}{p(p-1)\dots (p-N+1)}
\\ & \le
\exp\left((k-\delta N)\log p+N\log(2Nk)+O(N^2/p)\right)
\\ & =e^{-\Omega(\delta N\log p)}
\\ & \le e^{-\Omega (m)}.
\end{align*}

Second, assume that some nonconstant $Q$ has more than $(1-\delta)N$ distinct $u_i$'s; collect these values of $u_i$ into $U\subset \mc{U}$.  We show that it is unlikely for there to be more than $3\delta N$ many main coordinates $i$ where $B(u_i)=1$. Indeed, if the $N$ choices of $u_i$ were all distinct, we would have $\on{Bin}(N, \delta)$ such $i$. Since there are at most $N-|U|\le \delta N$ duplicate labels $u_i$ by the first step, the number of main coordinates $i$ where $B(u_i)=1$ is at most $\delta N+\on{Bin}(N, \delta)$. Temporarily letting $Z_2\sim \on{Bin}(N, \delta)$, we see that 
\begin{align*}
\mb{P}_S\left(\exists \text{ nonconstant $Q$ with $\ge 3\delta N$ many $i\in \binom{[m]}{t}\times[H]$ where $B(u_i)=1$}\right) 
& \le e^{-\Omega(m)}+p^k\mb{P}(Z_2\ge 2\delta N)
\\ & \le e^{-\Omega(m)} +p^ke^{-2\delta^2 N}
\\ & \le e^{-\Omega(m)} +e^{-\Omega(\delta^2 N)}
\\ & \le e^{-\Omega(m)} 
\end{align*}

Third, suppose that some nonconstant $Q$ satisfies the first two steps, so $Q$ has at least $(1 - \delta)N$ distinct values of $u_i$, and that there are at most $3 \delta N$ main active coordinates $i$.  
Let
\[ I':=\left\{i\in \binom{[m]}{t}\times[H]:B(u_i)=1\right\}\]
be this set of main active coordinates.

For every main coordinate $i\notin I'$, $B(u_i)=0$ so i.i.d. random bits $C_i(u_i)$ are used. Then, the number of such constraints satisfied by each $Q$ is bounded by $\on{Bin}(N-|I'|, 1/2)$, conditional on $\boldsymbol{\alpha}$ and $B$. For the constraints $i\in I'$, we assume they automatically hold, which allows us to bound the satisfaction ratio simultaneously for every $O$ as follows:
for such $Q$, $n\cdot s(Q)\le |I'|+2+Z_3$ where $Z_3\sim \on{Bin}(N, 1/2)$ temporarily, so
$s(Q)\ge\frac12+\eps$ implies 
\[Z_3
    \ge
    \left(\frac12+\eps\right)(N+2)-3\delta N-2
    \ge
    \frac N2+\delta N,\]
since $\delta=\eps/100$. Therefore, we see that
\begin{align*}
\mb{P}_S\left(\exists\, O, \exists \text{ nonconstant $Q$ with }s(Q)\ge \frac{1}{2}+\eps\right) 
& \le e^{-\Omega(m)}+p^k\mb{P}(Z_3\ge N/2+\delta N)
\\ & \le e^{-\Omega(m)} +p^ke^{-2\delta^2 N}
\\ & \le e^{-\Omega(m)} +e^{-\Omega(\delta^2 N)}
\\ & \le e^{-\Omega(m)} 
\end{align*}

Fourth, suppose $Q$ is constant, so we assume $Q(X)=2u+b$ for some $u\in \mc{U}$ and $b\in\{0, 1\}$, as $F_i(p-1)=0$ for every $i$. Now, if $B(u)=0$, then i.i.d. random bits $C_i(u)$ are used at every main coordinate $i$. Exactly as the third case, the probability of existing such $Q$ with $s(Q)\ge 1/2+\eps$ is $e^{-\Omega(m)}$. 

Finally, it suffices to consider $Q(X)=2u+b$ where $B(u)=1$. Let $a=a(u)$, fix any realization of the source membership oracles $O=(O_1, \dots, O_m)$ from \cref{def:source-ensemble}, and let $r=r(a)$ be the fraction of constraints of the source instance $z=E(a)$ satisfies. We show that $s(Q)>1/2+\eps$ forces $b=0$ and $r\notin[\eta,1-\eta]$. Fix any $b$ and $r$. Recall that for every $(J, h)\in \binom{[m]}{t}\times [H]$, by some parity arithmetic
\[
F_{J,h}(Q(\alpha_{J,h})) =F_{J,h}(2u+b) 
= \varphi_{J,h}(u)\oplus b
=
\mbf{1}\left\{ \#\{j\in J: O_{j}(E_{j}(a))=0\} \text{ is even}\right\}\oplus b.
\]
For active $u$, the value $F_{J,h}(2u+b)$ is independent of $h$. Every $t$-subset $J$ occurs exactly $H$ times among the main constraints, so averaging their satisfaction is the same as averaging over a uniform $J$. Now, we claim that
\begin{equation}
\begin{aligned}
\label{eq:consistency-constant}
s(2u)= \mb{P}_{J\in\binom{[m]}{t}}\left(\#\{j\in J: O_{j}(E_{j}(a))=0\} \text{ is even}\right) \pm O\left(\frac{1}{n}\right)=\frac{1+(2r-1)^t}{2} \pm O\left(\frac{1}{n}+\frac{t^2}{m}\right). \\
s(2u+1)= \mb{P}_{J\in\binom{[m]}{t}}\left(\#\{j\in J: O_{j}(E_{j}(a))=0\} \text{ is odd}\right) \pm O\left(\frac{1}{n}\right)=\frac{1-(2r-1)^t}{2} \pm O\left(\frac{1}{n}+\frac{t^2}{m}\right). \\
\end{aligned}
\end{equation}
The $O(1/n)$ term corresponds to constraints $F_O$ and $F_B$. We justify this inequality as follows: consider sampling $J$ instead as $t$ many i.i.d. uniform random indices $j_1, \dots, j_t\in[m]$, in which case the desired parity event---say the even case---is exactly the event that a $\on{Bin}(t, 1-r)$ random variable is even, which has probability $1/2+(2r-1)^t/2$. Sampling with replacement in \cref{eq:consistency-constant}, the total variation is at most 
\[\mb{P}(j_1, \dots, j_t\text{ not distinct})\le \mb{P}(\exists\text{ distinct } t_1, t_2\in [t]:j_{t_1}=j_{t_2})\le \frac{t^2}{m}.\]
Using \cref{eq:consistency-constant}, we see that by our parameter choices \cref{eq:def-t,eq:def-p-new,eq:def-param}, $s(2u+1)\le 1/2+\eps$ and that 
\[ s(Q)\le \frac{1+(1-2\eta)^t}{2}+O\left(\frac{1}{n}+\frac{t^2}{m}\right)< \frac{1}{2}+\eps\]
whenever $r\in [\eta, 1-\eta]$. This depends on $O$ only through the event that $r(a(u))\in [\eta, 1-\eta]$.

Therefore, with probability $1-e^{-\Omega(m)}$ over the setup $S$, the first four steps hold, in which case we have that, for every realization $O$, the only $Q(X)$ with $s(Q)> 1/2+\eps$ are $Q(X)=2u$ such that $B(u)=1$ and the satisfaction ratio of $E(a(u))$ on the source instance given by realization $O$ is less than $\eta$ or greater than $1-\eta$.
\end{proof}

Now, we prove \cref{thm:main-opi-advantage}, which follows from the following theorem. 

\begin{theorem}
\label{thm:balanced-opi-advantage}
Consider the ensemble in \cref{def:balanced-target}, with parameters
chosen in
\cref{eq:def-t,eq:def-p-new,eq:def-param}.  Every such OPI instance is balanced, with rate
\[
R=\frac{k}{n}
=\Theta\left(
\frac{1}{(\log n)^4(\log\log n)^3}
\right)\quad\text{and}\quad
\log p=\Theta\left((\log n)^4\log\log n\right)
\]
as $m\to\infty$.
Under membership-oracle access to the input lists, the following hold:
\begin{enumerate}
\item There is a uniform quantum algorithm running in
polynomial time that outputs a polynomial $Q$ with
$\deg Q<k$ and $s(Q)=1$ with probability at least
\[
1-e^{-\Omega(m)}
=1-e^{-\Omega((\log n)^2)}.
\]

\item Every randomized classical algorithm
with polynomially many membership queries and unrestricted computation between
queries, and outputting
a polynomial $Q$ with $\deg Q<k$ satisfies
\[
\mb{P}\left(s(Q)>\frac12+\eps\right)
\le e^{-\Omega(m)}
=e^{-\Omega((\log n)^2)}\quad\text{where}\quad \eps
=\Theta\left(\frac{1}{\log\log n}\right).
\]
\end{enumerate}
\end{theorem}
As a sanity check, observe that this is an unconditional relative separation, and that Prange's algorithm is close to optimal, i.e. it has expected satisfaction 
\[
\frac{1+R}{2}=\frac12+\frac{1}{\log ^{O(1)}n}\le \frac{1}{2}+\eps.
\]
\begin{proof}
The instances are balanced by definition. We first check the parameters. Since $t=\Theta(\log\log m)$,
\[
\binom mt\le m^t=2^{O(\log m\log\log m)}=2^{o(\sqrt m)},
\]
and so $n=N+2\sim2^{\sqrt m}$. Hence, failure probabilities are $e^{-\Omega(m)}
=e^{-\Omega((\log n)^2)}$, slack $\eps =\Theta(1/\log\log n)$, and alphabet size $\log p=\Theta(m^2\log m)$. Finally,
\begin{equation}
\label{eq:advantage-final-rate}
\frac{k}{n}
=\Theta\left(\frac{\delta^2}{\log p}\right)
=\Theta\left(\frac1{m^2(\log m)^3}\right)
=\Theta\left(\frac1{(\log n)^4(\log\log n)^3}\right).
\end{equation}

For (1), we construct the quantum algorithm by solving the underlying instance from \cref{def:source-ensemble} via the quantum algorithm in \cref{thm:advantage-yz}. To simulate oracle query $O_j(z_j)$ for $j\in [m]$ and some $z_j\in \Sigma$, we efficiently find via \cref{eq:advantage-encoder} some $a\in \F_2^D$ with $E_j(a)=z_j$, compute
$u=u(j, a)\in \F_p$ via \cref{eq:advantage-embedding}, and return the query outcome
\[ F_O(2u) = O_{j(u)}(E_{j(u)}(a(u))) = O_j(z_j).\]

With probability $1-e^{-\Omega(m)}$, we find some perfects solution $z\in C_0$ to the corresponding instance in \cref{def:source-ensemble}, i.e. $O_j(z_j)=1$ for every $j$. Then, 
we find via \cref{eq:advantage-encoder} some $a\in \F_2^D$ with $E(a)=z$. Now, we test $u(j, a)=a+(j-1)2^D$ for $j=1, 2, \dots , \lceil m/\delta\rceil$ until we find $j$ where $B(u(j, a))=1$. We test this by computing $u(j, a)$ via \cref{eq:advantage-embedding} and querying
\[F_B(2u)=\varphi_B(u)=B(u).\]

Since $p\gg (m/\delta)2^D$, the $u(j, a)$ are distinct and they all belong to $\mc{U}$, so the probability that we cannot find some $j$ where $B(u(j, a))=1$ is at most
\[(1-\delta)^{\lceil m/\delta\rceil}\le e^{-m}\]
If we find some $j$ where $B(u(j, a))=1$, output the constant polynomial $Q(X)=2u(j, a)$. We see that this polynomial satisfies every constraint $F_i$ by \cref{prop:consistency}. 
Finally, the number of queries and all arithmetic are polynomial in $m$ and
$\log p$, hence polynomial in $n$.

For (2), fix any classical algorithm $\mc{A}$ for the target ensemble, making polynomially many target queries, say at most $n^d$ for a fixed $d$. We build a classical algorithm to solve the underlying source ensemble of \cref{def:source-ensemble}. To do so, we first sample a random setup $S=(\boldsymbol{\alpha}, B, C)$ as in \cref{def:balanced-target}. We sample gate and filler bits only when first needed and store them for repeated queries; all setup randomness is independent of $O$. With probability $1-e^{-\Omega(m)}$, the setup $S$ is \emph{consistent}, i.e. \cref{prop:consistency} holds. 

Now, fix any source instance given by any realization of oracles $O=(O_1, \dots, O_m)$, and solve it by running $\mc{A}$ on the target instance given by $O, S$ and simulating each classical target queries $F_i$ via $O$ as follows. We return $F_i(p-1)=0$ for every $i$. For a query $F_i(2u+b)$ with
$u\in\mc U$ and $b\in\{0,1\}$, compute $j(u),a(u)$ via
\cref{eq:advantage-parser} and proceed as follows:
\begin{itemize}
    \item If $i=O$, compute $E_{j(u)}(a(u))$ via
    \cref{eq:advantage-encoder}, make one source query to return
    \[
    O_{j(u)}(E_{j(u)}(a(u)))\oplus b = \varphi_O(u)\oplus b = F_O(2u+b).
    \]

    \item If $i=B$, return $B(u)\oplus b$. This requires only setup $S$ and no source query.

    \item If $i=(J,h)$, first inspect $B(u)$.
    If $B(u)=0$, return $C_{J,h}(u)\oplus b$.
    Otherwise, query $O_j(E_j(a(u)))$ for every $j\in J$ and compute the parity to return $F_i(2u+b)$ via \cref{eq:advantage-label-predicates}.
    This requires at most $t$ source queries.
\end{itemize}
Now, suppose $\mc{A}$ on this instance $(O, S)$ outputs polynomial $Q$, if $Q(X)=2u$
for some $u\in\mc U$, output $a(u)$; otherwise, output any fixed message
in $\F_2^D$. 

We analyze this algorithm on the source ensemble from \cref{def:source-ensemble} parametrized by $O$. 
This simulation has exactly the distribution of the target ensemble and uses at most
\[
tn^d
=2^{O_d(\sqrt m)}
<2^{m^{2/3}}-m
\]
source queries for sufficiently large $m$. 
By \cref{thm:advantage-yz}, the satisfaction ratio of the output must be in $[\eta, 1-\eta]$ with probability $1-e^{-\Omega(m)}$ over the source ensemble. 
By \cref{prop:consistency}, except on a setup event of probability
$e^{-\Omega(m)}$, consistency implies that every output with $s(Q)>1/2+\eps$ satisfies $Q(X)=2u$ for some $B(u)=1$ and
$r(a(u))\in [0, \eta)\cup (1-\eta, 1]$.
Together, the output $Q$ of $\mc{A}$ has satisfaction ratio $s(Q)> 1/2+\eps$ with probability at most $e^{-\Omega(m)}$. This holds for every algorithm $\mc{A}$ making polynomially many queries, proving (2).
\end{proof}

We remark that the quantum algorithm we apply is not the generic DQI algorithm in \cite{dqi} or the related algorithms in~\cite{chailloux1,chailloux2,hy,sjo}, nor the algorithm in \cref{sec:algo}. Indeed, generically applying these quantum algorithms give $1/2+o(1)$ satisfaction ratio guarantee since the rate of the code is vanishing. 
Instead, we apply the algorithm from \cite{yz} to the Yamakawa-Zhandry source ensemble (\cref{def:source-ensemble}) that we transform the OPI ensemble into, and then translate back.  We remark that one could also use the original DQI algorithm of \cite{dqi} (adapted for a folded version) on the transformed YZ ensemble.

\begin{remark}[Alternative Parametrization]
\label{rem:alternative-param}
By tuning the parameter choices, we obtain an analogous version of \cref{thm:balanced-opi-advantage} as follows: Let $\eps>0$ be any absolute constant.  Then we may choose any sufficiently large fixed even $t$ satisfying \cref{eq:def-t}, rather than the least such $t$. If we take $H=1$ and $N=\binom{m}{t}$, then $C$ has length $n\asymp m^t$ and rate
\[
R=\Theta_{\eps,t}\left(\frac1{m^2\log m}\right)
=\Theta_{\eps,t}\left(\frac{n^{-2/t}}{\log n}\right).
\]
Since $n$ is now polynomial in $m$ (rather than $2^{\Theta(\sqrt m)}$), hardness of polynomial-query classical algorithms (rather than $2^{O(\sqrt m)}$ queries) for the source ensemble \cref{def:source-ensemble} suffices for this target ensemble. In this version, for any constant $\eps>0$ and $c>0$, choosing also $t>2/c$ gives such a target ensemble with field size $p=\exp(n^{\Theta(1)})$, rate $R\gg n^{-c}$, polynomial-query classical algorithms have ceiling $1/2+\eps$, and error probability $e^{-\Omega(m)}=e^{-\Omega(n^{1/t})}$.
\end{remark}

\bibliographystyle{alpha}
\bibliography{ref}

\appendix
\section{Arbitrary Bias One-bit LLR Below Rate 1/2}
\label{sec:app-llr}

In this appendix, we complete the proof of \cref{thm:main-llr}(2) by removing the balance assumption from \cref{subsec:llr-below}. The idea is that for very biased leakages, $g_i$ may have large
supremum norm when $\rho_i$ approaches zero or one, and the $L^2(\nu)$ proxy recalled below need not be small in that regime. 

\begin{equation}
\FL(v,v')=\frac{1}{4}\E_{Z\sim\nu}
\left|
\frac{\P(L_v=Z)-\P(L_{v'}=Z)}{\nu(Z)}
\right|^2\quad\text{where}\quad \nu := \bigotimes_{i=1}^n\on{Ber}(\rho_i).
\end{equation}
This is illustrated in the following example.
\begin{example}
Take $\ell_i(x)=\one\{x=0\}$, $v=0$,
and $v'\notin C$. The all-one transcript has probabilities
$p^{-k}$ and zero under the two translates, but probability
$p^{-n}$ under $\nu$, giving
\[
\FL(0,v')\ge\frac14p^{n-2k}.
\]
which blows up when rate $R=k/n<1/2$.
\end{example}

To handle this, we instead apply the $L^1(\nu)$ triangle
inequality directly. While it is weaker than the $L^2(\nu)$ proxy at balance, it retains the density factor that helps when $g_i$ leakages are biased. 

\begin{lemma}
\label{lem:centered-l1-llr} 
Define $c_i:=2\sqrt{\rho_i(1-\rho_i)}$  for each $i\in [n]$. For every
$v,v'\in\F_p^n$,
\begin{equation}
\label{eq:centered-l1-llr}
2\TV(L_v,L_{v'})
\le
\sum_{\emptyset\ne A\subset[n]}
\left(\prod_{i\in A}c_i\right)
\abs{
\sum_{\substack{y\in C^\perp\\\supp(y)=A}}
\left[e_p(-\langle v,y\rangle)-e_p(-\langle v',y\rangle)\right]
\prod_{i\in A}\wh g_i(y_i)
}.
\end{equation}
\end{lemma}

\begin{proof}
By \cref{eq:llr-proxy-key} and the triangle inequality,
\begin{align*}
2\TV(L_v,L_{v'})
&=
\E_{Z\sim\nu}
\abs{\frac{\P(L_v=Z)-\P(L_{v'}=Z)}{\nu(Z)}}\\
&\le
\sum_{\emptyset\ne A\subset[n]}
\E_{Z\sim\nu}\abs{\psi_A(Z)}
\abs{
\sum_{\substack{y\in C^\perp\\\supp(y)=A}}
\left[e_p(-\langle v,y\rangle)-e_p(-\langle v',y\rangle)\right]
\prod_{i\in A}\wh g_i(y_i)
}.
\end{align*}
The empty-support term vanishes. Since $\nu$ is a product distribution,
\[
\E_{Z\sim\nu}\abs{\psi_A(Z)}
=
\prod_{i\in A}\E_{Z_i\sim\nu_i}\abs{\psi_i(Z_i)}
=
\prod_{i\in A}
\left(
\rho_i\sqrt{\frac{1-\rho_i}{\rho_i}}
+
(1-\rho_i)\sqrt{\frac{\rho_i}{1-\rho_i}}
\right)
=
\prod_{i\in A}c_i.
\]
Substituting this identity proves \cref{eq:centered-l1-llr}.
\end{proof}

For the strongly biased coordinates, we apply the $L^1(\nu)$ proxy
together with the following Brascamp--Lieb inequality.  This inequality
was proved by an entropy argument in \cite{hy}; we include a short
H\"older proof.
\begin{lemma}[Affine MDS Brascamp--Lieb]
\label{lem:bl-mds}
Let $C\subset\F_p^n$ be an $[n,k]$ MDS code of rate $R=k/n$, and
let $u\in\F_p^n$. Then, for arbitrary functions
$h_1,\ldots,h_n:\F_p\to\R_{\ge0}$,
\begin{equation}
\label{eq:bl-mds-normalized}
    \E_{x\in u+C}\prod_{i=1}^n h_i(x_i)
    \le
    \prod_{i=1}^n
    \left(\E_{a\in\F_p}h_i(a)^{1/R}\right)^R.
\end{equation}
Equivalently, writing $\|\cdot\|_{1/R}$ for the unnormalized
$\ell^{1/R}(\F_p)$ norm,
\[
    \sum_{x\in u+C}\prod_{i=1}^n h_i(x_i)
    \le
    \prod_{i=1}^n\|h_i\|_{1/R}.
\]
\end{lemma}

\begin{proof}
Every coordinate belongs to exactly $\binom{n-1}{k-1}
    =
    R\binom nk$
sets in $\binom{[n]}k$. Hence
\[
    \prod_{i=1}^n h_i(x_i)
    =
    \prod_{J\in\binom{[n]}k}
    \left(\prod_{j\in J}h_j(x_j)\right)^{
        1/\left(R\binom nk\right)}.
`\]
Applying H\"older's inequality with exponent $\binom nk$ to each of the
$\binom nk$ factors gives
\[
    \E_{x\in u+C}\prod_{i=1}^n h_i(x_i)
    \le
    \prod_{J\in\binom{[n]}k}
    \left(
      \E_{x\in u+C}
      \prod_{j\in J}h_j(x_j)^{1/R}
    \right)^{1/\binom nk}.
\]
Since $C$ is MDS, projection onto every $J\in\binom{[n]}k$ is a
bijection from $u+C$ to $\F_p^J$. Therefore
\[
    \E_{x\in u+C}
    \prod_{j\in J}h_j(x_j)^{1/R}
    =
    \prod_{j\in J}
    \E_{a\in\F_p}h_j(a)^{1/R}.
\]
Substituting this identity and using that each coordinate belongs to
an $R$ fraction of the sets in $\binom{[n]}k$, we obtain
\[
    \E_{x\in u+C}\prod_{i=1}^n h_i(x_i)
    \le
    \prod_{J\in\binom{[n]}k}
    \left(
      \prod_{j\in J}
      \E_{a\in\F_p}h_j(a)^{1/R}
    \right)^{1/\binom nk} =
    \prod_{i=1}^n
    \left(\E_{a\in\F_p}h_i(a)^{1/R}\right)^R.
\]
This proves \cref{eq:bl-mds-normalized}. Multiplying by
$|C|=p^k=p^{Rn}$ gives the equivalent unnormalized formulation.
\end{proof}

We now prove the arbitrary bias case of \cref{thm:main-llr}(3). 

\begin{lemma}
\label{lem:arbitrary-llr-below}
There exists an absolute constant $\delta>0$ such that for every generalized Reed-Solomon code $C\subset\F_p^n$ of fixed rate $R\ge1/2-\delta$ over a prime
field $\F_p$ with $p\ge n$, and every one-bit leakage tuple, uniformly over $v,v'\in\F_p^n$,
\[
\TV(L_v,L_{v'})\le e^{-\Omega(n)}.
\]
\end{lemma}

We briefly highlight the proof idea, in comparison to the proof of \cref{lem:balanced-llr-below}. 
We split by support size in comparison with  $m:=\lfloor(2-\delta)k\rfloor$ as usual. The small supports are handled by the same
$\FL^+$ and \cref{lem:wftl}  as in the proof of \cref{lem:balanced-llr-below}. For large supports,
we use the $U^3$ bound when few coordinates are strongly biased, since the
functions $g_i$ are uniformly bounded near balance and have a $U^3$ gap, so we follow the proof of \cref{lem:balanced-llr-below}; when many coordinates are strongly biased, 
$\|g_i\|_\infty$ may be large, but $c_i$ is small and
$c_i\|g_i\|_\infty\le2$, so we use \cref{lem:bl-mds} to exploit these small
factors using $\E|g_i|^2=1$, without paying the full $\Vert g_i\Vert_\infty$.

\begin{proof}
For $R>1/2$, it follows \cref{cor:one-bit-above}, so we fix
$R\in[1/2-\delta,1/2]$. Fix $q_0<1$ supplied by \cref{lem:u3-gap} with $\tau=1/4$. Complementing individual leaked bits,
we assume $\rho_i\ge1/2$. Recall that
\[
\E|g_i|^2=1,
\qquad
c_i\|g_i\|_\infty=2\rho_i\le2.
\]
We choose $\varepsilon>0$ sufficiently small and then
$\delta>0$ sufficiently small depending on $\varepsilon$.
Define support cutoff $m:=\lfloor(2-\delta)k\rfloor$.
By \cref{lem:centered-l1-llr} and Fourier inversion,
\[
\TV(L_v,L_{v'})
\le
\frac12\sum_{u\in\{v,v'\}}
\sum_{\emptyset\ne A\subset[n]}
\left(\prod_{i\in A}c_i\right)
\abs{\E_{x\in u+C}\prod_{i\in A}g_i(x_i)}.
\]
Fix a translate $u$. 
We bound the sum over $A$ in three cases, depending on
\[
t:=\abs{\left\{i\in A:\rho_i>\frac{1+\varepsilon}{2}\right\}}.
\]

\begin{itemize}
\item Suppose $0<|A|\le m$. Forgetting phases and using $c_i\le1$, we proceed exactly as the proof of
\cref{lem:balanced-llr-below}: 
we apply \cref{lem:wftl} 
with $t=m$, $\gamma=2/\pi+O(p^{-2})$ via
\cref{cor:L-infty-bound}, and $\Gamma=1$ via Parseval.
These verifications hold uniformly over the biases. Since
$2k-m\sim\delta k\asymp n$,
\begin{equation}
\sum_{0<|A|\le m}
\left(\prod_{i\in A}c_i\right)
\abs{\E_{x\in u+C}\prod_{i\in A}g_i(x_i)}\le
\sum_{\substack{0\ne y\in C^\perp\\|y|\le m}}
\prod_{i\in\supp(y)}\abs{\wh g_i(y_i)}
\le
n^{O(1)}\left(\frac2\pi\right)^k
=e^{-\Omega(n)}.
\end{equation}

\item Suppose $|A|>m$ and $t\le\eps n$, let
\[
d:=\left\lfloor\frac{3k-n}{2}\right\rfloor
\ge\left(\frac14-\frac32\delta\right)n-O(1).
\]
For sufficiently small $\varepsilon,\delta$, there are at least
$d$ near-balanced coordinates in $A$, since
$|A|-t>m-\varepsilon n>d$. Choose $J$ of size $d$ among them.
For $j\in J$, we have $2\rho_j-1\in[0,\varepsilon]$.
Since $c_jg_j=2\one_{S_j}-2\rho_j$, \cref{lem:u3-gap}
and the triangle inequality give
\[
c_j\|g_j\|_{U^3}
\le
\|2\one_{S_j}-1\|_{U^3}+2\rho_j-1
\le q_0+\varepsilon.
\]
Moreover, $c_i\|g_i\|_\infty=2\rho_i$ is at most
$1+\varepsilon$ on near-balanced coordinates and at most two
everywhere.
Apply \cref{lem:block-u3} to $c_i g_i$ on $A$ and the constant function
one outside $A$, we obtain
\begin{align*}
\left(\prod_{i\in A}c_i\right)
\abs{\E_{x\in u+C}\prod_{i\in A}g_i(x_i)}\le
2^t(1+\varepsilon)^n(q_0+\varepsilon)^d\le
2^{\left[
\varepsilon+\log_2(1+\varepsilon)
+\left(\frac14-\frac32\delta\right)
\log_2(q_0+\varepsilon)+o(1)
\right]n}.
\end{align*}

\item Suppose $|A|>m$ and $t>\eps n$, then
each of these $t$ coordinates satisfies
$c_i^2=4\rho_i(1-\rho_i)\le1-\varepsilon^2$.
Apply \cref{lem:bl-mds} to $u+C$, taking $h_i=|g_i|$ for
$i\in A$ and $h_i=\one$ otherwise.  We obtain
\[
\left(\prod_{i\in A}c_i\right)
\abs{\E_{x\in u+C}\prod_{i\in A}g_i(x_i)}
\le
\prod_{i\in A}
c_i\left(\E|g_i|^{1/R}\right)^R.
\]
Since $R\le1/2$ and $\E|g_i|^2=1$,
$\E|g_i|^{1/R}
\le
\|g_i\|_\infty^{1/R-2}\E|g_i|^2
=
\|g_i\|_\infty^{1/R-2}$.
Hence, as $c_i\|g_i\|_\infty\le2$,
\[
c_i\left(\E|g_i|^{1/R}\right)^R
\le
c_i\|g_i\|_\infty^{1-2R}
\le
2^{1-2R}c_i^{2R}.
\]
Multiplying over $A$ and recalling that $c_i^2\le1-\eps^2$ for $t$-many coordinates $i$, we have
\[
\left(\prod_{i\in A}c_i\right)
\abs{\E_{x\in u+C}\prod_{i\in A}g_i(x_i)}
\le
2^{(1-2R)|A|}\prod_{i\in A}c_i^{2R}\le
2^{(1-2R)n}(1-\varepsilon^2)^{Rt}\le
2^{\left[
2\delta+
\left(\frac12-\delta\right)\varepsilon
\log_2(1-\varepsilon^2)
\right]n}.
\]
\end{itemize}

To sum the last two cases, \cref{eq:high-support-count}
bounds the number of supports with $|A|>m$ by
$2^{\left[
H_2\left(\frac52\delta-\delta^2\right)+o(1)
\right]n}$.
Now, we choose $\delta,\eps>0$ depending on constant $q_0 <1$. First choose $\varepsilon>0$ sufficiently small that
\[
\varepsilon+\log_2(1+\varepsilon)
+\frac14\log_2(q_0+\varepsilon)<0.
\]
Then choose $\delta>0$ sufficiently small that
\begin{align*}
H_2\left(\frac52\delta-\delta^2\right)
+\varepsilon+\log_2(1+\varepsilon)
+\left(\frac14-\frac32\delta\right)
\log_2(q_0+\varepsilon)&<0,\\
H_2\left(\frac52\delta-\delta^2\right)
+2\delta+
\left(\frac12-\delta\right)\varepsilon
\log_2(1-\varepsilon^2)&<0.
\end{align*}
Both expressions are strictly negative at $\delta=0$.
Thus all three cases contribute $e^{-\Omega(n)}$, uniformly
over the leakage tuple and the translates, proving the claim.
\end{proof}

\section{Near-balanced OPI Discrepancy Below Rate 1/2}
\label{sec:app-opi}

We complete \cref{thm:main-opi}(2). Unlike the preceding appendix on the arbitrary-bias case of \cref{thm:main-llr}(2),
we only need to perturb the balanced OPI argument by a small
constant in the density. Then,
in the same support split by $m=\lfloor(2-\delta)k\rfloor$ as in
\cref{eq:puncture-size}, the low-support estimate remains valid. For the
high supports, we need to be a bit more careful.
\begin{lemma}
\label{lem:near-balanced-opi-below}
 There exists an absolute constant $\delta>0$ such that for every
generalized Reed-Solomon code $C\subset\F_p^n$ of fixed rate $R\ge1/2-\delta$ over a prime
field $\F_p$ with $p\ge n$, then uniformly over every OPI instance with common
density $\rho_i=\rho\ge1/2-\delta$ and $v\in\F_p^n$,
\[
\FO(v)\le e^{-\Omega(n)}.
\]

\end{lemma}
\begin{proof}
For $R>1/2$, it follows \cref{lem:positive-opi-above}, so we fix
$R\in[1/2-\delta,1/2]$. Let $m:=\lfloor(2-\delta)k\rfloor$ and
$\lambda:=\sqrt{(1-\rho)/\rho}$. As in the proof of
\cref{lem:balanced-opi-below}, we split by $m$ the support size and bound
\[
\FO(v)
\le
\sum_{\substack{0\ne y\in C^\perp\\|y|\le m}}
\lambda^{|y|}
\prod_{i\in\supp(y)}|\wh g_i(y_i)|
+
\sum_{|A|>m}
\lambda^{|A|}
\abs{\E_{x\in v+C}\prod_{i\in A}g_i(x_i)}.
\]
We bound the two terms separately.
\begin{itemize}
\item For small supports, use the same local bounds as in the
proof of \cref{lem:positive-opi-above}:
\[
\sum_{\xi\ne0}\left(\lambda|\wh g_i(\xi)|\right)^2
=\frac{1-\rho}{\rho},
\qquad
\lambda|\wh g_i(\xi)|
\le\frac{\sin(\pi\rho)}{\pi\rho}+O_\rho(p^{-2}).
\]
On the fixed interval $|\rho-1/2|\le\delta$, the first quantity
is bounded and the second has a uniform limiting bound below one.
Since $2k-m\sim\delta k\asymp n$, \cref{lem:wftl} gives
\[
\sum_{\substack{0\ne y\in C^\perp\\|y|\le m}}
\lambda^{|y|}
\prod_{i\in\supp(y)}|\wh g_i(y_i)|
\le
n^{O(1)}
\left(\frac{\sin(\pi\rho)}{\pi\rho}\right)^k
=e^{-\Omega(n)}.
\]

\item For large supports, fix $A$ with $|A|>m$.
The point is that the weighted functions themselves remain
close to balanced sign functions:
\[
\lambda g_i
=
\frac{\one_{S_i}-\rho}{\rho}
=
\frac{2\one_{S_i}-1-(2\rho-1)}{2\rho}
\implies
\lambda\|g_i\|_\infty
=
\max\left\{1,\frac{1-\rho}{\rho}\right\}
\le
\frac{1+2\delta}{1-2\delta}.
\]
Fix $q_0<1$ from \cref{lem:u3-gap} with $\tau=1/4$.
Since $|2\rho-1|\le2\delta$, the same lemma and the triangle
inequality give
\[
\lambda\|g_i\|_{U^3}
\le
\frac{\|2\one_{S_i}-1\|_{U^3}+|2\rho-1|}{2\rho}
\le
\frac{q_0+2\delta}{1-2\delta}<1
\]
for sufficiently small $\delta$.
Let
\[
d:=\left\lfloor\frac{3k-n}{2}\right\rfloor
\ge
\left(\frac14-\frac32\delta\right)n-O(1).
\]
Apply \cref{lem:block-u3}, using $\lambda g_i$ on $A$, the constant function
one outside $A$, and any $d$ coordinates from $A$ for the
$U^3$ bounds. We obtain
\[
\lambda^{|A|}
\abs{\E_{x\in v+C}\prod_{i\in A}g_i(x_i)}
\le
\left(\frac{1+2\delta}{1-2\delta}\right)^n
\left(\frac{q_0+2\delta}{1-2\delta}\right)^d.
\]
Counting supports by \cref{eq:high-support-count}, their total
contribution is at most
\[
2^{\left[
H_2\left(\frac52\delta-\delta^2\right)
+\log_2\left(\frac{1+2\delta}{1-2\delta}\right)
+\left(\frac14-\frac32\delta\right)
\log_2\left(\frac{q_0+2\delta}{1-2\delta}\right)
+o(1)
\right]n}.
\]
The limiting exponent at $\delta=0$ is
$\frac14\log_2q_0<0$. Choosing $\delta>0$ sufficiently small
therefore makes this contribution $e^{-\Omega(n)}$.
\end{itemize}
Combining the two estimates proves the lemma.
\end{proof}

Finally, suppose only that $\rho_i\ge1/2-\delta$. Put
$t:=\lceil(1/2-\delta)p\rceil$, and choose each $T_i$ independently
and uniformly among the $t$-subsets of $S_i$. For sufficiently
large $p$, their common density $t/p$ belongs to
$[1/2-\delta,1/2+\delta]$. Apply \cref{prop:subsample} to $v+C$:
\[
\FO(v;C,\mbf S)
\le\E_T\FO(v;C,T_1\times\cdots\times T_n)
\le e^{-\Omega_R(n)}.
\]
This proves the discrepancy bound for lists $\mbf{S}$ and completes
\cref{thm:main-opi}(2) at $v=0$ by
\cref{lem:translated-opi-proxy}.

\section{Proof sketch of \cref{fact:hmsy-model}}
\label{app:hmsyProofSketch}
In this appendix, we sketch the proof of \cref{fact:hmsy-model}, which is implicit in \cite{HMSY26}.  For the reader's convenience, we first re-state the fact.

\hmsyfact*
\begin{proof}[Proof sketch.]
The proof of \cref{fact:hmsy-model} is implicit in the proof of Proposition 10.3 in \cite{HMSY26}.  For completeness, we sketch here how to put together the pieces in \cite{HMSY26} to obtain this result.  Below, when we refer to  ``steps'' in the proof, we are referring to the steps in the proof of Proposition 10.3 in \cite{HMSY26}.  We emphasize that this proof sketch is meant to be read alongside \cite{HMSY26}, and is not a stand-alone proof.

For notations, our code $C$ has length $n$ and dimension $k$, the Gowers norm is $U^s$, the evaluation points are $\alpha\in\F_p^n$, and the structured family is $\mc{G}$; in \cite{HMSY26}, the code has length $m$, dimension $r$, the Gowers norm is $U^\ell$, the evaluation points are $a\in\F_p^m$, and the structured class is $\mc{G}_{p, \ell}$.

In Step 1(a) and 1(b) of the proof of Proposition 10.3 in \cite{HMSY26}, they define a quantity called $q_{\on{cert}}$, and observe that $q_{\on{cert}}=O(1)$. By equation (79) of \cite{HMSY26}, the total number of revealed coordinates is at most $(s-1)^2q_{\on{cert}}=O(1)$. This is the source of the $O(1)$ terms in \cref{fact:hmsy-model}(1).

Now, the event $\mathsf{Good}$ is defined in Step 1(d) and equations (82)-(84) in \cite{HMSY26}, which depends only on evaluation places and the fixed parameters. Moreover, $\mb{P}(\mathsf{Good}^c)\le p^{-1}\exp(n^{O(1)})$, which is exponentially small under our choice of $p$. We condition on $\mathsf{Good}$ for the rest of the proof.

In Step 5(c), \cite{HMSY26} reveals the $O(1)$ coordinates and calls $u$ the information encoding both the positions revealed and the values of a random codeword at these positions.
In Step 6, \cite{HMSY26} defines $\mathsf{E}_i^{\on{fin}}$ in equation (97), which is $V_i$ in \cref{fact:hmsy-model};  and also defines $A_i^\star: V_i\to [0, 1]$, which is $\widetilde{G}_i$ in \cref{fact:hmsy-model}.
Then, the Fourier coefficients $\mu_i$ in equation (98) has support exactly $S_i$ in our notation, so their relation space $V_0$ on $I$ from page 86 is precisely our dual auxiliary code $\widetilde{C}^\perp$. On page 86, \cite{HMSY26} checks exactly that it has dimension $O(1)$ and distance $k-O(1)$, proving \cref{fact:hmsy-model}(1).

Note that in Step 5(b), before conversion to the total variation bounds needed for LLR which sums over bit strings $y$,  \cite{HMSY26} derives on page 84 pointwise upper bounds on
\begin{equation}
\label{eq:hmsy-triangle}
|P_f(y)-Q_f(y)|\le |P_f(y)-\widetilde{P_f}(y)|+|Q_f(y)-\widetilde{Q_f}(y)|+|\widetilde{P_f}(y)-\widetilde{Q_f}(y)|
\end{equation}
and shows the first two terms are exponentially small. To deduce \cref{eq:hmsy-fourier-reduction}, we observe that the first two terms are chosen small enough to be absorbed by $\eps$ and that
$\E_{x\in C}\prod_{i=1}^nG_i(x_i)=P_f(\one)$ and $\prod_{i=1}^n\nu_i=Q_f(\one)$, so it suffices to bound the last term in the display above.

We turn to bounding the last term of \cref{eq:hmsy-triangle} for $y=\one$ and fixing any revealed coordinates $u$, which \cite{HMSY26} calls $P^{tor}_u(\one)-Q^{tor}_u(\one)$ on page 88 in Step 7. The revealed-coordinate factors are common and lie in $[0,1]$, so their average is bounded by the largest conditional discrepancy. The errors in (96) and (101) are absorbed into $\eps$ as well. \cite{HMSY26} bounds for $y=\one$ and $\gamma$ renamed to $y$ below (and in \cref{fact:hmsy-model}) that
\begin{equation}
\label{eq:tor-u}
\begin{aligned}
|P^{tor}_u(\one)-Q^{tor}_u(\one)| & \le (1+O(\rho_{s-1, p})) \sum_{0\ne y\in \prod_{i\in I}V_i^\star}\prod_{i\in I}|\widehat{A_i^\star}(y_i)|\cdot |M_y - N_y|
\\ & \le (1+O(\rho_{s-1, p})) \sum_{0\ne y\in \widetilde{C}^\perp}\prod_{i\in I}|\widehat{A_i^\star}(y_i)|+e^{-\Omega(n)},
\end{aligned}
\end{equation}
where $\rho_{s-1, p}$ is exponentially small by equations (100) and (102) of \cite{HMSY26}, and the second inequality follows the bounds in the main display on page 88 and equation (103) of \cite{HMSY26}. They also immediately sum over all bit strings, but we retain only $\one$.

At this point, upon identifying $A_i^\star = \widetilde{G}_i$, we have essentially proven \cref{fact:hmsy-model} except for $\mb{E}\widetilde{G}_i \le \nu_i$. To do so, we recall that $\nu_i:=\mb{E}_{z\in\F_p}G_i(z)$ and define $a_i:=\mb{E}_{v\in V_i}A_i^\star (v)$ and $m_i= \mb{E}_{z\in\F_p}A_i^\star (v_i(z))$ where polynomials $v_i(z)=\sum_{h=1}^{s-1}\beta_{i, h}z^h$ are defined in (97) and Step 6. 

Returning to Step 5(b), we note that $d_{\on{TV}}(Q_f, \widetilde{Q}_f)$ is exponentially small, which implies as we reveal the $O(1)$ coordinates in $u$ one by one, the mean changes by $e^{-\Omega(n)}$ upon taking marginals. The affine precomposition in Step 5(c) permutes $\F_p$, so it preserves the mean. 
Hence, $|\nu_i-m_i|\le e^{-\Omega(n)}$ including the approximation error in equation (96).

Moreover, averaging over uniform $v\in V_i$ versus $v_i(z)$ for uniform $z\in\F_p$ incurs exponentially small additive error at most $\rho_{s-1, p}\sum_{\xi\ne0}|\widehat{A_i^\star}(\xi)|$ from equations (100) and (102) of \cite{HMSY26}, using the support bound (98) and field-size estimate (104). Hence, $|\nu_i-a_i|\le \delta$ for exponentially small $\delta$. Then, we can rescale down $\widetilde{G}_i = c_iA_i^\star$ where $c_i=\min\{1, \nu_i/a_i\}$ to ensure $\mb{E}\widetilde{G}_i\le \nu_i$, which preserves the range $[0, 1]$ and changes the constant $1+O(\rho_{s-1, p})$ in \cref{eq:tor-u} by a factor of
\[ 
\prod_{i\in I}c_i^{-1}\le \left(1+\frac{\delta}{\nu}\right)^{|I|}\le e^{n\delta /\nu}=1+o(1),
\]
as $\delta$ is exponentially small and $\nu$ is constant. We relax the constant simply to $2$ in \cref{eq:hmsy-fourier-reduction}.
 \end{proof}

\end{document}